\documentclass[11pt]{article}

 \usepackage{geometry}
\usepackage{titling}
\usepackage{bm}
\usepackage{fullpage,amsthm,amssymb,amsmath,amsfonts,rotate}
\usepackage{graphicx,algorithmic,algorithm}
\usepackage[]{datetime}
\usepackage{aliascnt} 
\usepackage{hyperref,todonotes}
\usepackage{lmodern}
\usepackage{tikz}
\usetikzlibrary{calc,3d}
\usetikzlibrary{intersections,decorations.pathmorphing,shapes,decorations.pathreplacing,fit}
\usepackage{amsmath,ifthen,color,eurosym}
\usepackage{wrapfig,cite}
\usepackage{latexsym,amssymb,url}
\usepackage{subcaption}
\usepackage{cleveref}

\hypersetup{
 pdftitle = {Hitting Topological Minor Models in Planar Graphs is Fixed Parameter Tractable},
 colorlinks = true,
 urlcolor = Mygreen,
 linkcolor = blue!50!black,
 citecolor = Mygreen,
 menucolor = {red}
 filecolor = {cyan},
 anchorcolor = {green}
 urlcolor = {magenta},
 runcolor = {cyan},
 linkbordercolor = {blue}
}

\newcommand{\sshow}[2]{\ifthenelse{\equal{#1}{0}}{#2}{}}

\newcommand{\bd}{{\bf bor}}
\newcommand{\ann}{{\sf ann}}
\newcommand{\inter}{{\sf int}}
\newcommand{\remove}[1]{}

\newcommand{\bigmid}{\;\big|\;}
\newcommand{\w}{\operatorname{{\bf w}}}
\newcommand{\cupall}{\pmb{\pmb{\bigcup}}}
\newcommand{\seg}{partially ${\rm\Delta}$-embedded graph}

\newcounter{func}
\newcommand{\newfun}[1]{f_{\refstepcounter{func}\label{#1}\thefunc}}
\newcommand{\funref}[1]{\hyperref[#1]{f_{\ref*{#1}}}} 

\newcounter{con}
\newcommand{\newcon}[1]{c_{\refstepcounter{con}\label{#1}\thecon}}
\newcommand{\conref}[1]{\hyperref[#1]{c_{\ref*{#1}}}} 

\usepackage[spanish,polutonikogreek,english,]{babel}
\usepackage[utf8]{inputenc}

\usepackage[T1]{fontenc} 
\usepackage{alphabeta}

\newcommand{\bound}[1]{{\bf #1}}

\newtheorem{proposition}{Proposition}[section]

\newtheorem{lemma}{Lemma}[section]

\newtheorem{corollary}{Corollary}[section]

\newtheorem{theorem}{Theorem}[section]
\newtheorem{conjecture}{Conjecture}[section]

\newcommand{\tw}{{\mathbf{tw}}}
\renewcommand{\Bbb}[1]{{\mathbb{#1}}}

\newcommand{\hh}{\end{document}}

\usepackage{float}
\usepackage{tikz,pgf}
\usetikzlibrary{backgrounds}
\usetikzlibrary{calc,3d}
\usetikzlibrary{intersections,decorations.pathmorphing,shapes,decorations.pathreplacing,fit,shadows,fadings} 
\usepackage{color}

\definecolor{Black}{rgb}{0,0, 0}
\definecolor{Blue}{rgb}{0, 0 ,1}
\definecolor{Red}{rgb}{1, 0 ,0}
\definecolor{White}{rgb}{1, 1, 1}
\definecolor{Grey}{rgb}{.6, .6, .6}
\definecolor{Mygreen}{rgb}{.0, .7, .0}
\definecolor{Yellow}{rgb}{.55,.55,0}
\definecolor{mustard}{rgb}{1.0, 0.86, 0.35}
\definecolor{applegreen}{rgb}{0.55, 0.71, 0.0}
\definecolor{darkturquoise}{rgb}{0.0, 0.81, 0.82}
\definecolor{celestialblue}{rgb}{0.29, 0.59, 0.82}
\definecolor{green-yellow}{rgb}{0.68, 1.0, 0.18}
\definecolor{crimsonglory}{rgb}{0.75, 0.0, 0.2}
\definecolor{darkmagenta}{rgb}{0.55, 0.0, 0.55}

\tikzset{red node/.style={draw=red, circle, fill = red, minimum size = 4pt, inner sep = 0pt}}
\tikzset{yellow node/.style={draw=yellow, circle, fill = yellow, minimum size = 4pt, inner sep = 0pt}}
\tikzset{blue node/.style={draw=celestialblue, circle, fill =celestialblue, minimum size = 4pt, inner sep = 0pt}}
\tikzset{triangle/.style = { regular polygon, regular polygon sides=3, rotate=180}}
\tikzset{small red/.style={draw=red, triangle, fill = red, minimum size = 2pt, inner sep = 0pt}}

\tikzset{black node/.style={draw, circle, fill = black, minimum size = 4pt, inner sep = 0pt}}
\tikzset{small black node/.style={draw, circle, fill = black, minimum size = 3pt, inner sep = 0pt}}
\tikzset{model node/.style={draw=celestialblue, circle, fill = celestialblue, minimum size = 5pt, inner sep = 0pt}}
\tikzset{model node small/.style={draw=celestialblue, circle, fill = celestialblue, minimum size = 3pt, inner sep = 0pt}}
\tikzset{rep node/.style={draw=red, circle, fill = red, minimum size = 3pt, inner sep = 0pt}}
\tikzset{track node 1/.style={draw, circle, fill = black, minimum size = 2pt, inner sep = 0pt}}
\tikzset{track node 2/.style={draw=black!30!white, circle, fill = black!30!white, minimum size = 2pt, inner sep = 0pt}}
\tikzset{track node 3/.style={draw=black!10!white, circle, fill = black!10!white, minimum size = 2pt, inner sep = 0pt}}

\usepackage{latexsym,url}

\usepackage{enumerate}
\usepackage{cite}
\usepackage{ifthen}
\usepackage{enumitem}

\usepackage{titling}
\thanksmarkseries{arabic}
\newcommand*\samethanks[1][\value{footnote}]{\footnotemark[#1]}

\begin{document}

\title{\bf\Large  Hitting Topological Minor Models in Planar Graphs is Fixed Parameter Tractable\thanks{A preliminary version of the  results of this paper appeared in~\cite{GolovachST20hitt}.
}}

\author{Petr A. Golovach%
	\thanks{Department of Informatics, University of Bergen, Norway.}$\ ^,$\thanks{Supported by the Research Council of Norway via the projects ``CLASSIS'' and ``MULTIVAL''.}~$^{,}$\thanks{Supported by the Research Council of Norway and the French Ministry of Europe and Foreign Affairs, via the Franco-Norwegian project PHC AURORA 2019.}%
	\and 
	Giannos Stamoulis\thanks{LIRMM, Univ Montpellier, CNRS, Montpellier, France.}\hspace{1.55mm}$^,$\thanks{Supported by projects DEMOGRAPH (ANR-16-CE40-0028),  ESIGMA (ANR-17-CE23-0010), and the French-German Collaboration ANR/DFG Project UTMA (ANR-20-CE92-0027).}
	\and
	Dimitrios M. Thilikos\,{\samethanks[4]}~$^,${\samethanks[5]}\hspace{1.55mm}$^,{\samethanks[6]}$
}	
	
\date{}

\maketitle

\begin{abstract}
\noindent For a finite collection of graphs ${\cal F}$, the \textsc{${\cal F}$-TM-Deletion} problem has as input an $n$-vertex  graph $G$ and an integer $k$ and asks  whether there exists a set $S \subseteq V(G)$ with $|S| \leq k$ such that $G \setminus S$ does not contain any of the graphs in ${\cal F}$ as a topological minor. We prove that for every such ${\cal F}$, \textsc{${\cal F}$-TM-Deletion} is fixed parameter tractable on planar graphs. Our algorithm runs in a $2^{\mathcal{O}(k^2)}\cdot n^{2}$ time or, alternatively in $2^{\mathcal{O}(k)}\cdot n^{4}$ time. Our techniques can easily  be extended to graphs that are  embeddable on any fixed surface.
\end{abstract}
\medskip



%


\section{Introduction}
\subsection{The \texorpdfstring{${\cal P}$}{P}-deletion problem and its variants}

In general, a ${\cal P}$-{\sc deletion} problem is determined by some graph class ${\cal P}$ and asks, given an $n$-vertex  graph $G$ and an integer $k$, whether $G$ can be 
transformed to a graph in ${\cal P}$ after the deletion of $k$ vertices.  
In other words, the class ${\cal P}$ represents some desired property 
that we want to impose on the input graph after deleting $k$ vertices.
This is a general 
graph modification problem with great expressive power
as it encompasses many problems, depending on the choice of the property ${\cal P}$.
Unfortunately for most instances of ${\cal P}$, this problem is not expected 
to admit a polynomial time algorithm.
Lewis and Yannakakis showed in~\cite{LewisY80then} that for any non-trivial and hereditary graph class ${\cal P}$, the ${\cal P}$-vertex deletion problem is ${\sf NP}$-complete.  Given this hardness result, an attractive alternative 
is to consider the standard parameterized version of the problem, called {\sl p}\,-{\sc ${\cal P}$-deletion} where the parameter  is the number $k$ of vertex deletions. In this case the  challenge is to investigate for which instances of ${\cal P}$, {\sl p}\,-{\sc ${\cal P}$-deletion} is fixed parameter tractable  (or, in short, is {\sf FPT}), i.e., 
it can be solved by an $f(k)\cdot n^{\mathcal{O}(1)}$-time algorithm (also called an {\sf FPT}-algorithm), for some function  $f:\Bbb{N}\to\Bbb{N}$. There is a long line of research on this general question. In many cases, this concerns 
particular properties and possible optimizations of the {\em parametric dependence} $f(k)$ (see e.g.~\cite{BodlaenderHL14}). However, it is interesting to notice that  {\sf FPT}-algorithms 
exist for  general families of properties.
In this direction the more general (and compact) results 
concern properties ${\cal P}$ that can be characterized by the exclusion of 
some finite set ${\cal F}$ of graphs of at most $h$ vertices or edges with respect to some partial ordering relation $\leq $. We define
$${\cal P}_{{\cal F},\leq }=\{G\mid \forall H\in {\cal F} :  H\not\leq  G\}$$
and ask whether {\sl p}\,-{\sc ${\cal P}_{{\cal F},\leq }$-deletion} is {\sf FPT}.
Let us now consider the general status of this problem for the main known instances of the partial ordering relation $\leq $.
\medskip

\noindent(1) $\leq $ is the {\sl contraction}\footnote{A graph $G$ is a {\em contraction} of a graph $G'$ if $G$ can be obtained from $G$ by applying edge contractions.} relation: then there are graphs $H$ such that  {\sc ${\cal P}_{\{H\},\leq }$-deletion} is {\sf NP}-complete even for the case where $k=0$. For instance one may take $H$ to be
the path on 4 vertices, as indicated in~\cite{BrouwerV87}. Using the terminology of fixed parameter complexity, this implies that there are choices of ${\cal F}$ such that  {\sl p}\,-{\sc ${\cal P}_{{\cal F},\leq }$-deletion} is {\sf para-NP}-complete.
\medskip

\noindent(2) $\leq $ is the {\sl induced minor}\footnote{A graph $G$ is an {\em induced minor} of a graph $G'$ if $G$ can be obtained from some contraction of $G'$ after removing vertices.} relation: as in the previous case, there are choices of ${\cal F}$ such that  {\sl p}\,-{\sc ${\cal P}_{{\cal F},\leq }$-deletion} is {\sf para-NP}-complete. For instance, one may consider ${\cal F}$ to contain the graph in~\cite[Theorem 4.3]{FellowsKMP95}.\medskip

\noindent(3) $\leq $ is the {\sl subgraph} or the {\sl induced subgraph} 
relation: because of the result of Cai in~\cite{Cai96fixe}, {\sl p}\,-{\sc ${\cal P}_{{\cal F},\leq }$-deletion} is {\sf FPT}, for every ${\cal F}$.  In particular, the result in~\cite{Cai96fixe}
implies an $\mathcal{O}(h^{k}n^{h+1})$-time algorithm for both these problems. However, if instead we  parameterize {\sc ${\cal P}_{{\cal F},\leq }$-deletion} by $h$, then there are instances of ${\cal F}$ for which the problem is ${\sf W[1]}$-hard even for $k=0$: just take ${\cal F}=\{K_{h}\}$ in order  to generate the {\sl p}-{\sc Clique} problem.\medskip

\noindent(4) $\leq $ is the  {\sl minor}\footnote{A graph $G$ is an  {\em minor} of a graph $G'$ is $G$ is the contraction of some subgraph of $G'$.} relation:  again  {\sl p}\,-{\sc ${\cal P}_{{\cal F},\leq  }$-deletion} is {\sf FPT}, for every ${\cal F}$.
To see this, observe
 that, for every $k$, the 
set of  {\sf yes}-instances of this problem is closed under taking of minors.  On the other 
hand, Robertson and Seymour~\cite{RobertsonS04GMXX} proved that graphs are well-quasi-ordered with respect to the minor relation. These two facts together 
imply that  there is a finite set ${\cal B}_{k}$ (whose size depends on $k$ and $h$) such that $(G,k)$ is a {\sf yes}-instance  if and only if $G$ contains no graph in ${\cal B}_{k}$ as a minor. As minor-checking for a graph on $c$ vertices can be done in $\mathcal{O}_{c}(n^3)$ time~\cite{RobertsonRXIII}, we derive  the (non-constructive) existence of an $\mathcal{O}_{k,h}(n^3)$-time algorithm (see \autoref{sadfadsfsdffgdsgdfgfg} for the definition of the ${\cal O}_{k,h}(\cdot )$ notation).
This result was made constructive in \cite{AdlerGK08comp}. Recently,  a $2^{k^{\mathcal{O}_{h}(1)}}\cdot n^2$ time algorithm  for {\sl p}\,-{\sc ${\cal P}_{{\cal F},\leq  }$-deletion} was designed in~\cite{SauST20anfp}.

\subsection{Our contribution.}
A graph $H$ is a {\em topological minor} of a graph $G$ if $G$ contains as a subgraph 
some subdivision\footnote{A graph $G$ is a {\em subdivision} of a graph $G'$
if $G$ can be obtained from $G'$ if we replace its edges by paths with the same endpoints.} of $H$ and we denote this by $H\preceq G$. We consider the  problem  {\sl p}\,-{\sc ${\cal P}_{{\cal F},\preceq }$-deletion} that in the rest of this paper we  call  \textsc{${\cal F}$-TM-Deletion}. Notice that this problem is more general than its 
 counterpart for the minor relation (case (4) above) as, for every 
 graph $H$, there exists a finite set of graphs ${\cal H}$ such that 
 a graph  $G$ contains $H$ as a minor if and only if $G$ contains some graph in  ${\cal H}$ as a topological minor. However as graphs are not well-quasi-ordered with respect to the topological minor relation, the  parameterized  complexity of \textsc{${\cal F}$-TM-Deletion} remained open for a while.\medskip
  
In this paper we prove that \textsc{${\cal F}$-TM-Deletion} is {\sf FPT}
for inputs  restricted to planar graphs. Moreover, we develop  results and techniques that may serve as the base for further {\sf FPT}-algorithms for  {\sl p}\,-{\sc ${\cal P}_{{\cal F},\leq  }$-deletion} on planar graphs, when $\leq$ is the induced minor or the contraction relation (see \autoref{asdgasdfgdsfgdfg} for a discussion). 
We  stress that, until very recently,  
the parameterized complexity of this problem  was unknown. For an update 
on the current status of the general problem see \autoref{asdfsfgddfhdgshfghg}.\medskip

Let  ${\cal F}$ be a finite set of graphs.  We use $h({\cal F})$ for the maximum number of vertices or edges  of a graph in ${\cal F}$, i.e., {$h({\cal F})=\max\{|V(H)|, |E(H)|\mid H \in {\cal F}\}$}.
We also write ${\cal F}\npreceq G$ to denote the fact that none of the graphs in ${\cal F}$ is a topological minor of $G$. 
We define the  parameter ${\bf tm}_{\cal F}$ so that, for every graph $G$,
$${\bf tm}_{\cal F}(G)=\min\{k\mid \exists S\subseteq V(G): |S|\leq k \wedge {\cal F}\npreceq G\setminus S\}.$$ The main result of this paper is the following:

\begin{theorem} \label{dfassdfasfdfdsad}
There exists an algorithm that given a finite set of graphs ${\cal F}$, a $k\in\Bbb{N}$, and an $n$-vertex planar graph $G$,
outputs whether ${\bf tm}_{\cal F}(G)\leq k$ in  $2^{\mathcal{O}_h (k^2)}\cdot n^2$ time, or, alternatively, $\mathcal{O}(k\cdot n^4)+\mathcal{O}_{h}(n^4)+2^{\mathcal{O}_{h}(k)}\cdot n^2$ time, where $h=h({\cal F})$. 
\end{theorem}

We stress that the algorithm of \autoref{dfassdfasfdfdsad} can be straightforwardly 
modified so as to output a set $S$ of size at most $k$ that intersects all models of the graphs in ${\cal F}$.
{A version of \autoref{dfassdfasfdfdsad} without the explicit parametric dependences on the running times appeared in~\cite{GolovachST20hitt}.}

\subsection{High level description of our algorithm} 
Our main approach towards proving \autoref{dfassdfasfdfdsad} is the application of the so-called {\em irrelevant vertex technique}.  
This technique was introduced for the first time  by Roberston and Seymour in~\cite{RobertsonRXIII}
for the design of an {\sf FPT}-algorithm for  the {\sc Disjoint Paths} problem, parameterized by the number of terminals. Subsequently,
it was applied, in diverse ways, for the design of {\sf FPT}-algorithms 
for several graph-theoretical problems and is now considered
as a powerful technique of parameterized algorithm
design~\cite{AdlerKKLST17irre,CyganMPP13thep,GolovachKMT17thep,GolovachHP13,GroheKMW11find,JansenLS14anea,KawarabayashiK12,KawarabayashiM08,KawarabayashiMR08,KawarabayashiR10oddc,Marx10chor,MarxS12obta}.
We also refer to \cite[Chapter~7]{cygan2015parameterized} for a high-level
overview of the irrelevant vertex technique.
The general algorithmic paradigm of the irrelevant vertex technique
takes advantage of some structural characteristic of the input graph
in order to detect, in {\sf FPT}-time, 
some vertex, called {\em irrelevant}, whose 
removal from $G$ generates an equivalent instance of the problem.
By recursing on the produced equivalent instance we end up with  
a graph where the structural parameter is bounded (by some function of $k$).
This fact permits the resolution of the problem with other techniques -- typically by dynamic programming.
Most of the times, this structural parameter 
is  treewidth (see \autoref{sadfadsfsdffgdsgdfgfg}
for the formal definition) and this is the one that we use in this paper.
Towards proving \autoref{dfassdfasfdfdsad}, the application of the irrelevant vertex technique is based on \autoref{dfasdfdsad} that we present below.

Let $G$ be a graph, $R$ be a subset of $V(G)$, and $k$ be a non-negative integer.
We say that $(G,R,k)$ is a {\em ${\bf tm}_{\cal F}$-triple} if there exists an $S\subseteq R$
such that $|S|\leq k$ and ${\cal F}\npreceq G\setminus S$.
Intuitively, the set $R$ can be seen as the set of vertices that are
possible candidates for a solution $S$. Aiming to remove (irrelevant) vertices from the given graph,
we also make progress by reducing $R$. This is formulated in the next result.

\begin{theorem}\label{dfasdfdsad}
There exists a function $\newfun{nntenethjehe}:\Bbb{N}\to\Bbb{N}$, and an algorithm with the following specifications:

\smallskip\noindent {\bf Find\_Irrelevant\_Vertex}$(k,h,G,R)$\\
\noindent{\sl Input:}  $k,h\in \Bbb{N}$, an $n$-vertex planar graph $G$,
and a set $R\subseteq V(G)$.	

\noindent{\sl Output:}
\begin{enumerate}
\item an (irrelevant) vertex $v\in V(G)$ and a set $R'\subseteq R$ such that,
for every graph class ${\cal F}$ where $h({\cal F})\leq h$,
$(G,R,k)$ is a  ${\bf tm}_{\cal F}$-triple if and only if
$(G\setminus v, R', k)$  is a ${\bf tm}_{\cal F}$-triple or 
\item a tree decomposition of $G$ of width
at most $\funref{nntenethjehe}(h)\cdot k$.
\end{enumerate}

Moreover, this algorithm runs in  $2^{\mathcal{O}_h (k^2)}\cdot n$ time, or, alternatively, $\mathcal{O}(k\cdot n^3)+\mathcal{O}_{h}(n^3)+2^{\mathcal{O}_{h}(k)}\cdot n$ time.
\end{theorem}
After applying the algorithm of \autoref{dfasdfdsad} at most $n$ times,  the problem is reduced  to instances of bounded treewidth.  As topological minor containment can be defined by a formula in Monadic Second Order Logic, i.e., an MSOL formula,~\cite[Appendix D]{KimLPRRSS16line} 
and vertex deletion to some MSOL definable property is also MSOL definable, it follows from Courcelle's Theorem~\cite{CourcelleM93} (see also~\cite{BoriePT92,ArnborgLS91easy,Seese96line}) that the problem for reduced instances  can be solved in $\mathcal{O}_{k,h}(n)$ time. 
To solve the version of the problem where a certificate of the solution is asked for, one can use the version of Courcelle's Theorem \cite{CourcelleM93}
that returns such a certificate, if it exists.
However, to achieve the parametric dependencies in the running times of \autoref{dfassdfasfdfdsad}, we have to avoid the use of Courcelle's Theorem when solving the problem on instances of bounded treewidth.
We devote \autoref{dsfnskdalgnaklds} to describe how to develop a dynamic programming algorithm (\autoref{sssswfgfegwergewgwergwegr}) that can solve the problem on instances of treewidth at most $w$ in $2^{\mathcal{O}_{h}(w\log w)}\cdot n$ time, or, alternatively, in $\mathcal{O}(n^3)+2^{\mathcal{O}_{h}(w)}\cdot n$ time. \autoref{dfassdfasfdfdsad} follows. We stress that each one of these running times has some advantage against the other. In the first case, we have a linear, in $n$, algorithm whose parametric dependence on $k$
is super-exponential. In the second, we drop the parametric dependency to a single exponential one to the cost 
of a worst polynomial dependency on $n$.
\medskip

In the rest of this section we give an outline on how \autoref{dfasdfdsad} is proved. All combinatorial concepts used in this description are presented in an intuitive way;   formal definitions can be found in  \autoref{sadfadsfsdffgdsgdfgfg}. 
Given a tuple of variables ${\bf x}=(x_{1},\dots, x_{q})$ by the term {\em ${\bf x}$-big/small} we refer to a quantity that is lower/upper bounded by some  (unbounded) function of ${\bf x}$.
Alternatively, we use the term  {\em {\bf x}-many/few} that is defined analogously.
We work on some embedding of $G$ in the plane.

\paragraph{Walls and annuli.}   An important combinatorial object is  the one of a {\sl $r$-wall }, as the one  in \autoref{asfdsfdsfasdfsdfdsf}, that can be seen as the union of $r$ horizontal paths intersected by $r$ vertical paths. The layers of a wall $W$ are defined as indicated  in \autoref{asfdsfdsfasdfsdfdsf}.

\begin{figure}[ht]\centering
\sshow{0}{\includegraphics[width=10cm]{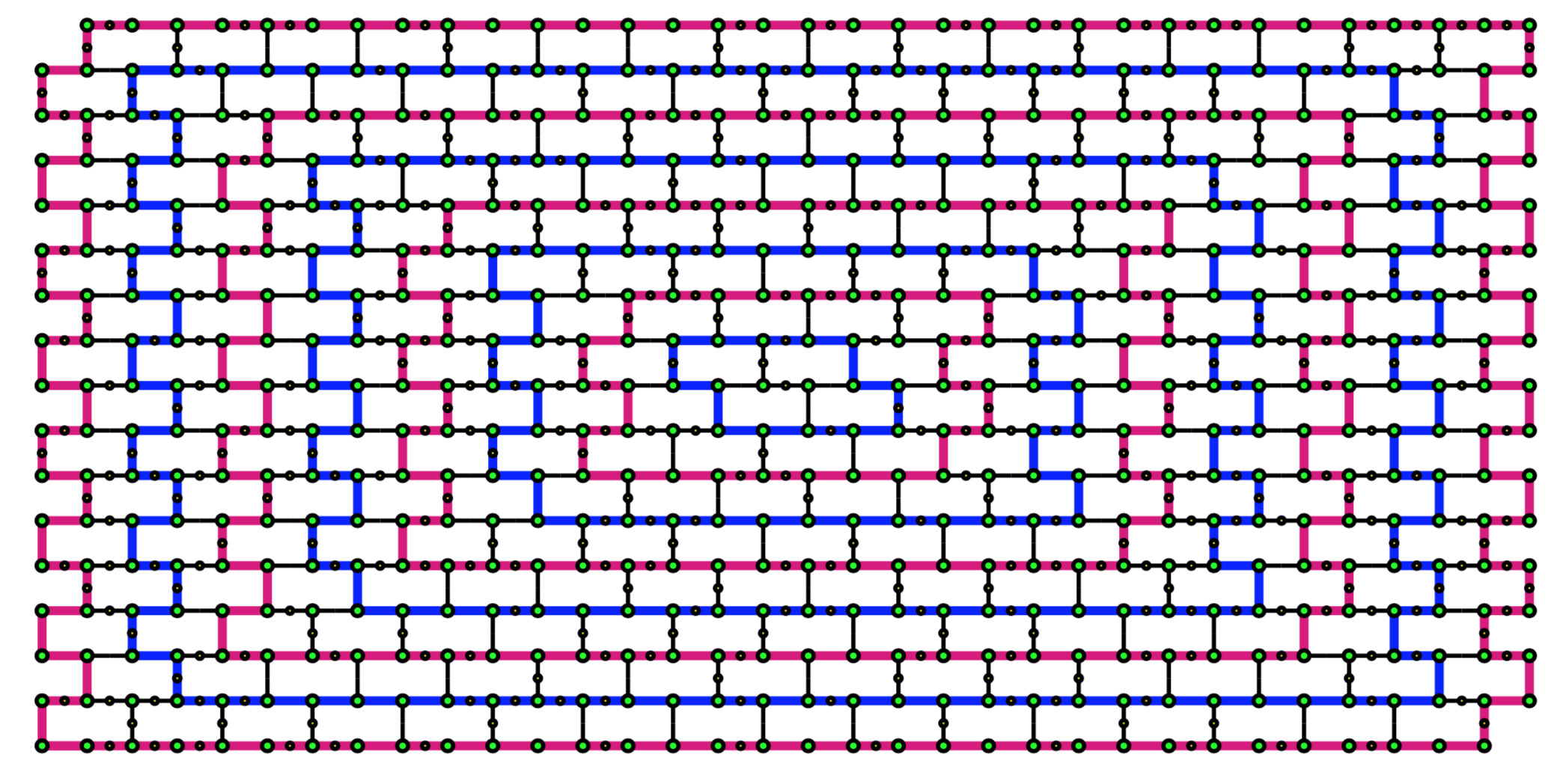}}
  \caption{\small A $17$-wall and its 8 layers.}
\label{asfdsfdsfasdfsdfdsf}
\end{figure}
We call the outermost layer {\em perimeter} of the wall $W$.
Combining the results of \cite{BodlaenderDDFLP16,GolovachKMT17thep,GuT12,AdlerDFST11} we know that if the treewidth of a planar graph 
is $(k,h)$-big, then we can find a $(k,h)$-big  wall in $G$ such that the subgraph of $G$, called {\em the compass} of $W$, inside the closed disk ``cropped'' by  the perimeter of $W$ has $(k,h)$-small  treewidth (see \autoref{something_good}). This additional property will permit us to compute (possible) partial solutions on subgraphs of the compass of $W$.

The next step is to  detect some more structure in the wall $W$ that is intuitively depicted in the left side of \autoref{asdfdsfdsdsfdsfsgfdgsdhf}.  We first distinguish the collection ${\cal C}$  of the $(k,h)$-many  outermost layers, drawn in yellow, and then we consider in the rest of $W$  a packing of $(k,h)$-many  $(h)$-big walls, drawn in green. This is done in \autoref{gdgdasgfhjkuyi}.

We now work on the ``annulus'' of the $(k,h)$-many outer layers of $W$. For this, it is convenient to see those cycles as ``crossed'' by a collection ${\cal P}$ of disjoint paths (that are monotone subpaths of the horizontal/vertical paths of $W$)  
called {\em rails}.
We call this system of cycles and rails {\em railed annulus}, denoted by ${\cal A}=({\cal C},{\cal P})$. See the right side of \autoref{asdfdsfdsdsfdsfsgfdgsdhf} for an example of a railed annulus with 5 cycles and 8 rails.

\paragraph{Combing topological minor models.}
Notice that if $H$ is a topological minor of a graph $G$, then 
this is materialized by a pair $(M,T)$ where $M$ is a subgraph of $G$ 
and $T$ is a set of vertices of $M$, called {\em branches},
such that all vertices of $V(M)\setminus T$ have degree 2.
We say that  $(M,T)$ is a {\em topological minor model} of $H$ in $G$ if 
a graph isomorphic to $H$ is created after dissolving in $M$ all vertices in $V(M)\setminus T$
(which means deleting every such vertex and making its two neighbors adjacent).
For simplicity, assume that ${\cal F}=\{H\}$ and recall that ${\bf tm}_{\cal F}(G)\leq k$
if there is a set $S\subseteq V(G), |S|\leq k$, called from now on {\em solution set},
that intersects all  topological minor models of $H$ in $G$.

 Our next aim is to  analyze how topological minor models of $H$ 
may cross the cycles and the rails of a railed annulus ${\cal A}=({\cal C},{\cal P})$.
{For this reason, using \cite[Corollary 1]{GolovachST22comb} (see also~\cite{GolovachST20hitt} for a conference version), we prove that if the branches of $(M,T)$ are situated {\sl outside} the annulus and the
annulus is $(h)$-big then it is possible to find an alternative ``rail-combed'' 
model $(M',T')$ of $G$, whose intersection with the ``middle cycle'' of ${\cal A}$ 
consists only of  $(h)$-few  rail vertices. We refer to this theorem as the {\em ``model combing theorem''} (\autoref{jklnlnlk}).}

\paragraph{Representations of topological minor models.}
Using the model combing theorem, we can pick an $(h)$-small collection ${\cal P}'$
of the rails of ${\cal A}$ for which the following holds: for every topological minor model $(M,T)$ of $H$ that crosses ${\cal A}$, there is a disk $\Delta$ bounded by some cycle $C$ of ${\cal A}$ 
and a ``combed'' (through ${\cal P}'$) version  $(M',T')$ of $(M,T)$ that {\sl represents} $(M,T)$ in the sense that a set of vertices that are ``not so close'' to $C$, intersects $M\cap \Delta$ if and only if the same set 
intersects $M'\cap \Delta$. From now on we refer to the instances of $M'\cap \Delta$ as the {\em inner 
combed models} of ${\cal A}$ and we can see them as models representing the ``inner part'' of all annulus-crossing models.

\begin{figure}[ht]
	\centering\scalebox{1.05}{
	\sshow{0}{\begin{tikzpicture}[scale=0.28]
	\filldraw[mustard] (0,0) rectangle (22,16);
	\fill[white] (2,2) rectangle (20,14);
	\foreach \i in {0,...,4}{
		\draw (\i/2,\i/2) rectangle (22-\i/2, 16-\i/2);
	}
	\foreach \a in {0,6}{
		\foreach \b in {0,6,12}
		{
			\begin{scope}[xshift=\b cm, yshift = \a cm]
			\foreach \i in {5,...,10}{
				\fill[applegreen] (2.5,2.5) rectangle (7.5,7.5);
				\fill[white] (4.5,4.5) rectangle (5.5,5.5);
				\draw[-] (\i/2,\i/2) rectangle (10-\i/2, 10-\i/2);
				
			}
			\end{scope}
		}
	}\foreach \a in {0,6}{
		\foreach \b in {0,6,12}
		{
			\begin{scope}[xshift=\b cm, yshift = \a cm]
			\foreach \i in {5,...,9}{
				\draw[-] (\i/2,\i/2) rectangle (10-\i/2, 10-\i/2);
				
			}
			\end{scope}
		}
	}
	\end{tikzpicture}}~~~~
	\sshow{0}{\begin{tikzpicture}[scale=.31]
	\foreach \x in {2,...,6}{
		\draw[line width =0.6pt] (0,0) circle (\x cm);
	}
	\begin{scope}[on background layer]
	\fill[mustard] (0,0) circle (6 cm);
	\fill[applegreen!98!white] (0,0) circle (2 cm);
	\end{scope}
	
	\node (P3) at (45:7) {$P_{3}$};
\node[small black node] (P11) at (45:6) {};
\node[small black node] (P21a) at (30:5) {};
\node[small black node] (P21b) at (40:5) {};
\node[small black node] (P31a) at (35:4) {};
\node[small black node] (P31b) at (50:4) {};
\node[small black node] (P41a) at (45:3) {};
\node[small black node] (P41b) at (25:3) {};
\node[small black node] (P51) at (40:2) {};
\draw[line width=1pt] (P11) -- (P21a) -- (P21b) -- (P31a)  (P31b) -- (P41a) -- (P41b) -- (P51);

\node (P4) at (70:7) {$P_{4}$};
\node[small black node] (P12) at (70:6) {};
\node[small black node] (P22a) at (80:5) {};
\node[small black node] (P22b) at (75:5) {};
\node[small black node] (P32a) at (90:4) {};
\node[small black node] (P32b) at (75:4) {};
\node[small black node] (P42a) at (85:3) {};
\node[small black node] (P42b) at (70:3) {};
\node[small black node] (P52) at (75:2) {};
\draw[line width=1pt] (P12) -- (P22a) -- (P22b) -- (P32a) (P32b) -- (P42a) -- (P42b) -- (P52);

\node (P5) at (115:7) {$P_{5}$};
\node[small black node] (P13a) at (120:6) {};
\node[small black node] (P13b) at (110:6) {};
\node[small black node] (P23a) at (110:5) {};
\node[small black node] (P23b) at (115:5) {};
\node[small black node] (P33a) at (120:4) {};
\node[small black node] (P33b) at (130:4) {};
\node[small black node] (P43a) at (135:3) {};
\node[small black node] (P43b) at (120:3) {};
\node[small black node] (P53) at (125:2) {};
\draw[line width=1pt] (P13a) -- (P13b) -- (P23a) -- (P23b) -- (P33a) -- (P33b) -- (P43a) -- (P43b) -- (P53);

\node (P6) at (165:7) {$P_{6}$};
\node[small black node] (P14a) at (170:6) {};
\node[small black node] (P14b) at (160:6) {};
\node[small black node] (P24) at (155:5) {};
\node[small black node] (P34) at (160:4) {};
\node[small black node] (P44a) at (165:3) {};
\node[small black node] (P44b) at (180:3) {};
\node[small black node] (P54) at (170:2) {};
\draw[line width=1pt] (P14a) -- (P14b) -- (P24) -- (P34) -- (P44a) -- (P44b) -- (P54);

\node (P7) at (190:7) {$P_{7}$};
\node[small black node] (P18a) at (190:6) {};
\node[small black node] (P28a) at (185:5) {};
\node[small black node] (P28b) at (220:5) {};
\node[small black node] (P38a) at (210:4) {};
\node[small black node] (P38b) at (225:4) {};
\node[small black node] (P48a) at (205:3) {};
\node[small black node] (P48b) at (220:3) {};
\node[small black node] (P58) at (200:2) {};
\draw[line width=1pt] (P18a) -- (P28a) to [bend right=15]  (P28b);
\draw[line width=1pt] (P28b) --  (P38a) -- (P38b) -- (P48a) -- (P48b) -- (P58);

\node (P8) at (235:7) {$P_{8}$};
\node[small black node] (P15) at (235:6) {};
\node[small black node] (P25a) at (240:5) {};
\node[small black node] (P25b) at (250:5) {};
\node[small black node] (P35b) at (245:4) {};
\node[small black node] (P45a) at (250:3) {};
\node[small black node] (P45b) at (240:3) {};
\node[small black node] (P55) at (235:2) {};
\draw[line width=1pt] (P15) -- (P25a)  -- (P25b) --  (P35b) -- (P45a) -- (P45b) -- (P55);

\node (P1) at (295:7) {$P_{1}$};
\node[small black node] (P16a) at (290:6) {};
\node[small black node] (P16b) at (300:6) {};
\node[small black node] (P26a) at (295:5) {};
\node[small black node] (P26b) at (310:5) {};
\node[small black node] (P36a) at (300:4) {};
\node[small black node] (P36b) at (290:4) {};
\node[small black node] (P46a) at (310:3) {};
\node[small black node] (P46b) at (325:3) {};
\node[small black node] (P56) at (320:2) {};
\draw[line width=1pt] (P16a) -- (P16b) -- (P26a) -- (P26b) -- (P36a) -- (P36b) -- (P46a) -- (P46b) -- (P56);

\node (P2) at (0:7) {$P_{2}$};
\node[small black node] (P17a) at (5:6) {};
\node[small black node] (P17b) at (-5:6) {};
\node[small black node] (P27a) at (0:5) {};
\node[small black node] (P27b) at (-10:5) {};
\node[small black node] (P37a) at (-15:4) {};
\node[small black node] (P37b) at (0:4) {};
\node[small black node] (P47a) at (10:3) {};
\node[small black node] (P47b) at (-5:3) {};
\node[small black node] (P57) at (-5:2) {};
\draw[line width=1pt] (P17a) -- (P17b) -- (P27b)  -- (P27a) -- (P37a) -- (P37b) -- (P47a) -- (P47b) -- (P57);

\node (C5) at (90:1.2) {$C_{5}$};

\node (D5) at (-70:0.7) {$\Delta_{\rm in}$};
	
	\end{tikzpicture}}}\medskip\medskip
	\caption{{\sf Left}: The partition of a wall into a yellow annulus and several green subwalls. {\sf Right}: An example of a $(5,8)$-railed annulus depicted in yellow; its innermost cycle is $C_5$ and the disk $\Delta_{\rm in}$ bounded by $C_5$ is depicted in green.}
	\label{asdfdsfdsdsfdsfsgfdgsdhf}
\end{figure}
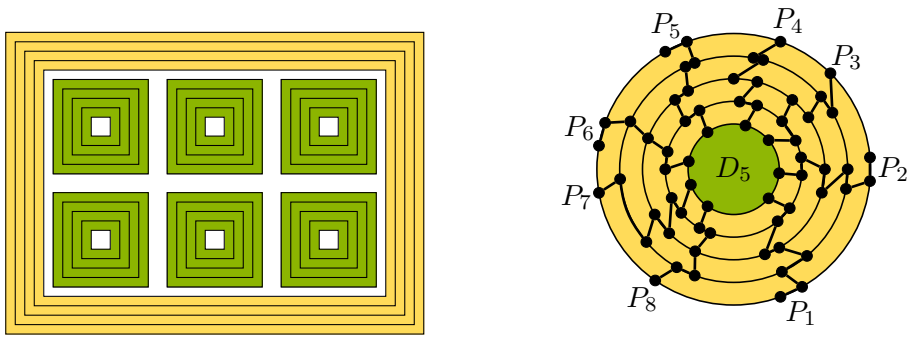

\paragraph{Reducing the solution space.}
The next step is  to compute, for every cycle $C$ of ${\cal A}$, a set $S_{C}$ of at most $(k,h)$-many vertices intersecting each possible inner combed model of ${\cal A}$ (it is possible that $S_{C}$ is an empty set).
This computation can be done by the dynamic programming algorithm of \autoref{sssswfgfegwergewgwergwegr} that can find (partial) solutions of the problem on subgraphs of the compass of $W$ that have $(k,h)$-small treewidth.
Let  $\Delta_{\rm in}$
be the disk bounded by the innermost cycle of ${\cal C}$ (cycle $C_{5}$ in \autoref{asdfdsfdsdsfdsfsgfdgsdhf}). 
We then compute $S_{\rm in}=\Delta_{\rm in}\cap (\bigcup_{C\in {\cal C}}S_{C})$ and 
observe that $S_{\rm in}$ has $(k,h)$-small size. Based on the fact that the inner 
combed models represent the inner part of all models crossing ${\cal A}$
and the fact that all these models are intersected by subsets of at most $k$
vertices  whose restriction in $\Delta_{\rm in}$ is in $S_{\rm in}$, we 
prove that if  $G\setminus S$ does not contain any topological minor model of $H$, then 
we can replace $S\cap \Delta_{\rm in}$ by vertices of $S_{\rm in}$ to obtain a new solution that is not larger than $S$ (\autoref{fsfsdfdsdsafsasdsdadffasdasfd}).
This is an important restriction of the solution space of the problem 
in what concerns its intersection with  $\Delta_{\rm in}$.
As there are $(k,h)$-many
$(h)$-big subwalls packed inside  $\Delta_{\rm in}$, there is a subwall whose compass can be avoided by all possible 
solution sets.
{The compass of such a wall is called {\em solution-free}.}
In the above, $H$ might be any graph on $h$ vertices, however its is more convenient to think about some specific (planar) graph $H$ in ${\cal F}$.

\paragraph{Finding an irrelevant vertex.}  We  now fix our attention to the solution-free  compass of 
some $(h)$-big subwall of $W$. Once again, we see this wall as a railed annulus ${\cal A}'$ and 
use the model combing theorem
in order to represent all ways topological minor models of $H$ can ``invade'' 
the compass of $W$ by combed topological models going through the rails of ${\cal A}'$.
This, in turn, permits us to detect a vertex $v$ of the solution-free compass of $W$ such that if a solution set 
$S$ intersects a topological 
minor model that contains $v$, then it should also intersect some representation of it that 
avoids $v$, therefore $v$ is irrelevant (\autoref{fsfsdfdsdsafsadffasdasfd1}).

\subsection{Organization of the paper}
In \autoref{sadfadsfsdffgdsgdfgfg}, we give some definitions and preliminaries. In \autoref{dsfnskdalgnaklds} we present a way to design a dynamic programming algorithm that solves the problem in bounded treewidth graphs. In \autoref{asfdsasdfdsf} we present the two main subroutines of the algorithm of \autoref{dfasdfdsad} and in \autoref{sasdfsdfdsfsdfsdgffdsgdfgsdfgd} we prove  \autoref{dfasdfdsad}. We conclude in \autoref{dnfjklsdgnklsdnbm} by discussing the running time dependancy on $h$ of our algorithm, the extension of our results to bounded Euler genus graphs, some recent advances on the study of the problem on general graphs, and some open problems.

\section{Definitions and preliminaries}
\label{sadfadsfsdffgdsgdfgfg}

We denote by $\Bbb{N}$ the set of all non-negative  integers. Given an $n\in\Bbb{N}$, we denote by $\Bbb{N}_{\geq n}$ the set containing 
all integers equal or greater than $n$.
Given  two integers $x$ and $y$ we define by $[x,y]=\{x,x+1,\ldots,y-1,y\}$. Given an $n\in\Bbb{N}_{\geq 1}$, we also define $[n]=\{1,\ldots,n\}$.
Let $U$ be a set, $r\in \Bbb{N}_{\geq 1}$, and ${\cal A}=[A_{1},\ldots,A_{r}]\subseteq (2^{U})^r$, ${\cal B}=[B_{1},\ldots,B_{r}]\subseteq (2^{U})^r$.
We say that  ${\cal A}\subseteq {\cal B}$ if for all $i\in[r]$, $A_{i}\subseteq B_{i}$.
Also, if $S\subseteq U$ we denote ${\cal A}\cap S=[A_{1}\cap S,\ldots,A_{r}\cap S]$. 
\smallskip

Let $(x_{1},\ldots,x_{l})\in \Bbb{N}^{l}$ and $\chi,\psi: \Bbb{N}
\rightarrow \Bbb{N}$.
We use the notation $\chi(n)=\mathcal{O}_{x_{1},\ldots,x_{l}}(\psi(n))$ to denote that there exists a computable
function $f:\Bbb{N}^{l} \rightarrow \Bbb{N}$
such that  $\chi(n)=\mathcal{O}( f(x_{1},\ldots,x_{l})\cdot \psi(n))$.

\subsection{Basic concepts on graphs}
All graphs in this paper are undirected,  finite, and they do not  have loops or multiple edges.
Unless stated otherwise, we denote by $n$ the number of vertices of the graph under consideration.
If $G_{1}=(V_{1},E_{1})$ and $G_{2}=(V_{2},E_{2})$ are graphs, then we denote
$G_{1}\cap G_{2}=(V_{1}\cap V_{2},E_{1}\cap E_{2})$ and $G_{1}\cup G_{2}=(V_{1}\cup V_{2},E_{1}\cup E_{2})$.
Also, given a graph $G$ and a set $S\subseteq V(G)$, we denote by $G\setminus S$ the graph obtained if we remove from $G$ the vertices in $S$, along with their incident edges.
Given a vertex $v\in V(G),$ we denote by $N_{G}(v)$ the set of vertices of $G$ that are adjacent to $v$ in $G.$
Also, given a set $S\subseteq V(G)$, we set $N_G(S) = \bigcup_{v \in S}N_G(v)$ and $N_G [S]=N_G (S)\cup S$.
We denote by $\partial(S)$ the set of vertices in $S$ that have a neighbor in $V(G)\setminus S$.
Given a graph $G$, we say that the pair $(A,B)$ is a {\em separation} of $G$ 
if $A\cup B=V(G)$ and there is no edge in $G$ with one endpoint in $A\setminus B$ 
and the other in $B\setminus A$.
A {\em path} ({\em cycle}) in a graph $G$ is a connected subgraph with all vertices of degree 
at most (exactly) 2. A path is {\em trivial} if it has only one vertex and is empty if it is the empty graph (i.e., the graph with empty vertex set).
%

\paragraph{Partially disk-embedded graphs.}
A {\em closed disk} (resp. {\em open disk}) $\Delta$
is a set homeomorphic to the set $\{(x,y)\in\Bbb{R}^2\mid x^2+y^2\leq 1\}$
(resp. $\{(x,y)\in \Bbb{R}^2\mid x^2+y^2< 1\}$).
We use $\bd(\Delta)$ to denote the boundary of $\Delta$ and $\inter(\Delta)$ to denote the open disk $\Delta\setminus \bd(\Delta)$.
When we embed a graph $G$ in the
plane or in a disk, we treat $G$ as a set of points. This permits us to make
set operations operations between graphs and sets of points.
We say that a graph $G$ is {\em partially disk-embedded in some closed disk $\Delta$}, 
if there is some subgraph $K$ of $G$ that is embedded in $\Delta$
such that $\bd(\Delta)$ is a cycle of $K$  and $(V(G)\cap \Delta,V(G)\setminus\inter(\Delta))$
is a separation of $G$. From now on, we use the term {\em partially $\Delta$-embedded graph $G$}
to denote that a graph $G$ is  partially disk-embedded in some closed disk $\Delta$.
{We also call the graph $K$
{\em the compass}
of the partially $\Delta$-embedded graph $G$ and we always assume that we accompany
a partially $\Delta$-embedded graph $G$ with the embedding of its compass in $\Delta$, that is the set $G\cap \Delta$.}

\paragraph{Grids and walls.}
Let  $k,r\in\Bbb{N}.$ The
{\em $(k\times r)$-grid} is the Cartesian product of two paths on $k$ and $r$ vertices, respectively.
An  {\em elementary $r$-wall}, for some odd $r\geq 3,$ is the graph obtained from a
$(2 r\times r)$-grid with vertices $(x,y),$
$x\in[2r]\times[r],$ after the removal of the
``vertical'' edges $\{(x,y),(x,y+1)\}$ for odd $x+y,$ and then the removal of
all vertices of degree one. 
Notice that, as $r\geq 3$,  an elementary $r$-wall is a planar graph
that has a unique (up to topological isomorphism) embedding in the plane
such that all its finite faces are incident to exactly six edges.  
The {\em perimeter} of an elementary $r$-wall is the cycle bounding its infinite face.
Given an elementary wall $\overline{W},$ a {\em vertical path} of $\overline{W}$ is  one whose 
vertices, in order of appearance, are 
$(i,1),(i,2),(i+1,2),(i+1,3),
(i,3),(i,4),(i+1,4),(i+1,5),
(i,5),\ldots,(i,r-2),(i,r-1),(i+1,r-1),(i+1,r)$, for some $i\in \{1,3,\ldots,2r-1\}$.
 Also, a {\em horizontal path} of $\overline{W}$
 is one whose 
vertices, in order of appearance, are $(1,j),(2,j),\ldots,(2r,j)$, for some  $j\in[2,r-1]$, 
or  $(1,1),(2,1),\ldots,(2r-1,1)$ or  $(2,r),(3,r),\ldots,(2r,r)$.

An {\em $r$-wall} is any graph $W$ obtained from an elementary $r$-wall $\overline{W}$
after subdividing edges (see \autoref{asfdsfdsfasdfsdfdsf}).  We call the vertices that were added after the subdivision operations {\em subdivision vertices}.
The {\em perimeter} of $W$, denoted by ${\sf perim}(W)$, is the cycle of $W$ whose non-subdivision vertices are the vertices of the perimeter of $\overline{W}$.
Also, a vertical (resp. horizontal) path of $W$ is a subdivided vertical (resp. horizontal) path of $\overline{W}$.
An {\em $r'$-subwall} $W'$ of a wall $W$ is any $r'$-wall that is a subgraph of $W$ 
and whose horizontal/vertical paths are subpaths of the horizontal/vertical paths of $W$.

A subgraph $W$ of a graph $G$ is called a {\em wall} of 
$G$ if $W$ is an $r$-wall for some odd $r\geq  3$ and we refer to $r$ as the {\em height} of the wall $W$.
Let $W$ be a wall of a graph $G$ and $K'$ be the connected component of $G \setminus V({\sf perim}(W))$ that contains $W \setminus V({\sf perim}(W))$.
The compass of $W$, denoted by ${\sf compass}(W)$, is the graph 
$G[V(K')\cup V({\sf perim}(W))]$.
Observe that $W$ is a subgraph of ${\sf compass}(W)$ and ${\sf compass}(W)$ is connected.

The {\em layers} of an $r$-wall $W$  are recursively defined as follows.
The first layer of $W$ is its perimeter.
For $i=2,\ldots,(r-1)/2,$ the $i$-th layer of $W$ is the $(i-1)$-th layer of the subwall $W'$ obtained from $W$ after removing from $W$ its perimeter and all occurring vertices of degree one.
Notice that each $(2r+1)$-wall has $r$ layers  (see \autoref{asfdsfdsfasdfsdfdsf}).

\paragraph{Treewidth.} 
A \emph{tree decomposition} of a graph $G$
is a pair $(T,\chi)$ where $T$ is a tree and $\chi: V(T)\to 2^{V(G)}$
such that
\begin{enumerate}
	\item $\bigcup_{t \in V(T)} \chi(t) = V(G)$;
	\item for every edge $e$ of $G$ there is a $t\in V(T)$ such that 
	$\chi(t)$
	contains both endpoints of $e$ and
	\item for every $v \in V(G)$, the subgraph of ${T}$
	induced by $\{t \in V(T)\mid {v \in \chi(t)}\}$ is connected.
\end{enumerate}
The \emph{width} of $(T,\chi)$ is defined as
$\w(T,\chi):=\max\big\{\left|\chi(t)\right|-1 \mid t\in V(T)\big\}.$
The \emph{treewidth of $G$} is defined as $\tw(G):= 
\min\left\{\w(T,\chi) \bigmid (T,\chi) \text{ is a tree decomposition of }G\right\}.$

The following result follows combining the results of \cite{BodlaenderDDFLP16,GolovachKMT17thep,GuT12,AdlerDFST11}.
It intuitively states that given a $q\in \Bbb{N}$ and a planar graph $G$  with ``big'' enough treewidth, we can find a $q$-wall of $G$ whose compass has ``small'' enough treewidth.
\begin{proposition}\label{something_good}
There exists a constant $\newcon{sdafsdfsd}$  and an algorithm with the following specifications:
	
\smallskip\noindent {\bf Find\_Wall}$(G,q)$\\
\noindent{\sl Input:} a planar graph $G$ and a $q\in \Bbb{N}_{\geq 3}$.

\noindent{\sl Output:}	
\begin{enumerate}
	\item Either a $q$-wall $W$ of $G$ whose  compass  has treewidth at most $\conref{sdafsdfsd}\cdot q$ or  
	\item a tree decomposition of $G$ of width at most $\conref{sdafsdfsd}\cdot q$.
\end{enumerate}
Moreover, this algorithm runs in $2^{\mathcal{O}(q^2)}\cdot n$ time,
or, alternatively, in $2^{\mathcal{O}(q)}\cdot n^2$ time.	
\end{proposition}

The above algorithm uses first the
single exponential {\sf FPT}-approximation of treewidth
by \cite{BodlaenderDDFLP16} and as long as the treewidth is not small 
enough then it finds a $q$-wall $W$ by either using
the algorithm of \cite{AdlerDFST11}, that runs in $2^{\mathcal{O}(q^2)}\cdot n$ time, or 
the algorithm of \cite{GuT12} that runs in $\mathcal{O}(n^2)$ time.
The treewidth of the compass of $W$ is bounded by applying
the main idea of \cite[Lemma 4.2]{GolovachKMT17thep}.
{
We present a proof for completeness.

\begin{proof}[Proof of~\autoref{something_good}]
The following algorithm is a slight modification of the algorithm
Compass in~\cite[Subsection 4.2]{GolovachKMT17thep}.
The version presented here uses \cite[Theorem VI]{BodlaenderDDFLP16} and
the algorithms of \cite{AdlerDFST11,GuT12} to obtain the claimed running times.

We set $\conref{sdafsdfsd}:= 94$.
We start by applying the single-exponential $5$-approximation algorithm
of Bodlaender et al. for treewidth \cite[Theorem VI]{BodlaenderDDFLP16},
which outputs either a report that the treewidth of $G$ is larger than
$18q+1$ or a tree decomposition of $G$ of width at most $5\cdot (18q+1)+4$.
Observe that $5\cdot (18q+1)+4\leq 94 q,$ for $q\geq 3$.
In the latter case, we return the obtained tree decomposition of $G$.
In the former case,
i.e., where the treewidth of $G$ is larger than
$18q+1$, we know~\cite[Lemma 2.1]{GolovachKMT17thep} that
$G$ contains a $2q$-wall as a minor.
Such a $q$-wall $W$ can be found using either the minor-checking algorithm of \cite{AdlerDFST11}
that runs in $2^{\mathcal{O}(q^2)}\cdot n$ time, or 
the algorithm of \cite{GuT12} that runs in $\mathcal{O}(n^2)$ time.
Next, among the four vertex-disjoint $q$-subwalls of $W$,
we obtain the one, say $W'$,
whose compass has the minimum number of vertices.
We recursively apply the algorithm of~\cite[Theorem VI]{BodlaenderDDFLP16}
with input the compass of $W'$ and the integer $18q+1$.
For more details, we refer the reader to~\cite[Subsection 4.2]{GolovachKMT17thep}.
\end{proof}
}

\subsection{Railed annuli}\label{klagakdfhglka}
In this subsection we present the notion of {\it railed annulus}, introduced in \cite{KaminskiT12}, a ``wall-like'' graph  as in the right side of \autoref{asdfdsfdsdsfdsfsgfdgsdhf}, that is the union of a collection of cycles and a collection of paths ``crossing'' these cycles.
In order to define railed annuli, we first give the definitions of {\it nested sequences of cycles} and {\it annuli}.

\paragraph{Nested cycles and annuli.}
Let $G$ be a  partially $\Delta$-embedded graph and let ${\cal C}=[C_{1},\ldots,C_{r}]$, $r\geq 2$, 
be a collection of  vertex-disjoint cycles of the compass of $G$. We say that the sequence ${\cal C}$ is a {\em $\Delta$-nested sequence of cycles} of $G$
if every $C_{i}$ is the boundary of an open
disk $D_{i}$  
such that  $\Delta\supseteq D_{1}\supseteq\cdots \supseteq D_{r}$. From now on,
each $\Delta$-nested sequence ${\cal C}$ will be accompanied 
with the sequence $[D_{1},\ldots,D_{r}]$ of the corresponding open disks 
as well as 
the sequence $[\overline{D}_{1},\ldots,\overline{D}_{r}]$ of their closures. 
Given $x,y\in[r]$ where $x\leq y$, 
we call the set  $\overline{D}_{x}\setminus D_{y}$ the {\em $(x,y)$-annulus} of ${\cal C}$
and we denote it by $\ann({\cal C},x,y)$. Finally, we say that $\ann({\cal C},1,r)$
is the {\em annulus} of ${\cal C}$ and we denote it by $\ann({\cal C})$.  
 
 \paragraph{Railed annuli.}
Let $r\in\Bbb{N}_{\geq 3}$ and $q\in \Bbb{N}_{\geq 3}$. Assume also that $r$ is an odd number.
An {\em $(r,q)$-railed annulus} of a partially $\Delta$-embedded graph $G$ is a pair ${\cal A}=({\cal C},{\cal P})$ where 
${\cal C}=[C_{1},\ldots,C_{r}]$  is a $\Delta$-nested collection of cycles of $G$ and ${\cal P}=[P_{1},\ldots,P_{q}]$ is a  collection of pairwise vertex-disjoint 
 paths in $G$ such that 
\begin{itemize}
	\item For every $j\in[q],$\ $P_{j}\subseteq \ann({\cal C})$.
	\item For every $(i,j)\in[r]\times[q],$   $C_{i}\cap P_{j}$ is  a non-empty path, that we denote $P_{i,j}$.
\end{itemize}
We refer to the paths of ${\cal P}$ as the {\em rails} of ${\cal A}$ and to the cycles of ${\cal C}$ as the {\em cycles} of ${\cal A}$.

Let ${\cal A}=({\cal C},{\cal P})$ be an $(r,q)$-railed annulus of a partially $\Delta$-embedded graph $G$. We call $\overline{D}_{r}$ (resp. $\overline{D}_{1}$) the {\em inner (resp. outer) disk} of ${\cal A}$. We also extend the notion of an annulus and we say that the {\em annulus} of ${\cal A}=({\cal C}, {\cal P})$ is the annulus of ${\cal C}$.\medskip

We now prove the following lemma which intuitively states that there is an algorithm that given a ``big enough'' wall, outputs a collection of  railed annuli whose number and size will be useful in the proof of \autoref{dfasdfdsad}.

%
%
%
%
%

\begin{lemma}\label{gdgdasgfhjkuyi}
	There exists a function $\newfun{ddansjbndaj}:\Bbb{N}^3\to\Bbb{N}$ and an algorithm with the following specifications:

	\smallskip\noindent {\bf Find\_Collection\_of\_Annuli}$(x,y,z,\Delta,G,W)$\\
	\noindent{\sl Input:} two odd integers $x,y\in \Bbb{N}_{\geq 3}$, an integer $z\in \Bbb{N}$,  a partially $\Delta$-embedded graph $G$ and a $q$-wall $W$ of the compass of $G$ whose perimeter is the boundary of $\Delta$ and such that $q\geq \funref{ddansjbndaj}(x,y,z)$.
	
	\noindent{\sl Output:}  a closed disk $\Delta'\subseteq \Delta$ and a collection $\mathfrak{A}=\{{\cal A}_{0},{\cal A}_{1},\ldots,{\cal A}_{z}\}$ of railed annuli of the compass of $G$ such that 
	\begin{itemize}
		\item ${\cal A}_{0}$ is an $(x,x)$-railed annulus  whose outer disk is  $\Delta$ and whose inner disk is $\Delta'$,   
		\item for $i\in[z]$, ${\cal A}_{i}$ is a $(y,y)$-railed annulus of $G\cap \inter(\Delta')$, and  
		\item for every $i,j\in [z]$ where $i\neq j$, the outer disk of ${\cal A}_{i}$ and the outer disk of ${\cal A}_{j}$ are disjoint.
	\end{itemize}
	Moreover, this algorithm runs in $\mathcal{O}(n)$ time and  $\funref{ddansjbndaj}(x,y,z)=\mathcal{O}(x+y\sqrt{z})$.
	
\end{lemma}

\begin{proof}
Let $y':= 2y+ \lceil y/4 \rceil$ and assume that $y'$ is an odd integer (otherwise, make it odd by adding 1) and let $\funref{ddansjbndaj}(x,y,z)=2x + \max\{\lceil x/4 \rceil,\lceil \sqrt{z}/2 \rceil \cdot y' \} +1$.
We argue that the following holds:\smallskip

\noindent{\em Claim:} Let $p\in \Bbb{Z}_{\geq 3}$ be an odd integer.
If $H$ is an $h$-wall of $G$, where $h$ is an odd integer such that $h\geq 2p + \lceil p/4\rceil$,
then $H$ contains a $(p,p)$-railed annulus ${\cal A}=({\cal C}, {\cal P})$,
where ${\cal C}=[C_{1}, \ldots, C_{p}]$
and
for every $i\in [p]$, $C_{i}$ is the $i$-th layer of $H$.
	
\smallskip
\noindent{\em Proof of Claim:}
Let $H$ be an $h$-wall of $G$, where $h\geq 2p + \lceil p/4\rceil$.
We define the $\Delta$-nested collection ${\cal C} = [C_{1},\ldots, C_{p}]$ of cycles of $G$, where, for every $i\in[p]$,  $C_{i}$ is the $i$-th layer of $H$.
Let $\hat{\cal P}$ be the collection of the vertical and horizontal paths of $H$ {that contain branch vertices of $W$ that are not in $\bigcup_{i\in [p]} V(C_i)$.}
Observe that, for every $ i\in [p]$, every path in $\hat{\cal P}$ also intersects $C_{i}$
and that $\hat{\cal P}\cap \ann{({\cal C})}$ is a collection of pairwise-vertex disjoint paths of $G$.
Also, notice that since $h-2p\geq \lceil p/4\rceil$, $\hat{\cal P}\cap \ann{({\cal C})}$ contains at least $p $ paths.
Let ${\cal P} := [P_{1}, \ldots, P_{p}]$ be a subset of $\hat{\cal P}\cap \ann{({\cal C})}$.
Then, ${\cal P}$ is a collection of pairwise vertex-disjoint paths of $G$ and it holds that for every $j\in[p]$,
$P_{j}\subseteq \ann({\cal C})$ and for every $(i,j)\in[p]\times[p],$   $C_{i}\cap P_{j}$ is  a non-empty path.
Therefore, $H$ contains a $(p,p)$-railed annulus ${\cal A} = ({\cal C}, {\cal P})$ of $G$ and the claim follows.
\smallskip

	Following the claim above, for $H:=W$, $h:=q$, and $p:=x$, since $q\geq 2x + \lceil x/4\rceil$, we deduce the existence of an $(x,x)$-railed annulus $A_{0}$ whose  inner disk is $\overline{D}_{x}$ and whose outer disk is $\overline{D}_{1}$ - that is $\Delta$.
	Observe that since $q-2x\geq \lceil\sqrt{z}/2 \rceil \cdot y' +1$, there exists an $r$-wall $\hat{W}$ of $G$ for some odd $r\in \Bbb{Z}_{\geq 3}$ such that $r\geq \lceil\sqrt{z}/2 \rceil \cdot y'$ and $\hat{W}\subseteq G\cap{D}_{x}$.
	
Now, notice that $\hat{W}$ contains a collection ${\cal W} = \{W_{1}',\ldots, W_{z}'\}$ of $z$ $y'$-subwalls of $W$ such that, for every $i,j\in[z], i\neq j$, ${\sf compass}(W_{i}')\cap {\sf compass}(W_{j}')=\emptyset$.
Therefore, for every $i\in[z]$, applying again the claim above for $H:=W_{i}'$, $h:=y'$ and $p:=y$, we deduce the existence of a $(y,y)$-railed annulus ${\cal A}_{i}$ of $W_{i}'$.
Furthermore, for every $i,j\in[z], i\neq j$,
the fact that ${\sf compass}(W_{i}')\cap {\sf compass}(W_{j}')=\emptyset$
implies that the outer disk of ${\cal A}_{i}$ and the outer disk of ${\cal A}_{j}$ are disjoint.
The proof concludes by setting
$\Delta'=\overline{D}_x$ and $\mathfrak{A} = \{{\cal A}_{0}, {\cal A}_{1}, \ldots, {\cal A}_{z}\}$.
\end{proof}
%
%

\subsection{Rerouting linkages inside railed annuli}\label{fjnsdlkgnls}

In the rest of this section we show how to reroute topological minor models inside railed annuli.
For this reason, in \autoref{fjnsdlkgnls}, we define the notion of a linkage, which we study as a subgraph of a partially disk-embedded graph.
It has been proved \cite[Corollary 1]{GolovachST22comb} that if a linkage $L$ of a partially disk-embedded graph invades a sufficiently large railed annulus inside the disk, then there is an equivalent linkage that is ``combed'' through the rails of the annulus.
In \autoref{dkgldjgklfdhgklsfh}, we extend this result (\autoref{fsfsdfdsdsafsadffasdasfd2}) to topological minor models, by treating the paths of the model as paths of the linkage, and we conclude with the ``model combing theorem'' (\autoref{jklnlnlk}) that allows us to reroute topological minor models in order to ``comb'' them through the rails of a sufficiently large railed annulus.

Before stating \autoref{fsfsdfdsdsafsadffasdasfd2} we need some definitions.

\paragraph{Linkages.}
A \emph{linkage} in a graph $G$ is a subgraph $L$ of $G$  whose connected components are all non-trivial paths. The {\em paths} of a linkage are its connected components and we denote them by ${\cal P}(L).$
The {\em size} of $L$ is the number of paths and is denoted by $|L|$.
The \emph{terminals} of a linkage $L$, denoted by $T(L)$, are the endpoints of the paths in ${\cal P}(L)$, and
the \emph{pattern} of $L$ is the set $\big\{\{s,t\}\mid {\cal P}(L)\mbox{ contains some $(s,t)$-path}\big\}.$ Two linkages $L_{1},L_{2}$ of $G$  are {\em equivalent} if they have the same pattern and we denote this fact by $L_{1}\equiv L_{2}$.

\paragraph{Linkages in railed annuli.}
Let $G$ be a \seg, let ${\cal A}=({\cal C},{\cal P})$ be a $(r,q)$-railed annulus of $G$ and $L$ be a linkage of $G$.
Given a set $D\subseteq \Delta$, then we say that $L$ is {\em  $D$-avoiding }
if $T(L)\cap D =\emptyset$.
We also say that $L$ is {\em ${\cal A}$-avoiding} if it is $\ann({\cal C})$-avoiding (see Figure~\ref{asdfdsghdfhgdfdfsvdfgdfdsafdsf}). 

\begin{figure}[ht]
	\centering
	\scalebox{0.8}{
	\sshow{0}{\begin{tikzpicture}[scale=.53]

	\foreach \x in {2,...,6}{
		\draw[line width =0.6pt] (0,0) circle (\x cm);
	}
	\begin{scope}[on background layer]
	\fill[celestialblue!10!white] (0,0) circle (6 cm);
	\fill[white] (0,0) circle (2 cm);
	\end{scope}
	
	\node[black node] (P11) at (45:6) {};
	\node[black node] (P21a) at (30:5) {};
	\node[black node] (P21b) at (40:5) {};
	\node[black node] (P31a) at (35:4) {};
	\node[black node] (P31b) at (50:4) {};
	\node[black node] (P41a) at (45:3) {};
	\node[black node] (P41b) at (25:3) {};
	\node[black node] (P51) at (40:2) {};
	
	\draw[line width=1pt] (P11) -- (P21a) -- (P21b) -- (P31a) -- (P31b) -- (P41a) -- (P41b) -- (P51);
	
	\node[black node] (P12) at (70:6) {};
	\node[black node] (P22a) at (80:5) {};
	\node[black node] (P22b) at (75:5) {};
	\node[black node] (P32a) at (90:4) {};
	\node[black node] (P32b) at (75:4) {};
	\node[black node] (P42a) at (85:3) {};
	\node[black node] (P42b) at (70:3) {};
	\node[black node] (P52) at (75:2) {};
	
	\draw[line width=1pt] (P12) -- (P22a) -- (P22b) -- (P32a) -- (P32b) -- (P42a) -- (P42b) -- (P52);
	
	\node[black node] (P13a) at (120:6) {};
	\node[black node] (P13b) at (110:6) {};
	\node[black node] (P23a) at (110:5) {};
	\node[black node] (P23b) at (115:5) {};
	\node[black node] (P33a) at (120:4) {};
	\node[black node] (P33b) at (130:4) {};
	\node[black node] (P43a) at (135:3) {};
	\node[black node] (P43b) at (120:3) {};
	\node[black node] (P53) at (125:2) {};
	
	\draw[line width=1pt] (P13a) -- (P13b) -- (P23a) -- (P23b) -- (P33a) -- (P33b) -- (P43a) -- (P43b) -- (P53);
	
	\node[black node] (P14a) at (170:6) {};
	\node[black node] (P14b) at (160:6) {};
	\node[black node] (P24) at (155:5) {};
	\node[black node] (P34) at (160:4) {};
	\node[black node] (P44a) at (165:3) {};
	\node[black node] (P44b) at (180:3) {};
	\node[black node] (P54) at (170:2) {};
	
	\draw[line width=1pt] (P14a) -- (P14b) -- (P24) -- (P34) -- (P44a) -- (P44b) -- (P54);
	
	\node[black node] (P18a) at (190:6) {};
	\node[black node] (P28a) at (185:5) {};
	\node[black node] (P28b) at (220:5) {};
	\node[black node] (P38a) at (210:4) {};
	\node[black node] (P38b) at (225:4) {};
	\node[black node] (P48a) at (205:3) {};
	\node[black node] (P48b) at (220:3) {};
	\node[black node] (P58) at (200:2) {};
	
	\draw[line width=1pt] (P18a) -- (P28a) to [bend right=15]  (P28b);
	\draw[line width=1pt] (P28b) --  (P38a) -- (P38b) -- (P48a) -- (P48b) -- (P58);

	\node[black node] (P15) at (235:6) {};
	\node[black node] (P25a) at (240:5) {};
	\node[black node] (P25b) at (250:5) {};
	
	\node[black node] (P35b) at (245:4) {};
	\node[black node] (P45a) at (250:3) {};
	\node[black node] (P45b) at (240:3) {};
	\node[black node] (P55) at (235:2) {};
	
	\draw[line width=1pt] (P15) -- (P25a)  -- (P25b) --  (P35b) -- (P45a) -- (P45b) -- (P55);
	
	\node[black node] (P16a) at (290:6) {};
	\node[black node] (P16b) at (300:6) {};
	\node[black node] (P26a) at (295:5) {};
	\node[black node] (P26b) at (310:5) {};
	\node[black node] (P36a) at (300:4) {};
	\node[black node] (P36b) at (290:4) {};
	\node[black node] (P46a) at (310:3) {};
	\node[black node] (P46b) at (325:3) {};
	\node[black node] (P56) at (320:2) {};
	
	\draw[line width=1pt] (P16a) -- (P16b) -- (P26a) -- (P26b) -- (P36a) -- (P36b) -- (P46a) -- (P46b) -- (P56);
	
	\node[black node] (P17a) at (10:6) {};
	\node[black node] (P17b) at (-5:6) {};
	\node[black node] (P27a) at (0:5) {};
	\node[black node] (P27b) at (-10:5) {};
	\node[black node] (P37a) at (-15:4) {};
	\node[black node] (P37b) at (0:4) {};
	\node[black node] (P47a) at (10:3) {};
	\node[black node] (P47b) at (-5:3) {};
	\node[black node] (P57) at (-5:2) {};
	
	\draw[line width=1pt] (P17a) -- (P17b) -- (P27b)  -- (P27a) -- (P37a) -- (P37b) -- (P47a) -- (P47b) -- (P57);
	
	\draw[crimsonglory, line width=1.5pt] plot [smooth, tension=1.2] coordinates { (120:1) (140:2) (120:3.5) (100:4.5) (145:5) (150:6.5)};
	\draw[crimsonglory, line width=1.5pt] plot [smooth, tension=1] coordinates { (245:6.5) (230:5) (230:3) (240:1.5) (270:2.5) (285:4) (280:6.5)};
	\draw[crimsonglory, line width=1.5pt] plot [smooth, tension=1.2] coordinates { (40:6.5) (20:5) (-10:2.5) (-40:6.5)};
	\node[black node]  () at (120:1) {};
	\node[black node]  () at (150:6.5) {};
	\node[black node]  () at (245:6.5) {};
	\node[black node]  () at (280:6.5) {};
	\node[black node]  () at (40:6.5) {};
	\node[black node]  () at (-40:6.5) {};

	\end{tikzpicture}}}
	\caption{An example of a railed annulus ${\cal A}$ and a linkage $L$ (depicted in red) that is ${\cal A}$-avoiding.}
	\label{asdfdsghdfhgdfdfsvdfgdfdsafdsf}
\end{figure}

Let  $r=2t+1$. Let also $s\in[r]$ where $s=2t'+1$.
Given some $I\subseteq [q]$, we say that a linkage  $L$ is {\em $(s,I)$-confined in ${\cal A}$} if 
$$L\cap \ann({\cal C},t+1-t',t+1+t')\subseteq \bigcup_{i\in I}P_{i}.$$

We are now ready to state the following result from \cite{GolovachST20hitt}, 
whose  proof can be found in \cite{GolovachST22comb}.

\begin{proposition}[\!\!\cite{GolovachST20hitt,GolovachST22comb}]
	\label{fsfsdfdsdsafsadffasdasfd2}
	There exist two functions  $\newfun{axfsfsd}, \newfun{axfsfadsrdsfsd}:\Bbb{N}\to\Bbb{N}$, where the images of $\funref{axfsfadsrdsfsd}$ are even, such that 
	for every odd $s\in \Bbb{N}_{\geq  1}$ and every $\ell\in\Bbb{N}$,
	if $G$ is a \seg,  ${\cal A}=({\cal C},{\cal P})$ is a  $(r,q)$-railed 
	annulus of $G$, where  $r=\funref{axfsfadsrdsfsd}(\ell)+s$ and $q\geq  5/2\cdot\funref{axfsfsd}(\ell)$, 
	$L$ is an ${\cal A}$-avoiding  linkage of size at most $\ell$, and $I\subseteq [q]$, where $|I|> \funref{axfsfsd}(\ell)$, then
	$G$ contains a linkage
	$\tilde{L}$ where $\tilde{L}\equiv L$, $\tilde{L}$ is ${\cal A}$-avoiding, $\tilde{L}\setminus \ann({\cal C}) \subseteq L \setminus \ann({\cal C})$, and  $\tilde{L}$ is $(s,I)$-confined in ${\cal A}$.
Moreover, $\funref{axfsfadsrdsfsd}(\ell)=\mathcal{O}((\funref{axfsfsd}(\ell))^{2})$.
\end{proposition}

It follows from the result in~\cite{AdlerKKLST17irre}, that $\funref{axfsfsd}(\ell)=2^{\mathcal{O}(\ell)}$, when $G$ is a planar graph. Furthermore, if $G$ is a graph of 
Euler genus at most  $\gamma$, then  $\funref{axfsfsd}(\ell)=2^{\mathcal{O}_{\gamma}(\ell)}$, because of the result of Mazoit in~\cite{Mazoit13asin}.  We stress that \cite{GolovachST22comb} contains the proof of a  more general version of \autoref{fsfsdfdsdsafsadffasdasfd2} where linkages are {\sl $t$-scattered}, i.e., their paths are within distance at least $t$.

\subsection{Rerouting topological minors}\label{dkgldjgklfdhgklsfh}
We say that $(M,T)$   is a {\em {\sf tm}-pair} if $M$ is  a graph, $T\subseteq V(M)$, and  all vertices in 
$V(M)\setminus T$ have degree two. We denote by ${\sf diss}(M,T)$ the graph obtained 
from  $M$ by dissolving all vertices  in $V(M)\setminus T$.
A {\em {\sf tm}-pair} of a graph $G$  is a  {\sf tm}-pair $(M,T)$ where 
$M$ is a subgraph of $G$. 
Given two graphs $H$ and $G,$ we say that a {\sf tm}-pair $(M,T)$ of $G$,  is a {\em topological minor model of $H$ in $G$} if $H$ is isomorphic to ${\sf diss}(M,T)$. 
We call the vertices in $T$ {\em branch} vertices of $(M,T)$.

\paragraph{Topological minor models in railed annuli.}
Let $G$ be a \seg, let $H$ be a graph, ${\cal A}=({\cal C},{\cal P})$ be a $(r,q)$-railed annulus of $G$.
Let  $r=2t+1$. Let also $s\in[r]$ where $s=2t'+1$.
Given some $I\subseteq [q]$, we say that a topological minor model $(M,T)$ of $H$ in $G$ is {\em $(s,I)$-confined in ${\cal A}$} if 
$$M\cap \ann({\cal C},t+1 -t',t+1+t')\subseteq \bigcup_{i\in I}P_{i}.$$

Intuitively, the above definition demands that $M$ traverses the ``middle'' $(s,q)$-annulus 
by intersecting it only at the rails of ${\cal A}$.

Our algorithms are strongly based on the following combinatorial result.

\begin{theorem}[Model Combing]
\label{jklnlnlk}
There exist two functions $\funref{axfsfsd}, \funref{axfsfadsrdsfsd}:\Bbb{N}\to\Bbb{N}$, {where the images of $\funref{axfsfadsrdsfsd}$ are even},  such that
if
\begin{itemize}
	\item $s$ is a positive odd integer,
	\item $H$ is a graph on at most $g$ edges,
	\item $G$ is a \seg,
	\item  ${\cal A}=({\cal C},{\cal P})$ is an  $(r,q)$-railed
	annulus of $G$, where  $r = \funref{axfsfadsrdsfsd}(g)+2+s$ and $q\geq  5/2 \cdot \funref{axfsfsd}(g)$,
	\item $(M,T)$ is a  topological minor model of $H$ in $G$ such that $T\cap {\sf ann}({\cal A})=\emptyset$, and
	\item $I\subseteq [q]$ where $|I|> \funref{axfsfsd}(g)$,
\end{itemize} then
$G$ contains a topological minor model $(\tilde{M},\tilde{T})$ of $H$ in $G$ such that
\begin{enumerate}
	\item 	 $\tilde{T}=T$,
	\item   $\tilde{M}$ is $(s,I)$-confined in ${\cal A}$ and
	\item $\tilde{M}\setminus \ann({\cal A})\subseteq {M}\setminus \ann({\cal A})$.
\end{enumerate}
Moreover, $\funref{axfsfadsrdsfsd}(g)=\mathcal{O}((\funref{axfsfsd}(g))^{2})$.
\end{theorem}

\begin{proof}
	Let $s$ be a positive odd integer, $H$ be a graph on $g$ edges, $G$ be a \seg, ${\cal A}=({\cal C},{\cal P})$ be a  $(r,q)$-railed 
	annulus of $G$, where  $r= \funref{axfsfadsrdsfsd}(g) +2+s$ and $q\geq  5/2\cdot \funref{axfsfsd}(g)$, $(M,T)$ be a topological minor model of $H$ in $G$  such that $T\cap \ann({\cal A})=\emptyset$.
	
	Notice that all the connected components of $M\setminus T$ are paths of $G$. Let $L$ be the linkage of $G\setminus T$ created by taking the union of all non-trivial connected components of $M\setminus T$ (see \autoref{sdfadfsdfgdgfdgfdg}). Observe that
	${\cal P}(L)$ is the set of all paths of $G$ connecting neighbors of branch vertices of $M$ and consisting only of subdividing vertices of $M$ and that there is an one-to-one correspondence of ${\cal P}(L)$ with $E(H)$. Thus $|L|\leq g$.

	Let ${\cal A}'=([C_{2}, \ldots, C_{r-1}], {\cal P}\cap \ann({\cal C},2,r-1))$ and keep in mind that ${\cal A}'$  is a  $(r',q)$-railed 
	annulus of $G$, where  $r'= \funref{axfsfadsrdsfsd}(g)+s$ and $q\geq  5/2\cdot \funref{axfsfsd}(g)$. The fact that $T\cap \ann({\cal A})=\emptyset$ implies that $T(L)\cap \ann({\cal A}')=\emptyset$ and thus $L$ is ${\cal A}'$-avoiding (see \autoref{sdfadfsdfgdgfdgfdg}).

	\begin{figure}[ht]
		\centering
	\scalebox{0.9}{		
		\sshow{0}{\begin{tikzpicture}[scale=.35]
		\foreach \x in {2,...,10}{
			\draw[line width =0.6pt] (0,0) circle (\x cm);
		}
		\begin{scope}[on background layer]
		\fill[applegreen!50!white] (0,0) circle (9 cm);
		\fill[white] (0,0) circle (3 cm);
		\end{scope}
		
		\node[track node 1] (1P1) at (30:2) {};
		\node[track node 1] (1P2) at (40:3) {};
		\node[track node 1] (1P3) at (45:4) {};
		\node[track node 1] (1P4) at (40:5) {};
		\node[track node 1] (1P5) at (50:6) {};
		\node[track node 1] (1P6) at (45:7) {};
		\node[track node 1] (1P7) at (40:8) {};
		\node[track node 1] (1P8) at (35:8) {};
		\node[track node 1] (1P9) at (30:9) {};
		\node[track node 1] (1P10) at (30:10) {};
		
		\foreach \i/\j in {1/2,2/3,3/4,4/5,5/6,6/7,7/8,8/9,9/10}{
			\draw[-] (1P\i) to (1P\j);
		}
		
		\node[track node 1] (2P1) at (70:2) {};
		\node[track node 1] (2P2) at (80:3) {};
		\node[track node 1] (2P3) at (75:4) {};
		\node[track node 1] (2P4) at (70:5) {};
		\node[track node 1] (2P5) at (80:6) {};
		\node[track node 1] (2P6) at (85:6) {};
		\node[track node 1] (2P7) at (80:7) {};
		\node[track node 1] (2P8) at (85:8) {};
		\node[track node 1] (2P9) at (80:9) {};
		\node[track node 1] (2P10) at (75:10) {};
		
		\foreach \i/\j in {1/2,2/3,3/4,4/5,5/6,6/7,7/8,8/9,9/10}{
			\draw[-] (2P\i) to (2P\j);
		}
		
		\node[track node 1] (3P1) at (130:2) {};
		\node[track node 1] (3P2) at (140:3) {};
		\node[track node 1] (3P3) at (135:4) {};
		\node[track node 1] (3P4) at (140:5) {};
		\node[track node 1] (3P5) at (130:6) {};
		\node[track node 1] (3P6) at (135:7) {};
		\node[track node 1] (3P7) at (140:7) {};
		\node[track node 1] (3P8) at (135:8) {};
		\node[track node 1] (3P9) at (130:9) {};
		\node[track node 1] (3P10) at (130:10) {};
		
		\foreach \i/\j in {1/2,2/3,3/4,4/5,5/6,6/7,7/8,8/9,9/10}{
			\draw[-] (3P\i) to (3P\j);
		}
		
		\node[track node 1] (4P1) at (180:2) {};
		\node[track node 1] (4P2) at (170:3) {};
		\node[track node 1] (4P3) at (175:3) {};
		\node[track node 1] (4P4) at (180:4) {};
		\node[track node 1] (4P5) at (170:5) {};
		\node[track node 1] (4P6) at (165:6) {};
		\node[track node 1] (4P7) at (170:7) {};
		\node[track node 1] (4P8) at (175:8) {};
		\node[track node 1] (4P9) at (170:9) {};
		\node[track node 1] (4P10) at (180:10) {};
		
		\foreach \i/\j in {1/2,2/3,3/4,4/5,5/6,6/7,7/8,8/9,9/10}{
			\draw[-] (4P\i) to (4P\j);
		}
		
		\node[track node 1] (5P1) at (210:2) {};
		\node[track node 1] (5P2) at (220:3) {};
		\node[track node 1] (5P3) at (225:4) {};
		\node[track node 1] (5P4) at (230:4) {};
		\node[track node 1] (5P5) at (220:5) {};
		\node[track node 1] (5P6) at (215:6) {};
		\node[track node 1] (5P7) at (220:7) {};
		\node[track node 1] (5P8) at (210:8) {};
		\node[track node 1] (5P9) at (220:9) {};
		\node[track node 1] (5P10) at (215:10) {};
		
		\foreach \i/\j in {1/2,2/3,3/4,4/5,5/6,6/7,7/8,8/9,9/10}{
			\draw[-] (5P\i) to (5P\j);
		}
		
		\node[track node 1] (6P1) at (270:2) {};
		\node[track node 1] (6P2) at (280:3) {};
		\node[track node 1] (6P3) at (275:4) {};
		\node[track node 1] (6P4) at (270:5) {};
		\node[track node 1] (6P5) at (265:6) {};
		\node[track node 1] (6P6) at (275:7) {};
		\node[track node 1] (6P7) at (285:8) {};
		\node[track node 1] (6P8) at (280:8) {};
		\node[track node 1] (6P9) at (270:9) {};
		\node[track node 1] (6P10) at (280:10) {};
		
		\foreach \i/\j in {1/2,2/3,3/4,4/5,5/6,6/7,7/8,8/9,9/10}{
			\draw[-] (6P\i) to (6P\j);
		}
		
		\node[track node 1] (7P1) at (-40:2) {};
		\node[track node 1] (7P2) at (-50:3) {};
		\node[track node 1] (7P3) at (-45:4) {};
		\node[track node 1] (7P4) at (-40:5) {};
		\node[track node 1] (7P5) at (-50:6) {};
		\node[track node 1] (7P6) at (-55:7) {};
		\node[track node 1] (7P7) at (-45:7) {};
		\node[track node 1] (7P8) at (-50:8) {};
		\node[track node 1] (7P9) at (-40:9) {};
		\node[track node 1] (7P10) at (-45:10) {};
		
		\foreach \i/\j in {1/2,2/3,3/4,4/5,5/6,6/7,7/8,8/9,9/10}{
			\draw[-] (7P\i) to (7P\j);
		}
		
		\node[model node] (M1) at (200:12) {};
		\node[model node] (M2) at (190:12) {};
		\node[model node] (M3) at (200:13) {};
		\node[model node] (M4) at (0:0) {};
		\node[model node] (M5) at (70:13) {};
		\node[model node] (M6) at (30:12) {};
		\node[model node] (M7) at (-20:15) {};

		\draw[celestialblue, line width=1pt] plot [smooth, tension=1] coordinates {(200:10) (200:8) (180:7) (160:6) (180:5.5) (160:2)};
		\draw[celestialblue, line width=1pt] (M1)-- (200:10) (160:2)-- (M4);
		\node[rep node] () at (200:10) {};
		\node[rep node] () at (160:2) {};
		
		\draw[celestialblue, line width=1pt] plot [smooth, tension=1] coordinates {(70:1) (90:4) (110:8) (90:10)};
		\draw[celestialblue, line width=1pt] (M4) -- (70:1) (90:10)--(M5);
		\node[rep node] () at (70:1) {};
		\node[rep node] () at (90:10) {};
		
		\draw[celestialblue, line width=1pt] plot [smooth, tension=1] coordinates {(230:11) (250:1) (290:8) (-40:10)};
		\draw[celestialblue, line width=1pt] (M1) -- (230:11) (-40:10)--(M7);
		\node[rep node] () at (230:11) {};
		\node[rep node] () at (-40:10) {};
		
		\draw[celestialblue, line width=1pt] plot [smooth, tension=1] coordinates { (0:1.5) (10:4) (-40:7) (-30:10.5)};
		\draw[celestialblue, line width=1pt] (M4) -- (0:1.5) (-30:10.5) --(M7);
		\node[rep node] () at (0:1.5) {};
		\node[rep node] () at (-30:10.5)  {};
		
		\draw[celestialblue, line width=1pt] plot [smooth, tension=1] coordinates {(-10:12) (20:6.5) (30:9) (20:11)};
		\draw[celestialblue, line width=1pt] (M7) -- (-10:12) (20:11)--(M6);
		\node[rep node] () at (-10:12) {};
		\node[rep node] () at (20:11) {};
		
		\draw[celestialblue, line width=1pt] plot [smooth, tension=1] coordinates {(-15:10) (-10:6) (60:6) (70:11)};
		\draw[celestialblue, line width=1pt] (M7) -- (-15:10) (70:11)--(M5);
		\node[rep node] () at (-15:10) {};
		\node[rep node] () at (70:11) {};
		
		\node[rep node] (m) at (45:12) {};
		\draw[celestialblue, line width=1pt] (M1)-- (M2) (M1) -- (M3) (M2) --  (M3) (M5) --(m) -- (M6);
		
		\end{tikzpicture}}}

		\caption{An example of a topological minor model $(M,T)$ of $H$ in $G$. Vertices of $T$ are depicted in blue while the neighbors of vertices of $T$ that are also subdividing vertices are depicted in red. Also, $\ann({\cal A}')$ is depicted in green.}
				\label{sdfadfsdfgdgfdgfdg}
	\end{figure}
		
	Let $I\subseteq [q]$, where $|I|>\funref{axfsfsd}(h)$. By applying \autoref{fsfsdfdsdsafsadffasdasfd2} for $s, g, G, {\cal A}', L$, and $I$ we obtain a linkage $\tilde{L}$ of $G$ such that $\tilde{L}\equiv  L$, $\tilde{L}$ is ${\cal A}'$-avoiding, $\tilde{L}\setminus \ann({\cal A'})\subseteq L \setminus \ann({\cal A'})$, and $\tilde{L}$ $(s,I)$-confined in ${\cal A}'$.
	We define
	$$\tilde{M} = (M\setminus L)\cup\tilde{L}.$$
	By definition, $(\tilde{M},T)$ is a topological minor model  of $H$ in $G$. Also, since $L, \tilde{L}\subseteq \ann({\cal A})$, then $\tilde{M}\setminus \ann({\cal A})\subseteq M \setminus \ann({\cal A})$. Finally, as $\tilde{L}$ is $(s,I)$-confined in ${\cal A}'$ then $\tilde{M}$ is $(s,I)$-confined in ${\cal A}$ as well.
\end{proof}

\section{Optimizing the Dynamic programming}
\label{dsfnskdalgnaklds}

According to the classic meta-algorithmic results of \cite{ArnborgLS91easy,BoriePT92} (see also~\cite{ArnborgP89,CourcelleM93})  computing ${\bf tm}_{\cal F}(G)$ can be done in $\mathcal{O}_{h,\tw}(n)$ time.
As we want to optimize the contribution of $k$ in our algorithm, we present in this section
a way to design a dynamic programming algorithm for computing ${\bf tm}_{\cal F}(G)$ 
in $2^{\mathcal{O}_{h}(\tw\log\tw)}n$ time. Actually, we give a more general statement of this result, \autoref{sssswfgfegwergewgwergwegr}, that 
will be useful in intermediate steps of our algorithm, presented in \autoref{asdfsdfdsfgsgdsg} and \autoref{gdgafssdfgsfhdsdhgshghsfgssghhfsgj}. In particular,  \autoref{sssswfgfegwergewgwergwegr}  will allow us to find (partial or complete) solutions to the (partial or complete) instances of the problem in bounded treewidth graphs.
In order to prove \autoref{sssswfgfegwergewgwergwegr}, we adapt the main ideas of \cite{arxivmoster} in our context.
Thus, in \autoref{ejflgjrlpe}, we define boundaried graphs and an equivalence relation among them with respect to the existence of certain topological minor models as subgraphs, which gives rise to a minimum-sized {\it representative} of each equivalence class.
In \autoref{kergjreklhjg}, we use known results from the protrusion machinery and bidimensionality theory \cite{arxivmoster,FominLST16, BasteST20hittI} to deduce that the size of each representative is bounded by a function of its boundary size.
Finally, in \autoref{jkelgjklrgj}, we define a notion of annotated boundaried graphs, the {\it enhanced boundaried graphs}, we extend the notions of equivalence of boundaried graphs to {\it meta-equivalence} of enhanced boundaried graphs, and we consider the {\it meta-representatives} of enhanced boundaried graphs.
We also prove that the number of different meta-representatives is bounded by a function of the boundary size (\autoref{fdhsdhsfghfgs}) and, using the dynamic programming tools of \cite{BasteST20hittI}, we conclude with an algorithm (\autoref{sssswfgfegwergewgwergwegr}) that computes the minimal size modulator of an enhanced boundaried graph to a given meta-representative.

\subsection{Boundaried graphs and representatives}\label{ejflgjrlpe}
We begin with some definitions, which originate in the seminal work of Bodlaender et al. \cite{BodlaenderFLPST16meta}.

\paragraph{Boundaried graphs.}
Let $t\in\Bbb{N}$.
A \emph{$t$-boundaried graph} is a triple $\bound{G} = (G,B,ρ)$ where $G$ is a graph, $B \subseteq V(G),$ $|B| =  t,$ and
$ρ : B \rightarrow [t]$ is an bijective function.
We call $B$ the {\em boundary} of ${\bf G}$ and we call the  vertices
of $B$ {\em the boundary vertices} of ${\bf G}$. We also
call $G$ {\em the underlying graph} of ${\bf G}$.
We say that the $t$-boundaried ${\bf G}'=(G',B',ρ')$ is a {\em subgraph}  of ${\bf G}$ if
$G'$ is a subgraph of $G$, $B'=B$, and $ρ'=ρ$. 
For $S\subseteq V(G)\setminus B$, we define ${\bf G}\setminus S$ to be the $t$-boundaried graph $(G',B, \rho)$ where $G'=G\setminus S$.
Also, for $B'\subseteq B$, we define the bijection $\rho[B']:B'\to [|B'|]$ such that for every $v\in B'$, $\rho[B'](v)=|\{u\in B'\mid \rho(u)\leq \rho(v)\}|$.
Two $t$-boundaried graphs ${\bf G}_{1}=(G_{1},B_{1},ρ_{1})$ and
${\bf G}_{2}=(G_{2},B_{2},ρ_{2})$ are {\em isomorphic}  if  $G_{1}$ is isomorphic to $G_{2}$ 
via a bijection $\phi: V(G_{1})\to V(G_{2})$ such that $ρ_{1}=ρ_{2}\circ\phi|_{B_{1}} ,$ i.e., the vertices of $B_{1}$ are mapped via $\phi$ to equally indexed vertices  of $B_{2}$.
A {\em boundaried graph} is any $t$-boundaried graph for some $t\in\Bbb{N}$.
As in \cite{RobertsonRXIII} (see also \cite{arxivmoster}), we define the {\em detail}
of a boundaried graph ${\bf G}=(G,B,\rho)$ as ${\sf detail}({\bf G}):=\max\{|E(G)|,|V(G)\setminus B|\}$.
Let $h,t\in \Bbb{N}$.
We denote by ${\cal B}^{(t)}$ the set of all  (pairwise non-isomorphic) $t$-boundaried graphs and by ${\cal B}^{(t)}_{h}$ the set of all (pairwise non-isomorphic) $t$-boundaried graphs with detail at most $h$.
We set ${\cal B}=\bigcup_{t\in\Bbb{N}} {\cal B}^{(t)}$.
{We say that a boundaried graph ${\bf G}=(G,B,\rho)$ is {\em planar} if $G$ is planar.}

We also define the {\em treewidth} of a boundaried graph ${\bf G}=(G,B,ρ)$, denoted by $\tw({\bf G})$, 
as the minimum width of a tree decomposition  $(T,\chi)$ of $G$ for which there is some $u\in V(T)$
such that $B\subseteq \chi(u)$. Notice that the treewidth of a $t$-boundaried graph is always lower bounded by $t-1$.



\paragraph{Equivalent boundaried graphs and representatives.}
We say that two boundaried graphs ${\bf G}_1=(G_1, B_1, \rho_1)$ and ${\bf G}_2=(G_2, B_2, \rho_2)$ are {\em compatible} if $\rho_{2}^{-1}\circ \rho_1$ is an isomorphism from $G_1[B_1]$ to $G_2[B_2]$.
Given two compatible boundaried graphs ${\bf G}_1=(G_1, B_1, \rho_1)$ and ${\bf G}_2=(G_2, B_2, \rho_2)$, we define ${\bf G}_1 \oplus {\bf G}_2$ as the graph obtained if we take the disjoint union of $G_1$ and $G_2$ and, for every $i\in[|B_1|]$ we identify vertices $\rho_{1}^{-1}(i)$ and $\rho_{2}^{-1}(i)$.


Given an $h\in \Bbb{N}$, we say that two boundaried graphs ${\bf G}_1$ and ${\bf G}_2$ are {\em $h$-equivalent}, denoted by ${\bf G}_1 \equiv_h {\bf G}_2$, if they are compatible and, for every graph $H$ on at most {$h$ vertices and $h$ edges (or, in other words, every $0$-boundaried graph $H$ with detail at most $h$)} and every boundaried graph ${\bf F}$ that is compatible with ${\bf G}_1$ (hence, with ${\bf G}_2$ as well), it holds that
$$ H\preceq {\bf F}\oplus {\bf G}_1 \Longleftrightarrow H \preceq {\bf F}\oplus {\bf G}_2.$$

Note that $\equiv_h$ is an equivalence relation on ${\cal B}$.  In the rest of this section we insist that 
${\cal B}$ contains only planar graphs, therefore  $\equiv_h$ is seen as an equivalence relation on
boundaried planar graphs.

A minimum-sized (first in terms of edges and then in terms of vertices) member of an equivalence class of $\equiv_h$ is called {\em representative} of $\equiv_h$. For every $t\in\Bbb{N}$, we denote by ${\cal R}^{(t)}_{h}$ the set of all $t$-boundaried graphs that are representatives of equivalence classes of $\equiv_h$. We also define the function ${\sf rep}: {\cal B}\to \bigcup_{t\in\Bbb{N}}{\cal R}^{(t)}_h$ that maps each boundaried graph to the representative of the equivalence class of $\equiv_h$ it belongs to.

\subsection{Bounding the size of a representative}\label{kergjreklhjg}
In this subsection we briefly present how the main idea of  \cite{arxivmoster} is applied in our context so as to bound the size of the representatives in ${\cal R}_h^{(t)}$.
We define a graph parameter that measures the minimum amount of vertices needed to {\it affect} every wall of the graph under consideration.
Following \cite[Corollary 25]{arxivmoster}, this parameter on representatives in ${\cal R}_h^{(t)}$ is linear in terms of $t$ (\autoref{ndsjgnamdgn}).
Moreover, we argue that, under the light of bidimensionality theory, this parameter is {\it contraction-bidimensional} and {\it linear-separable}.
We combine all above facts and employ a result of Fomin et al. \cite[Theorem 3.11]{FominLST16} and a slight extension of a result of Baste et al. \cite[Lemma 7.2]{BasteST20hittI} to obtain the desired linear bound on the size of every representative (\autoref{flmdgsd}).
To achieve this, we have to deal with {\it protrusion decompositions}, a notion which is also defined in this subsection and was introduced in \cite{BodlaenderFLPST16meta}.

\paragraph{Affecting walls in planar graphs.}
Let $r\in\Bbb{N}_{\geq 3}$, $G$ be a planar graph, $S\subseteq V(G)$, and $W$ be an $r$-wall of $G$.
We say that $S$ {\em affects $W$} if the following condition holds:
for every embedding $\Gamma$ of $G$ on the plane and every closed disk $\Delta$, if the connected component $\Delta_W$ of $\Bbb{R}^{2}\setminus {\sf perim}(W)$ that intersects $W$ is a subset of $\Delta$, then $S\cap \Delta_W \neq \emptyset$.
Given an $r\in\Bbb{N}_{\geq 3}$, we define the following graph parameter on planar graphs $${\bf p}_{r}(G) = \text{min}\{k\mid \exists S \subseteq V(G): \ |S|\leq k\  \wedge S \text{ affects every $r$-wall of $G$}\}.$$

{
From now on, functions $\funref{axfsfsd},\funref{axfsfadsrdsfsd}$ will always denote the functions of \autoref{jklnlnlk}.}
The following result is \cite[Corollary 25]{arxivmoster} {in the special case of planar graphs}.
\begin{proposition}\label{ndsjgnamdgn}
Let $h,t\in\Bbb{N}$.
There exists a function $\newfun{dgmksdgnl}: \Bbb{N}\to \Bbb{N}$ such that for every  ${\bf G}\in{\cal R}_h^{(t)}$ it holds that ${\bf p}_{\funref{dgmksdgnl}(h)}(G)\leq {t}$. 
Moreover, $\funref{dgmksdgnl}(h) = \mathcal{O}((\funref{axfsfsd}(h))^{3})$.
\end{proposition}

We comment that the function $\funref{dgmksdgnl}$ is obtained by \cite[Theorem 23 and Corollary 25]{arxivmoster} taking into account that in our case we deal with walls whose compass is planar.

\paragraph{Protrusion decompositions of boundaried graphs.}
Given a graph $G$, a set $X\subseteq V(G)$ is a {\em $\beta$-protrusion} of $G$ if $|\partial (X)|\leq \beta$ and $\tw (G[X])\leq \beta -1$.
Given a boundaried graph ${\bf G}=(G,B, \rho)$, a boundaried graph ${\bf G}'=(G', B', {\rho'})$ is a {\em $\beta$-protrusion} of ${\bf G}$ if 
\begin{itemize}

\item $V(G')$ is a $\beta$-protrusion of $G$,

\item $\tw ({\bf G}')\leq \beta-1$,

\item $\partial(V(G'))\subseteq B'$, 

\item $B\cap V(G')\subseteq B'$, and

\item $\rho'=\rho[B']$.
\end{itemize}

Given $\alpha, \ell\in \Bbb{N}$, an {\em $(\alpha, \beta)$-protrusion decomposition} of ${\bf G}$ is a sequence ${\cal P}=\langle R_0, \ldots, R_\ell\rangle$ of pairwise disjoint subsets of $V(G)$ such that

\begin{itemize}
\item $\bigcup_{i\in [0,\ell]} R_i=V(G)$,

\item $\max\{\ell, | R_0 |\}\leq \alpha$,

\item $B\subseteq R_0$,

\item for $i\in[\ell]$, the triple ${\bf G}'_{i}=(G'_{i},B'_{i},\rho'_{i})$, where $G_{i}'=G[N_G [R_i]]$, $B_{i}'=\partial(N_G [R_i])$, and $\rho_i'= \rho[B_{i}']\to [|B_{i}'|]$, is a $\beta$-protrusion of ${\bf G}$, and

\item for $i\in [\ell]$, $N_G (R_i)\subseteq R_0$.
\end{itemize}

\paragraph{Linear protrusion decompositions of representatives.}
Before we proceed, we give the definition of a {contraction-bidimensional} parameter, as defined in \cite{FominLST16}. A {\em graph parameter} is a function $\pi$ mapping graphs to non-negative integers.
We say that a graph parameter ${\bf \pi}$ is {\em contraction-bidimensional} if for every graph $G$ and every $e\in E(G)$ it holds that ${\bf \pi} (G/e) \leq {\bf \pi} (G)$ and ${\bf \pi} (\Gamma_k)=\Omega (k^2)$, where $\Gamma_k$ is the $(k\times k)$-triangulated grid\footnote{the graph  obtained from the $(k\times k)$-grid by adding, for all $1\leq x,y\leq k-1$, the edge with endpoints $(x+1,y)$ and $(x,y+1)$ and additionally making vertex $(k,k)$ adjacent to all the other vertices $(x,y)$ with $x\in\{1,k\}$ or $y\in \{1,k\}$, i.e., to the whole perimetric border of the grid.}.
Notice that ${\bf p}_r$ is contraction-bidimensional.

Moreover, it easy to observe that ${\bf p}_r$ is {\em linear-separable}, i.e., for every graph $G$, every set $S\subseteq V(G)$ of size ${\bf p}_r (G)$ that affects every $r$-wall of $G$, and every separation $(L,R)$ of $G$ it holds that  $||S\cap L| - {\bf p}_{r}(G[L])|=\mathcal{O}(|L\cap R|)$.

Since ${\bf p}_r$ is a contraction-bidimensional and linear-separable parameter on planar graphs, \cite[Theorem 3.11]{FominLST16} together with \autoref{ndsjgnamdgn} imply the following result.

\begin{lemma}\label{protrusion1}
Let $h,t\in\Bbb{N}$. There is a constant $c_h$  such that every ${\bf G}\in {\cal R}_{h}^{(t)}$ admits a
$(c_h \cdot t 
, c_h)$-protrusion decomposition.
{Moreover, $c_h=(\funref{axfsfsd}(h))^{\mathcal{O}(1)}$.}
\end{lemma}

%
%

Using \autoref{protrusion1},
the definition of {$(\alpha, \beta)$-protrusion decompositions},
and the arguments of the proof of~\cite[Lemma 7.2]{BasteST20hittI},
it follows
that the size of every ${\bf G}\in {\cal R}_{h}^{(t)}$ is $\mathcal{O}_h (t)$.

\begin{lemma}\label{flmdgsd}
There is a function $\newfun{ngklagnokr}:\Bbb{N}\to \Bbb{N}$ such that for every $t,h\in \Bbb{N}$, if ${\bf G}$ is a  boundaried graph in ${\cal R}_h^{(t)}$, then the underlying graph of ${\bf G}$ has at most $\funref{ngklagnokr}(h)\cdot t$ vertices.
{Moreover, $\funref{ngklagnokr}(h) = 2^{2^{2^{(\funref{axfsfsd}(h))^{\mathcal{O}(1)} \log \funref{axfsfsd}(h)}}}$.}
\end{lemma}

%
\begin{proof}
Let $s:=f(c_h,h)$, where $f$ is the function of~\cite[Lemma 7.2]{BasteST20hittI} and $c_h$ is the constant from~\autoref{protrusion1}.
From the proof of~\cite[Lemma 7.2]{BasteST20hittI} we can derive that
$f(c_h,h)=2^{2^{2^{\mathcal{O}(c_h \log c_h)}}}$.
Also, we set $\funref{ngklagnokr}(h) =(s+1)\cdot c_h$.

Let ${\bf G} = (G,B,\rho)$ be a boundaried graph in ${\cal R}_{h}^{(t)}$.
We will prove that $G$ has at most $\funref{ngklagnokr}(h)\cdot t$ vertices.
By~\autoref{protrusion1}, ${\bf G}$ admits a $(c_h \cdot t , c_h)$-protrusion decomposition.
Therefore, by definition,
there is an $\ell\in[0,c_h\cdot t]$ and a sequence
${\cal P}=\langle R_0, \ldots, R_\ell\rangle$
of pairwise disjoint subsets of $V(G)$ such that
\begin{itemize}
\item $\bigcup_{i\in [0,\ell]} R_i=V(G)$,

\item $\max\{\ell, | R_0 |\}\leq c_h\cdot t$,

\item $B\subseteq R_0$,

\item for $i\in[\ell]$, the triple ${\bf G}'_{i}=(G'_{i},B'_{i},\rho'_{i})$, where $G_{i}'=G[N_G [R_i]]$, $B_{i}'=\partial(N_G [R_i])$, and $\rho_i'= \rho[B_{i}']\to [|B_{i}'|]$, is a $c_h$-protrusion of ${\bf G}$, and

\item for $i\in [\ell]$, $N_G (R_i)\subseteq R_0$.
\end{itemize}
For every $i\in[\ell]$, we will show that $G'_{i}$ has at most $s$ vertices.
Suppose that, towards a contradiction, $|V(G'_{i})|>s$.
We sketch the proof of~\cite[Lemma 7.2]{BasteST20hittI} and we comment how to adjust it to our setting in order to obtain a contradiction to the fact that  ${\bf G}\in {\cal R}_{h}^{(t)}$.
First of all, in the statement of~\cite[Lemma 7.2]{BasteST20hittI} it is required that the family $\mathcal{F}$ contains a planar graph, an assumption that is not true in our case.
However, inside the proof this is only used in order to bound the treewidth of $G'_{i}$ (in fact, to bound its branchwidth, which we know that is upper-bounded by its treewidth).
Here,
the fact that the treewidth of $G'_{i}$ is at most $c_h$ is implied by the fact that ${\bf G}'_{i}$ is a $c_h$-protrusion of ${\bf G}$.
After this step, the proof of our lemma follows the same
arguments as the one of~\cite[Lemma 7.2]{BasteST20hittI}.
Intuitively, having a bound on the branchwidth of $G'_{i}$ we can consider a branch decomposition of $G'_{i}$ of 
bounded width.
Since we assume that  $|V(G'_{i})|>s$ and we have that  ${\bf G}\in {\cal R}_{h}^{(t)}$,
the number of edges of $G'_i$ is ``large enough'' and therefore the branch decomposition contains a ``long 
enough'' path from the root of the decomposition to a leaf.
Along this ``long enough''  path, we can find two boundaried graphs ${\bf G}''_{1}$ and ${\bf G}'_{2}$ such that
the underlying graph of ${\bf G}''_{2}$ is a subgraph of the underlying graph of ${\bf G}''_{1}$ that has
less edges and ${\bf G}''_{1}\equiv_h{\bf G}'_{2}$.
Therefore, by replacing ${\bf G}''_{1}$ with ${\bf G}'_{2}$, we obtain a boundaried graph that is $h$-equivalent to ${\bf G}$ and whose underlying graph has less edges than $G$, a contradiction to the fact that
 ${\bf G}\in {\cal R}_{h}^{(t)}$. For more details, we refer the reader to the proof of~\cite[Lemma 7.2]{BasteST20hittI}.
\end{proof}

\subsection{Enhanced boundaried graphs}\label{jkelgjklrgj}
In this subsection we define a notion of annotated boundaried graphs, the {\it enhanced boundaried graphs}.
Mirroring the equivalence relation defined in \autoref{ejflgjrlpe} for boundaried graphs, in this subsection we also define {\it meta-equivalences} and {\it meta-representatives} of enhanced boundaried graphs and we prove that the size of a meta-representative is also linear bounded in terms of its boundary size.
This directly implies that an almost single-exponential (with a logarithmic factor error) bound on the number of different meta-representatives (\autoref{dsgldglds}).
Finally, using the ideas of \cite{BasteST20hittI} concerning dynamic programming, we arrive to the main result of this section, \autoref{sssswfgfegwergewgwergwegr}.

\paragraph{Enhanced boundaried graphs.}
Let $t\in\Bbb{N}$ and $q\in\Bbb{N}_{\geq 1}$.
A {\em $q$-enhanced $t$-boundaried graph} is a triple $({\bf G}, {\cal Z}, {\cal V})$, where ${\bf G}=(G, B, \rho)$ is a $t$-boundaried graph, ${\cal Z}=\{Z_1, \dots, Z_q\}$ is a collection of non-empty subsets of $B$, and ${\cal V}=\{V_1, \ldots,V_q\}$ is a collection of non-empty subsets of $V(G)$, such that for every $i\in[q]$, $(G[V_i], Z_i, \rho[Z_i])$ is a boundaried graph.
A {\em $q$-enhanced boundaried graph} is a $q$-enhanced $t$-boundaried graph for some $t\in\Bbb{N}$.
The {\em treewidth} of a $q$-enhanced boundaried graph $({\bf G}, {\cal Z}, {\cal V})$ is the treewidth of ${\bf G}$.
As we did in the previous subsection, we consider only triples $({\bf G}, {\cal Z}, {\cal V})$ where $G$ is planar.

\paragraph{Meta-representatives of enhanced boundaried graphs.}
We say that two $q$-enhanced boundaried graphs $({\bf G}, {\cal Z}, {\cal V})$ and $({\bf G}', {\cal Z}', {\cal V}')$ are {\em $h$-meta-equivalent}, denoted by $({\bf G}, {\cal Z}, {\cal V})\equiv_h^{(q)} ({\bf G}', {\cal Z}', {\cal V}')$, if the following hold:
\begin{itemize}
\item ${\bf G}$ and ${\bf G}'$ are compatible (via the isomorphism ${\rho'}^{-1}\circ \rho$ from $G[B]$ to $G'[B']$) and 
\item for every $i\in[q]$, it holds that $(G[V_i], Z_i, \rho_i)\equiv_h (G'[V_i '], Z_i ', \rho_i ')$.
\end{itemize}

Notice that $\equiv_h^{(q)}$ defines an equivalence relation on $q$-enhanced boundaried graphs. The minimum-sized (first in terms of edges and then in terms of vertices)  member of each equivalence class of $\equiv_h^{(q)} $ is called a
{\em meta-representative} of $\equiv_h^{(q)}$. We denote by ${\cal R}_h^{(q,t)}$ the set of all $q$-enhanced $t$-boundaried graphs that are meta-representatives of $\equiv_h^{(q)}$. We call $|V(G)|$ the {\em size} of a {meta-representative} $({\bf G}, {\cal Z}, {\cal V})$  of $\equiv_h^{(q)}$.

\begin{lemma}\label{fdhsdhsfghfgs}
There is a function $\funref{ngklagnokr}:\Bbb{N}\to \Bbb{N}$ such that for every $t,h\in \Bbb{N}$ and $q\in \Bbb{N}_{\geq 1}$, every  $\bar{\bf G}\in {\cal R}_h^{(q,t)}$ has size at most $\funref{ngklagnokr}(h)\cdot q\cdot t$.
\end{lemma}

\begin{proof}
Let $t,h\in\Bbb{N}$, $q\in \Bbb{N}_{\geq 1}$ and let $({\bf G}, {\cal V}, {\cal Z})$ be a $q$-enhanced $t$-boundaried graph, where ${\bf G}=(G,B,\rho)$. 
For every $i\in[q]$, we denote by ${\bf H}_i$ the graph ${\sf rep}((G[V_i], Z_i, \rho_i))$ and  notice that, due to \autoref{flmdgsd}, the underlying graph of ${\bf H}_i$ has at most $\funref{ngklagnokr}(h)\cdot t$ vertices.
By setting ${\bf G}'=(\bigcup_{i\in [q]} H_i, B, \rho)$ and ${\cal V}'=\{V(H_1),\ldots, V(H_q)\}$,
we observe that the triple $({\bf G}', {\cal V}', {\cal Z})$
is a  $q$-enhanced $t$-boundaried graph that is $h$-meta-equivalent with $({\bf G}, {\cal V}, {\cal Z})$
and its size is at most $\funref{ngklagnokr}(h)\cdot q\cdot t$.
Thus, the meta-representative of the equivalence class of $\equiv_h^{(q)}$ that contains $({\bf G}, {\cal V}, {\cal Z})$ has size at most $\funref{ngklagnokr}(h)\cdot q\cdot t$.
\end{proof}

The following result is a direct consequence of \autoref{fdhsdhsfghfgs}, using the fact that the underlying graph $G$ of every meta-representative of $\equiv_h^{(q)}$ is planar, hence it has $\mathcal{O}(|V(G)|)$ edges.

\begin{corollary}\label{dsgldglds}
There is a function $\newfun{dskglasdgklnfdkn}: \Bbb{N}^3\to \Bbb{N}$ such that for every $t,h\in \Bbb{N}$ and $q\in\Bbb{N}_{\geq 1}$, $|{\cal R}_h^{(q,t)}|\leq \funref{dskglasdgklnfdkn}(h,q,t)$. {Moreover, it holds that $\funref{dskglasdgklnfdkn}(h,q,t)=2^{\mathcal{O}(\funref{ngklagnokr}(h)\cdot q\cdot t\cdot \log (q\cdot t))}.$}
\end{corollary}

Let  $q\in \Bbb{N}_{\geq 1}$ and $({\bf G}, {\cal V}, {\cal Z})$ be a $q$-enhanced  boundaried graph,
where ${\bf G}=(G,B,\rho)$. We say that a set $S\subseteq V(G)$ is {\em boundary-avoiding} if $S\cap B=\emptyset$.
The dynamic programming machinery of \cite[Theorem 1 and Theorem 5]{BasteST20hittI} together with \autoref{dsgldglds} 
and the single-exponential $5$-approximation algorithm of Bodlaender et al. for treewidth \cite[Theorem VI]{BodlaenderDDFLP16} yield the following result.

\begin{lemma}
\label{sssswfgfegwergewgwergwegr}
There is an algorithm with the following specifications:

\smallskip\noindent {\bf Compute\_rep}$(h,q,t,w,k,\bar{\bf G},R, \bar{\bf J})$\\
\noindent{\sl Input:} five integers $t,w, k\in\Bbb{N}$ and $h,q\in\Bbb{N}_{\geq 1}$, where $h\geq t, q$, a  $q$-enhanced (planar) $t$-boundaried graph $\bar{\bf G}$ of treewidth at most $w$, a boundary-avoiding set $R\subseteq V(G)$, and a meta-representative $\bar{\bf J}\in {\cal R}_h^{(q,t)}$.

\noindent{\sl Output:} if exists, the minimum-size set $S_{\bar{\bf J}}\subseteq R$ of size at most $k$ such that $\bar{\bf G}\setminus S_{\bar{\bf J}}\equiv_h^{(q)} \bar{\bf J}$.

\noindent This algorithm runs in $2^{\mathcal{O}(\funref{ngklagnokr}(h)\cdot w\log w)}\cdot n$ time   or, alternatively, in $\mathcal{O}(n^3)+2^{\mathcal{O}(\funref{ngklagnokr}(h)\cdot w)}\cdot n$ time.
\end{lemma}

\begin{proof}
The algorithm first applies the single-exponential $5$-approximation algorithm of Bodlaender et al. for treewidth \cite[Theorem VI]{BodlaenderDDFLP16} to compute a tree decomposition of the underlying graph of $\bar{\bf G}$ of width at most $5w$.
Then,
using the dynamic programming algorithm of~\cite[Theorem 1]{BasteST20hittI},
it checks,
for every $\bar{\bf I}\in {\cal R}_h^{(q,t)}$,
whether there is a set $S_{\bar{\bf I}}\subseteq V(G)$ of size at most $k$
such that $\bar{\bf G}\setminus S_{\bar{\bf I}}\equiv_h^{(q)} \bar{\bf I}$.
We can easily modify this dynamic programming algorithm 
so as it checks whether such a set $S_{\bar{\bf I}}$ is a subset of $R$ and to also output a minimum-size $S_{\bar{\bf I}}$ satisfying all above properties, if such exists.
As, due to~\autoref{dsgldglds}, $|{\cal R}_h^{(q,t)}|\leq \funref{dskglasdgklnfdkn}(h,q,t)$,
this algorithm runs in time $2^{\mathcal{O}(\funref{ngklagnokr}(h)\cdot w\log w)}\cdot n$.
Moreover,
we can replace the algorithm of~\cite[Theorem 1]{BasteST20hittI} with the one of~\cite[Theorem 5]{BasteST20hittI}, which runs in a special branch decomposition of the underlying graph of $\bar{\bf G}$ (called {\sl sphere-cut decomposition}) and  performs in time 
$\mathcal{O}(n^3)+2^{\mathcal{O}(\funref{ngklagnokr}(h)\cdot w)}\cdot n$.
\end{proof}

We stress that, in the case $q=1$, we simply refer to $q$-enhanced $t$-boundaried graphs as $t$-boundaried graphs, and to the equivalence relation $\equiv_{h}^{(q)}$ as $\equiv_h$. Following this, meta-representatives of $\equiv_h^{(1)}$ are just called representatives. \autoref{sssswfgfegwergewgwergwegr} is applied three times in this paper. First we use it in the proof of \autoref{fsfsdfdsdsafsasdsdadffasdasfd} in \autoref{asdfsdfdsfgsgdsg} towards reducing the solution size, second we use it 
with $q=1$  and $k=0$ in  the proof of \autoref{fsfsdfdsdsafsadffasdasfd1} in \autoref{gdgafssdfgsfhdsdhgshghsfgssghhfsgj}, and we also use it with  $t=0$, $k=0$, and $q=1$, in order to deduce \autoref{dfassdfasfdfdsad} from~\autoref{dfasdfdsad}.

\section{The two main subroutines of the algorithm}
\label{asfdsasdfdsf}

In this section, we provide two main subroutines that will be useful in the proof of \autoref{dfasdfdsad}.
In subsection \autoref{asdfsdfdsfgsgdsg}, we provide an algorithm (\autoref{fsfsdfdsdsafsasdsdadffasdasfd}), that allows us to ``safely'' reduce the set of possible candidates to a solution, while in \autoref{gdgafssdfgsfhdsdhgshghsfgssghhfsgj}, we provide an algorithm (\autoref{fsfsdfdsdsafsadffasdasfd1}) that outputs a ``big enough'' wall such that the vertices in its compass are irrelevant with respect to the existence of a solution to the problem.\medskip

Before proceeding to the algorithmic results, we provide some definitions that will facilitate the presentation of the proofs.

\paragraph{Boundaried graphs in railed annuli.}
Let ${\cal A}=({\cal C}, {\cal P})$ be a {$(r,q)$-railed annulus} of a partially $\Delta$-embedded graph $G$.
We can see each path $P_{j}$ in ${\cal P}$ as being oriented towards the ``inner'' part of $\Delta$, i.e., starting from an endpoint of  $P_{1,j}$ and finishing to an endpoint of $P_{r,j}$. 
For every $(i,j)\in[r]\times[q],$   we define $r_{i,j}$ as the first vertex of $P_{j}$ that appears in $P_{i,j}$
while traversing $P_{j}$ according to this orientation. 
Given an $i\in[r]$ and a $t\in [q]$, 
we define the $t$-boundaried graph ${\bf G}_{i,t}=(G_{i},B_{i,t},ρ_{i,t})$
where $G_{i}=G\cap \overline{D}_{i}$, $B_{i,t}=\{r_{i,1}\,\ldots,r_{i,t}\}$ and, for $j\in[t]$, $ρ_{i,t}(r_{i,j})=j$.

\subsection{Reducing the solution space}
\label{asdfsdfdsfgsgdsg}
We now prove the following lemma that intuitively states that there is an algorithm that 
given a graph $G$ and a ``big enough'' railed annulus ${\cal A}$ of $G$, it ``reduces'' the set of vertices that are candidates to the set $S$ that certifies that {${\bf tm}_{\cal F}(G)\leq k$}.

\begin{lemma}
\label{fsfsdfdsdsafsasdsdadffasdasfd}
There are three functions $\newfun{dsalgmdlaotpopyrt},\newfun{dsfndnvgdsqwert}: \Bbb{N}^{2}\to \Bbb{N}$, $\newfun{dfljhklgdfjhklgfj}: \Bbb{N}\to \Bbb{N}$ and an algorithm with the following specifications:

\smallskip\noindent {\bf Reduce\_Solution\_Space}$(k,h,w,{\cal F},\Delta,G,R,{\cal C},{\cal P})$\\
\noindent{\sl Input:} three  integers $k,h,w\in\Bbb{N}$,
a finite set ${\cal F}$ of graphs such that $h\leq h({\cal F})$, a partially $\Delta$-embedded graph $G$ whose compass has treewidth at most $w$, {a set $R\subseteq V(G)$}, and an  $(r,q)$-railed 
annulus ${\cal A}=({\cal C},{\cal P})$ of $G$, where  $r=\funref{dsfndnvgdsqwert}(h,k)$ 
and $q\geq  \funref{dfljhklgdfjhklgfj}(h)$. 

\noindent{\sl Output:} a set $R'\subseteq R$ such that
\begin{itemize}
\item $|R'\cap D_r|\leq \funref{dsalgmdlaotpopyrt}(h,k)$ and
\item if $(G,R,k)$ is a ${\bf tm}_{\cal F}$-triple then $(G,R',k)$ is a ${\bf tm}_{\cal F}$-triple. 
\end{itemize} 
Moreover, $\funref{dsfndnvgdsqwert}(h,k)=\mathcal{O}_{h} (k)$, $\funref{dsalgmdlaotpopyrt}(h,k)=\mathcal{O}_{h} (k^2)$, and the algorithm runs in $2^{\mathcal{O}_{h} (w \log w)}\cdot k\cdot n$ time,  or, alternatively, in $\mathcal{O}(k\cdot n^3)+2^{\mathcal{O}_{h}(w)}\cdot k\cdot n$ time.
\end{lemma}

\begin{proof}
Let $g:= \genfrac{(}{)}{0pt}{}{h}{2}$, $\lambda:=\funref{axfsfsd}(g) + 1$, $μ:=\funref{axfsfadsrdsfsd}(g)+3$,
\begin{align*}
\funref{dsfndnvgdsqwert}(h,k) := & \ (k+1)(h+1)\mu,\\
\funref{dsalgmdlaotpopyrt}(h,k):= & \ \funref{dskglasdgklnfdkn}(h+\lambda, h+1, (h+1)\cdot\lambda)\cdot k(k+1), \text{ and}\\
\funref{dfljhklgdfjhklgfj}(h):= & \ 5/2\cdot \funref{axfsfsd}(g).
\end{align*}
Given an $i \in [k+1]$, we define $A_{i}=\ann({\cal C},(i-1)(h+1)μ+1,i(h+1)μ)$ and for every $j\in[h+1]$ we define $Q_{i,j} = \ann({\cal C},(i-1)(h+1)μ+ (j-1)μ+1, (i-1)(h+1)μ+jμ).$
Intuitively, we partition ${\cal C}$ into $k+1$ sets of consecutive cycles (i.e., the cycles of $A_{i}, i\in[k+1]$) and then, for every $i\in [k+1]$ we further partition the set of cycles of $A_{i}$ into $h+1$ sets of consecutive cycles (i.e., the cycles of $Q_{i, j}$, $j\in [h+1]$). Notice that for every $i,j\in[k+1]\times[h+1]$, $|Q_{i,j}\cap {\cal C}| = μ$ (see \autoref{asdfdsfdsf}). 
\begin{figure}[H]
\centering\scalebox{1}{
\sshow{0}{\begin{tikzpicture}[scale=0.38]
\begin{scope}[xshift = 16 cm]
	\foreach \i in {0,0.5,3, 5.5, 6}{
		\draw[-] (\i,0) -- (\i, 7);
	}
	
	\draw[opacity=0.3] (1,0) -- (1,7);
	\draw[opacity=0.1] (1.5,0) -- (1.5,7);
	
	\draw[opacity=0.2] (2.5,0) -- (2.5,7);
	\draw[opacity=0.2] (3.5,0) -- (3.5,7);
	
	\draw[opacity=0.1] (4.5,0) -- (4.5,7);
	\draw[opacity=0.3] (5,0) -- (5,7);
\end{scope}
	
\begin{scope}[xshift=0.5 cm]
	\draw[-] (0.5,0) -- (0.5, 7);
	\draw[opacity=0.3] (1,0) -- (1,7);
	\draw[opacity=0.1] (1.5,0) -- (1.5,7);
	\draw[opacity=0.1] (2.5,0) -- (2.5,7);
	\draw[opacity=0.3] (3,0) -- (3,7);
	\draw[-] (3.5,0) -- (3.5, 7);
	\draw[opacity=0.1] (4,0) -- (4,7);
\end{scope}
	
\foreach \i in {5,11, 23,29}{
	\begin{scope}[xshift=\i cm]
		\draw[opacity=0.1] (0,0) -- (0,7);
		\draw[-] (0.5,0) -- (0.5, 7);
		\draw[opacity=0.3] (1,0) -- (1,7);
		\draw[opacity=0.1] (1.5,0) -- (1.5,7);
		\draw[opacity=0.1] (2.5,0) -- (2.5,7);
		\draw[opacity=0.3] (3,0) -- (3,7);
		\draw[-] (3.5,0) -- (3.5, 7);
		\draw[opacity=0.1] (4,0) -- (4,7);
	\end{scope}
	}

\begin{scope}[xshift=33.5  cm]

\draw[-] (0.5,0) -- (0.5, 7);
\draw[opacity=0.3] (1,0) -- (1,7);
\draw[opacity=0.1] (1.5,0) -- (1.5,7);
\draw[opacity=0.1] (2.5,0) -- (2.5,7);
\draw[opacity=0.3] (3,0) -- (3,7);
\draw[-] (3.5,0) -- (3.5, 7);
\end{scope}

	\node () at (1,8) {$C_{1}$};
	\node () at (37,8) {$C_{r}$};
	\node[anchor=south] () at (11.5, 7.8) {$C_{(i-1)(h+1)μ+1}$};
	\draw[red] (11.5,8) -- (11.5,7.5);
	\node[small red] () at (11.5,7.5) {};
	\node[anchor=south] () at (14, 9) {$C_{(i-1)(h+1)μ+(j-1)μ+1}$};
	\draw[red] (16,9) -- (16,7.5);
	\node[small red] () at (16,7.5) {};
	\node[anchor=south] () at (24, 9) {$C_{(i-1)(h+1)μ+jμ}$};
	\draw[red] (22,9) -- (22,7.5);
	\node[small red] () at (22,7.5) {};
	\node[anchor=south] () at (19, 10.2) {$C_{(i-1)(h+1)μ+(j-1)μ+\lceil μ/2\rceil}$};
	\draw[red] (19,10) -- (19,7.5);
	\node[small red] () at (19,7.5) {};
	\node[anchor=south] () at (26.5, 7.8) {$C_{i(h+1)μ}$};
		\draw[red] (26.5,8) -- (26.5,7.5);
	\node[small red] () at (26.5,7.5) {};
	\draw [decorate,decoration={brace,amplitude=10pt, mirror}] (1,-1.7) -- (8.5,-1.7) node [black,midway,yshift=-.5cm] {\footnotesize $A_1$};
	\draw [decorate,decoration={brace,amplitude=5pt, mirror}] (1,-.3) -- (4,-.3) node [black,midway,yshift=-.4cm] {\footnotesize $Q_{1,1}$};
	\draw [decorate,decoration={brace,amplitude=5pt, mirror}] (5.5,-.3) -- (8.5,-.3) node [black,midway,yshift=-.4cm] {\footnotesize $Q_{1,h+1}$};
	
	\draw [decorate,decoration={brace,amplitude=10pt, mirror}] (11.5,-1.7) -- (26.5,-1.7) node [black,midway,yshift=-.5cm] {\footnotesize $A_{i}$};
	\draw [decorate,decoration={brace,amplitude=5pt, mirror}] (11.5,-.3) -- (14.5,-.3) node [black,midway,yshift=-.4cm] {\footnotesize $Q_{i,1}$};
	\draw [decorate,decoration={brace,amplitude=5pt, mirror}] (16,-.3) -- (22,-.3) node [black,midway,yshift=-.4cm] {\footnotesize $Q_{i,j}$};
	\draw [decorate,decoration={brace,amplitude=5pt, mirror}] (23.5,-.3) -- (26.5,-.3) node [black,midway,yshift=-.4cm] {\footnotesize $Q_{i,h+1}$};
	
	\draw [decorate,decoration={brace,amplitude=10pt, mirror}] (29.5,-1.7) -- (37,-1.7) node [black,midway,yshift=-.5cm] {\footnotesize $A_{k+1}$};
	\draw [decorate,decoration={brace,amplitude=5pt, mirror}] (29.5,-.3) -- (32.5,-.3) node [black,midway,yshift=-.4cm] {\footnotesize $Q_{k+1,1}$};
	\draw [decorate,decoration={brace,amplitude=5pt, mirror}] (34,-.3) -- (37,-.3) node [black,midway,yshift=-.4cm] {\footnotesize $Q_{k+1,h+1}$};
	
	\node() at (4.8,-.3) {$\ldots$};
	\node() at (10.1,-1.5) {$\ldots$};
	\node() at (15.3,-.3) {$\ldots$};
	\node() at (22.7,-.3) {$\ldots$};
	\node() at (28.1,-1.5) {$\ldots$};
	\node() at (33.3,-.3) {$\ldots$};
	\end{tikzpicture}}}

	\caption{Visualization of the partition of the cycles of ${\cal A}$ into sets $A_{i}, i \in [k+1]$ and of the partition in sets $Q_{i,j}, i,j\in [k+1]\times [h+1]$.}
	\label{asdfdsfdsf}
\end{figure}

Also, we define for every $(i,j)\in[k+1]\times[h+1]$ the $\lambda$-boundaried 
graph $$(G_{i,j}, B_{i,j}, \rho_{i,j}):= {\bf G}_{(i-1)(h+1)μ+(j-1)μ+\lceil μ/2\rceil,\lambda}.$$
To get some intuition, notice that the vertices of $B_{i,j}$ lie on the ``middle'' cycle of $Q_{i,j}$ -- see \autoref{asdfdsfdsf}.
 
Now, for every $i\in[k+1]$, we aim to define a $(h+1)$-enhanced $((h+1)\cdot \lambda)$-boundaried graph obtained by the union of the $\lambda$-boundaried graphs $(G_{i,j}, B_{i,j}, \rho_{i,j})$, for $j\in[h+1]$.
Let  $i\in[k+1]$.
We set  $\bar{\rho}_i: \bigcup_{j\in[h+1]} B_{i,j}\to [(h+1)\cdot \lambda]$ to be the function that, for every $j\in[h +1]$, $$\bar{\rho}_i (v)=(j-1)\cdot \lambda +{\rho}_{i,j}(v), \text{ if }v\in B_{i,j}.$$
Notice that $\bar{\rho}_i$ is a bijection and that, for every $j\in[h+1]$,
$\bar{\rho}_i [B_{i,j}]=\rho_{i,j}$.
Thus, if we set
\begin{align*}
\bar{\bf G}_i:= & \ (\bigcup_{j\in[h+1]} {G}_{i,j},  \bigcup_{j\in[h+1]} B_{i,j}, \bar{\rho}_i),\\ 
{\cal Z}_{i}:= & \ \{B_{i,1},\ldots, B_{i,h+1}\}, \text{ and }\\
{\cal V}_{i}:= & \ \{V(G_{i,1}),\ldots, V(G_{i,h+1})\},
\end{align*}
we have that $(\bar{\bf G}_i, {\cal Z}_i, {\cal V}_i)$ is an $(h+1)$-enhanced $((h+1)\cdot \lambda )$-boundaried graph.

For every $i\in[k+1]$ and every meta-representative $\bar{\bf J}\in{\cal R}_{h+\lambda}^{(h+1,(h+1)\cdot \lambda )}$, let  $S_{i,\bar{\bf J}}$ be the minimum-size subset of $R\cap D_{i(h+1)μ}$
of at most $k$ vertices such that $(\bar{\bf G}_i, {\cal Z}_i, {\cal V}_i)\setminus S_{i,\bar{\bf J}}\equiv_{h+\lambda}^{(h+1)}  \bar{\bf J}$.
If such a set does not exist, then we set $S_{i,\bar{\bf J}}=\emptyset$.
We define $$R^*=(\bigcup_{\genfrac{}{}{0pt}{1}{i\in[k+1]}{\bar{\bf J}\in  {\cal R}_{h+\lambda}^{(h+1,(h+1)\cdot \lambda)}}} S_{i,\bar{\bf J}})\cap D_{r}$$
and $R'=(R\setminus D_r)\cup R^*$.
Observe that $|R^*| \leq (k+1)\cdot |{\cal R}_{h+\lambda}^{(h+1,(h+1)\cdot \lambda)}|\cdot k = \funref{dsalgmdlaotpopyrt}(h,k)$, and therefore $|R'\cap D_r|\leq \funref{dsalgmdlaotpopyrt}(h,k)$.
Also, notice that as the underlying graph of $\bar{\bf G}_{i}$ is a subgraph of the compass of $G$, it has treewidth at most $w$.
Moreover, since $S_{i,\bar{\bf J}}\subseteq R\cap D_{i(h+1)\mu}$, it holds that $R\cap \bigcup_{j\in[h+1]} B_{i,j} = \emptyset$.
Therefore, for every $i\in[k+1]$ and every $\bar{\bf J}\in  {\cal R}_{h+\lambda}^{(h+1,(h+1)\cdot \lambda)}$,
we compute $S_{i,\bar{\bf J}}$ by using the algorithm of \autoref{sssswfgfegwergewgwergwegr} for 
$h:=h+\lambda$,
$q:= h+1$,
$t:=(h+1)\cdot \lambda$,
$\bar{\bf G}:=\bar{\bf G}_i$, and
$R:=R\cap D_{i(h+1)\mu}$.
This algorithm runs in $2^{\mathcal{O}_h (w\log w)}\cdot n$ time, or, alternatively, in $\mathcal{O}(n^3)+2^{\mathcal{O}_{h}(w)}\cdot n$ time, since the underlying graph of each $\bar{\bf G}_i$ is planar.
Therefore, we can compute $R'$ as well,  in $2^{\mathcal{O}_h (w\log w)}\cdot (k+1)\cdot n$ time, or, alternatively, in $(k+1)\cdot \mathcal{O}(n^3)+(k+1)\cdot 2^{\mathcal{O}_{h}(w)}\cdot n$ time.

We now prove that  if $(G,R,k)$ is a ${\bf tm}_{\cal F}$-triple then $(G,R',k)$ is also a ${\bf tm}_{\cal F}$-triple.
In particular, we prove that for every graph $H$ on at most $h$ vertices and every  $S\subseteq R$, if  $|S|\leq k$ and $H\npreceq G\setminus S$,  then there is some $S'\subseteq R'$ such that $|S'|\leq k$ and $H\npreceq G\setminus S'$.

Let $H$ be  graph on at most $h$ vertices (and, therefore, of at most $g$ edges) and let $S\subseteq R$ such that  $|S|\leq k$ and $H\npreceq G\setminus S$.
As $r= (k+1)(h+1)μ$ and $|S|\leq k$, then by the pigeonhole principle there is some $\ell \in[k+1]$
such that $S\cap {A}_{\ell}=\emptyset.$
(In case there are many such $\ell$'s, we take the minimum one.)
Let $S_{\rm in}=S\cap D_{\ell(h+1)μ}$ and $S_{\rm out}=S\setminus \overline{D}_{(\ell-1)(h+1)μ+1}$. 
Let also $k_{\rm in}:=|S_{\rm in}|$ and $k_{\rm out}:=|S_{\rm out}|$
and keep in mind that $k_{\rm in}+k_{\rm out}=|S|\leq k$.

Let $\bar{\bf J}_{S}$ be the meta-representative of $(\bar{\bf G}_\ell, {\cal Z}_\ell, {\cal V}_\ell)\setminus S_{\rm in}$.
Let also $S_{\bar{\bf J}_{S}}$ be the minimum-size subset of $R\cap D_{\ell(h+1)μ}$
such that $(\bar{\bf G}_\ell, {\cal Z}_\ell, {\cal V}_\ell)\setminus S_{\bar{\bf J}_{S}}\equiv_{h+\lambda}^{(h+1)}  \bar{\bf J}_{S}$.
%
%
Clearly, $|S_{\bar{\bf J}_{S}}|\leq |S_{\rm in}|=k_{\rm in}$ and therefore $S_{\bar{\bf J}_{S}}=S_{\ell,\bar{\bf J}_{S}}$.
We now set $S'=S_{\bar{\bf J}_{S}}\cup S_{\rm out}$
and observe that,
since $S_{\bar{\bf J}_{S}}=S_{\ell,\bar{\bf J}_{S}}$,
$S_{\bar{\bf J}_{S}}\cap D_{r}\subseteq R^*$ and therefore $S'\subseteq R'$.
Since $|S'|\leq k$, it remains to prove that 
$H\npreceq G\setminus S'$.

\medskip

Let ${\cal H}$ be the set of all topological minor models of $H$ in $G$ and notice that 
for every $(M,T)\in{\cal H}$ it holds that $S\cap V(M)\neq\emptyset$, i.e., $S$ intersects at least one 
vertex of each graph in ${\cal H}$. Let ${\cal H}_{\ell}$ be the members of ${\cal H}$
that are intersected {\em only} by vertices in $S_{\rm in}$.

The next claim shows that
there is a collection of cycles of ${\cal A}$, {the ``middle'' cycles of $Q_{\ell, j}$'s,} such that 
for every {\sf tm}-pair $({M},{T})\in {\cal H}_{\ell}$
there is a 
cycle $C$ of this collection
and a {\sf tm}-pair $(\tilde{M},\tilde{T})\in {\cal H}_{\ell}$
that is equivalent to $({M},{T})$ and is ``combed in $C$'' in the sense that $\tilde{M}\cap C$ is a subgraph 
of the rails of ${\cal A}$.\medskip

\noindent{\em Claim: } For every ${(M,T)}\in {\cal H}_{\ell}$, there exists a $j_{M}\in[h+1]$ and a {\sf tm}-pair $(\tilde{M},\tilde{T})\in {\cal H}_{\ell}$, 
such that $\tilde{M}\setminus {A}_{\ell}\subseteq M\setminus {A}_{\ell}$ and
the graph $\tilde{M}\cap {C}_{y_{M}}$ is the union of the paths $\{P_{y_{M},c^M_{1}},\ldots,P_{y_{M},c^M_{z_{M}}}\}$, where $y_{M}=(\ell-1)(h+1)μ+(j_M-1)μ+\lceil μ/2\rceil$ and
$\{c^M_{1},\ldots,c^M_{z_{M}}\}\subseteq [\lambda]$ (see \autoref{sdfadfgdfgfdgdfhgfdhfgdh}).

\begin{figure}[ht]
	
	\centering\scalebox{1}{
	\sshow{0}{\begin{tikzpicture}[scale=0.3]
	\begin{scope}[xshift = 5 cm]
	\foreach \i in {0,0.5,3, 5.5, 6}{
		\draw[-] (\i,0) -- (\i, 10);
	}
	
	\draw[opacity=0.3] (1,0) -- (1,10);
	\draw[opacity=0.1] (1.5,0) -- (1.5,10);
	
	\draw[opacity=0.2] (2.5,0) -- (2.5,10);
	\draw[opacity=0.2] (3.5,0) -- (3.5,10);
	
	\draw[opacity=0.1] (4.5,0) -- (4.5,10);
	\draw[opacity=0.3] (5,0) -- (5,10);

	\draw[-, line width=0.7pt] (0,8) -- (0.5,8) -- (0.5,7);
	\draw[black!30!white, line width=0.7pt] (0.5,7) -- (1,7);
	\draw[black!10!white, line width=0.7pt] (1,7) -- (1.5,8);
	\draw[-, line width=0.7pt] (0,6) -- (0,5) -- (0.5,5.5);
	\draw[black!30!white, line width=0.7pt] (0.5,5.5) -- (1,6);
	\draw[black!10!white, line width=0.7pt] (1,6) -- (1.5,6);
	\draw[-, line width=0.7pt] (0,3) -- (0.5,4);
	\draw[black!30!white, line width=0.7pt] (0.5,4) -- (1,4) -- (1,3);
	\draw[black!10!white, line width=0.7pt] (1,3) -- (1.5,2.5);
	
	\node[track node 1] () at (0,8) {};
	\node[track node 1] () at (0,6) {};
	\node[track node 1] () at (0,3) {};
	\node[track node 1] () at (0,5) {};
	\node[track node 1] () at (0.5,8) {};
	\node[track node 1] () at (0.5,7) {};
	\node[track node 1] () at (0.5,5.5) {};
	\node[track node 1] () at (0.5, 4) {};
	
	\node[track node 2] () at (1,7) {};
	\node[track node 2] () at (1,6) {};
	\node[track node 2] () at (1,4) {};
	\node[track node 2] () at (1,3) {};
	
	\node[track node 3] () at (1.5,8) {};
	\node[track node 3] () at (1.5,6) {};
	\node[track node 3] () at (1.5,2.5) {};

	\draw[black!30!white, line width=0.7pt] (2.5,8) -- (3,9) -- (3.5,9);
	\draw[black!30!white, line width=0.7pt] (2.5,6) -- (2.5,5) -- (3,6) -- (3.5,5.5);
	\draw[black!30!white, line width=0.7pt] (2.5,2.5) -- (3,3);
	\draw[-, line width=0.7pt] (3,3) -- (3,2);
	\draw[black!30!white, line width=0.7pt] (3,2) -- (3.5,2);
	
	\node[track node 1] () at (3,9) {};
	\node[track node 1] () at (3,6)  {};
	\node[track node 1] () at (3,3) {};
	\node[track node 1] () at (3,2) {};
	
	\node[track node 2] () at (2.5,8) {};
	\node[track node 2] () at (3.5,9) {};
	\node[track node 2] () at (2.5,6) {};
	\node[track node 2] () at (2.5,5) {};
	\node[track node 2] () at  (3.5,5.5) {};
	\node[track node 2] () at (2.5,2.5) {};
	\node[track node 2] () at  (3.5,2) {};

	\draw[-, line width=0.7pt] (5.5, 7) -- (6,8);
	\draw[black!30!white, line width=0.7pt] (5,8) -- (5,7) -- (5.5, 7);
	\draw[black!10!white, line width=0.7pt] (4.5,9) -- (5,8);
	\draw[-, line width=0.7pt]  (5.5,6) -- (5.5,5) -- (6,5);
	\draw[black!30!white, line width=0.7pt](5,6) -- (5.5,6);
	\draw[black!10!white, line width=0.7pt](4.5,5.5) -- (5,6);
	\draw[-, line width=0.7pt] (5.5,1) -- (6,1) -- (6,0.5);
	\draw[black!30!white, line width=0.7pt] (5,2) -- (5.5,1);
	\draw[black!10!white, line width=0.7pt] (4.5,2) -- (5,2);
	
	\node[track node 1] () at (6,8) {};
	\node[track node 1] () at (5.5, 7) {};
	\node[track node 1] () at (5.5,6) {};
	\node[track node 1] () at (5.5,5) {};
	\node[track node 1] () at (6,5) {};
	\node[track node 1] () at (5.5,1) {};
	\node[track node 1] () at (6,1) {};
	\node[track node 1] () at (6,0.5) {};
	
	\node[track node 2] () at (5,8) {};
	\node[track node 2] () at (5,7) {};
	\node[track node 2] () at (5,6) {};
	\node[track node 2] () at (5,2) {};
	
	\node[track node 3] () at (4.5,9) {};
	\node[track node 3] () at (4.5,5.5) {};
	\node[track node 3] () at (4.5,2) {};
	
	\end{scope}
	
	\foreach \i in {0, 12}{
		\begin{scope}[xshift=\i cm]
		
		\draw[-] (0.5,0) -- (0.5, 10);
		\draw[opacity=0.3] (1,0) -- (1,10);
		\draw[opacity=0.1] (1.5,0) -- (1.5,10);
		\draw[opacity=0.1] (2.5,0) -- (2.5,10);
		\draw[opacity=0.3] (3,0) -- (3,10);
		\draw[-] (3.5,0) -- (3.5, 10);
		\end{scope}
	}
\draw[opacity=0.1] (4,0) -- (4,10);
\draw[opacity=0.1] (12,0) -- (12,10);


\draw [decorate,decoration={brace,amplitude=10pt, mirror}] (0.5,-2.5) -- (15.5,-2.5) node [black,midway,yshift=-.6cm] {$A_{\ell}$};
\draw [decorate,decoration={brace,amplitude=5pt, mirror}] (0.5,-.3) -- (3.5,-.3) node [black,midway,yshift=-.4cm] {$Q_{\ell,1}$};
\draw [decorate,decoration={brace,amplitude=5pt, mirror}] (5,-.3) -- (11,-.5) node [black,midway,yshift=-.4cm] {$Q_{\ell,j_{M}}$};
\draw [decorate,decoration={brace,amplitude=5pt, mirror}] (12.5,-.3) -- (15.5,-.3) node [black,midway,yshift=-.4cm] {$Q_{\ell,h+1}$};
	
	\node[model node] (M1) at (3,7) {};
	\node[model node] (M2) at (-1,4) {};
	\node[model node] (M3) at (1,3) {};
	\node[model node] (M4) at (12,7) {};
	\node[model node] (M5) at (15,5) {};
	\node[model node] (M6) at (17,2) {};
	\node[model node] (M7) at (13,1) {};
	
	\draw[celestialblue, line width=1pt] plot [smooth, tension=1] coordinates {(M1) (0,6) (-1,4) (M2)};
	\draw[celestialblue, line width=1pt] plot [smooth, tension=1] coordinates {(M2) (-1,3) (M3)};
	\draw[celestialblue, line width=1pt] plot [smooth, tension=1] coordinates {(M1) (2,6) (2.5,5)  (M3)};
	\draw[celestialblue, line width=1pt] plot [smooth, tension=1] coordinates {(M1) (5,8) (7,6) (10,8) (M4)};
	\draw[celestialblue, line width=1pt] plot [smooth, tension=1] coordinates {(M1) (4,4) (10,4) (7,1) (M7)};
	\draw[celestialblue, line width=1pt] plot [smooth, tension=1] coordinates {(M4) (13,5) (M7)};
	\draw[celestialblue, line width=1pt] plot [smooth, tension=1] coordinates {(M4) (15,7) (M5)};
	\draw[celestialblue, line width=1pt] plot [smooth, tension=1] coordinates {(M7) (14,2) (M5)};
	\draw[celestialblue, line width=1pt] plot [smooth, tension=1] coordinates {(M5) (16,4) (M6)};
	\draw[celestialblue, line width=1pt] plot [smooth, tension=1] coordinates {(M7) (15,1.5) (16,1) (M6)};
	\begin{scope}[xshift=25cm]

	\begin{scope}[xshift = 5 cm]
	\foreach \i in {0,0.5,3, 5.5, 6}{
		\draw[-] (\i,0) -- (\i, 10);
	}
	
	\draw[opacity=0.3] (1,0) -- (1,10);
	\draw[opacity=0.1] (1.5,0) -- (1.5,10);
	
	\draw[opacity=0.2] (2.5,0) -- (2.5,10);
	\draw[opacity=0.2] (3.5,0) -- (3.5,10);
	
	\draw[opacity=0.1] (4.5,0) -- (4.5,10);
	\draw[opacity=0.3] (5,0) -- (5,10);

	\draw[-, line width=0.7pt] (0,8) -- (0.5,8) -- (0.5,7);
	\draw[black!30!white, line width=0.7pt] (0.5,7) -- (1,7);
	\draw[black!10!white, line width=0.7pt] (1,7) -- (1.5,8);
	\draw[-, line width=0.7pt] (0,6) -- (0,5) -- (0.5,5.5);
	\draw[black!30!white, line width=0.7pt] (0.5,5.5) -- (1,6);
	\draw[black!10!white, line width=0.7pt] (1,6) -- (1.5,6);
	\draw[-, line width=0.7pt] (0,3) -- (0.5,4);
	\draw[black!30!white, line width=0.7pt] (0.5,4) -- (1,4) -- (1,3);
	\draw[black!10!white, line width=0.7pt] (1,3) -- (1.5,2.5);
	
	\node[track node 1] () at (0,8) {};
	\node[track node 1] () at (0,6) {};
	\node[track node 1] () at (0,3) {};
	\node[track node 1] () at (0,5) {};
	\node[track node 1] () at (0.5,8) {};
	\node[track node 1] () at (0.5,7) {};
	\node[track node 1] () at (0.5,5.5) {};
	\node[track node 1] () at (0.5, 4) {};
	
	\node[track node 2] () at (1,7) {};
	\node[track node 2] () at (1,6) {};
	\node[track node 2] () at (1,4) {};
	\node[track node 2] () at (1,3) {};
	
	\node[track node 3] () at (1.5,8) {};
	\node[track node 3] () at (1.5,6) {};
	\node[track node 3] () at (1.5,2.5) {};

	\draw[black!30!white, line width=0.7pt] (2.5,8) -- (3,9) -- (3.5,9);
	\draw[black!30!white, line width=0.7pt] (2.5,6) -- (2.5,5) -- (3,6) -- (3.5,5.5);
	\draw[black!30!white, line width=0.7pt] (2.5,2.5) -- (3,3);
	\draw[-, line width=0.7pt] (3,3) -- (3,2);
	\draw[black!30!white, line width=0.7pt] (3,2) -- (3.5,2);
	
	\node[track node 1] () at (3,9) {};
	\node[track node 1] () at (3,6)  {};
	\node[track node 1] () at (3,3) {};
	\node[track node 1] () at (3,2) {};
	
	\node[track node 2] () at (2.5,8) {};
	\node[track node 2] () at (3.5,9) {};
	\node[track node 2] () at (2.5,6) {};
	\node[track node 2] () at (2.5,5) {};
	\node[track node 2] () at  (3.5,5.5) {};
	\node[track node 2] () at (2.5,2.5) {};
	\node[track node 2] () at  (3.5,2) {};

	\draw[-, line width=0.7pt] (5.5, 7) -- (6,8);
	\draw[black!30!white, line width=0.7pt] (5,8) -- (5,7) -- (5.5, 7);
	\draw[black!10!white, line width=0.7pt] (4.5,9) -- (5,8);
	\draw[-, line width=0.7pt]  (5.5,6) -- (5.5,5) -- (6,5);
	\draw[black!30!white, line width=0.7pt](5,6) -- (5.5,6);
	\draw[black!10!white, line width=0.7pt](4.5,5.5) -- (5,6);
	\draw[-, line width=0.7pt] (5.5,1) -- (6,1) -- (6,0.5);
	\draw[black!30!white, line width=0.7pt] (5,2) -- (5.5,1);
	\draw[black!10!white, line width=0.7pt] (4.5,2) -- (5,2);
	
	\node[track node 1] () at (6,8) {};
	\node[track node 1] () at (5.5, 7) {};
	\node[track node 1] () at (5.5,6) {};
	\node[track node 1] () at (5.5,5) {};
	\node[track node 1] () at (6,5) {};
	\node[track node 1] () at (5.5,1) {};
	\node[track node 1] () at (6,1) {};
	\node[track node 1] () at (6,0.5) {};
	
	\node[track node 2] () at (5,8) {};
	\node[track node 2] () at (5,7) {};
	\node[track node 2] () at (5,6) {};
	\node[track node 2] () at (5,2) {};
	
	\node[track node 3] () at (4.5,9) {};
	\node[track node 3] () at (4.5,5.5) {};
	\node[track node 3] () at (4.5,2) {};
	
	\end{scope}
	
	\foreach \i in {0, 12}{
		\begin{scope}[xshift=\i cm]
		
		\draw[-] (0.5,0) -- (0.5, 10);
		\draw[opacity=0.3] (1,0) -- (1,10);
		\draw[opacity=0.1] (1.5,0) -- (1.5,10);
		\draw[opacity=0.1] (2.5,0) -- (2.5,10);
		\draw[opacity=0.3] (3,0) -- (3,10);
		\draw[-] (3.5,0) -- (3.5, 10);
		\end{scope}
	}
\draw[opacity=0.1] (4,0) -- (4,10);
\draw[opacity=0.1] (12,0) -- (12,10);


	\draw [decorate,decoration={brace,amplitude=10pt, mirror}] (0.5,-2.5) -- (15.5,-2.5) node [black,midway,yshift=-.6cm] {$A_{\ell}$};
	\draw [decorate,decoration={brace,amplitude=5pt, mirror}] (0.5,-.3) -- (3.5,-.3) node [black,midway,yshift=-.4cm] {$Q_{\ell,1}$};
	\draw [decorate,decoration={brace,amplitude=5pt, mirror}] (5,-.3) -- (11,-.5) node [black,midway,yshift=-.4cm] {$Q_{\ell,j_{M}}$};
	\draw [decorate,decoration={brace,amplitude=5pt, mirror}] (12.5,-.3) -- (15.5,-.3) node [black,midway,yshift=-.4cm] {$Q_{\ell,h+1}$};

	\node[model node] (M1) at (3,7) {};
	\node[model node] (M2) at (-1,4) {};
	\node[model node] (M3) at (1,3) {};
	\node[model node] (M4) at (12,7) {};
	\node[model node] (M5) at (15,5) {};
	\node[model node] (M6) at (17,2) {};
	\node[model node] (M7) at (13,1) {};
	
	\draw[celestialblue, line width=1pt] plot [smooth, tension=1] coordinates {(M1) (0,6) (-1,4) (M2)};
	\draw[celestialblue, line width=1pt] plot [smooth, tension=1] coordinates {(M2) (-1,3) (M3)};
	\draw[celestialblue, line width=1pt] plot [smooth, tension=1] coordinates {(M1) (2,6) (2.5,5)  (M3)};
	\draw[celestialblue, line width=1pt] plot [smooth, tension=1] coordinates {(M1) (5,8) (8,6) (10,8) (M4)};
	\draw[celestialblue, line width=1pt] plot [smooth, tension=1] coordinates {(M1) (3,2) (8,3)};
	\draw[celestialblue, line width=1pt] (8,2) -- (8,3);
	\draw[celestialblue, line width=1pt] plot [smooth, tension=1] coordinates {(8,2) (10,1) (M7)};
	\draw[celestialblue, line width=1pt] plot [smooth, tension=1] coordinates {(M4) (13,5) (M7)};
	\draw[celestialblue, line width=1pt] plot [smooth, tension=1] coordinates {(M4) (15,7) (M5)};
	\draw[celestialblue, line width=1pt] plot [smooth, tension=1] coordinates {(M7) (14,2) (M5)};
	\draw[celestialblue, line width=1pt] plot [smooth, tension=1] coordinates {(M5) (16,4) (M6)};
		\draw[celestialblue, line width=1pt] plot [smooth, tension=1] coordinates {(M7) (15,1.5) (16,1) (M6)};

\end{scope}
	\end{tikzpicture}}}
	\caption{Visualization of the statement of the Claim. $(M,T)$ is depicted in the left figure, while $(\tilde{M},\tilde{T})$ is depicted in the right figure.}
	\label{sdfadfgdfgfdgdfhgfdhfgdh}
\end{figure}

\noindent{\em Proof of Claim: }
Let ${(M,T)}\in {\cal H}_{\ell}$ and notice that $S_{\rm in}\cap V({M})\neq \emptyset$. As $|T|\leq h$, there is some $j_{M}\in[h+1]$
such that $T\cap Q_{\ell, j_{M}}=\emptyset$ (if many such $j_{M}$'s exist, take the minimum one).
We use notation ${\cal A}^{(M)}=({\cal C}^{(M)},{\cal P}^{(M)})$ instead of ${\cal A}\cap Q_{\ell, j_{M}}$.
Since $|{\cal C}^{(M)}|=\mu= \funref{axfsfadsrdsfsd}(g)+3$ and $|{\cal P}^{(M)}|=q\geq \funref{dfljhklgdfjhklgfj}(h)=  5/2\cdot \funref{axfsfsd}(g)$, we can now apply \autoref{jklnlnlk} for 
$s=1$, ${\cal A}:={\cal A}^{(M)}$, and $I=[\lambda]$ 
and obtain a  topological minor model $(\tilde{M}, \tilde{T})$ of $H$ in $G$
such that $\tilde{T}=T$, $\tilde{M}$ is $(1,I)$-confined in ${\cal A}^{(M)}$  and $\tilde{M}\setminus Q_{\ell, j_{M}} \subseteq M\setminus Q_{\ell, j_{M}}$, which implies that $\tilde{M}\setminus {A}_{\ell}\subseteq M\setminus {A}_{\ell}$.
Let $y_{M}=(\ell-1)(h+1)μ+(i_M-1)μ+\lceil μ/2\rceil$.
Notice  that $(\tilde{M}, \tilde{T})$ is a topological minor model in ${\cal H}_{\ell}$ 
whose intersection with $C_{y_{M}}$ is the union of some of the paths in $\{P_{y_{M},1},\ldots,P_{y_{M},\lambda}\}$, namely $\{P_{y_{M},c^M_{1}},\ldots,P_{y_{M},c^M_{z_{M}}}\}$, where $\{c^M_{1},\ldots,c^M_{z_{M}}\}\subseteq [\lambda]$.
The claim follows.
\medskip

Suppose, towards a contradiction, that the graph $G\setminus S'$ contains some topological minor model  $(M,T)$  of $H$ as a subgraph.
Since $H\npreceq G\setminus S$, it holds that $(M,T)$ is intersected only by vertices in $S_{\rm in}$ - thus $(M,T)\in {\cal H}_{\ell}$.
According to the Claim above, there is an $j_{M}\in[h+1]$ and a topological minor model $(\tilde{M},\tilde{T})\in {\cal H}_{\ell}$ such that  $\tilde{M}\setminus {A}_{\ell}\subseteq M\setminus {A}_{\ell}$ and the graph $\tilde{M}\cap C_{y_{M}}$ is the union of the paths $\{P_{y_{M},c^M_{1}},\ldots,P_{y_{M},c^M_{z_{M}}}\}$ where $y_{M}=(\ell-1)(h+1)μ+(j_M-1)μ+\lceil μ/2\rceil+1$ and
 $\{c^M_{1},\ldots,c^M_{z_{M}}\}\subseteq [\lambda]$. Note that $B_{\ell, j_M}\subseteq V(C_{j_M})$.
 Moreover, since $(M,T)$ is a topological minor model of $H$ in $G\setminus S'$ and $\tilde{M}\setminus {A}_{\ell}\subseteq M\setminus {A}_{\ell}$, we have that $(\tilde{M}, \tilde{T})$ is a topological minor model of $H$ in $G\setminus S'$.

We  consider the $\lambda$-boundaried graph $\tilde{\bf M}_{\rm in}=(\tilde{M}_{\rm in},B_{\ell,j_M},ρ_{\ell,j_M})$ where $\tilde{M}_{\rm in}=(\tilde{M}\cap \overline{D}_{y_M})\cup (B_{\ell,j_M}, \emptyset)$, i.e., the graph $\tilde{M}\cap \overline{D}_{y_M}$ together with the isolated vertices $B_{\ell, j_M}$.
We also define
{$$\tilde{M}_{\rm out}=(\tilde{M}\setminus (\overline{D}_{y_M}\setminus B_{\ell, j_M}))\cup (B_{\ell, j_M}, \emptyset)\mbox{~and~} \tilde{\bf M}_{\rm out}=(\tilde{M}_{\rm out},B_{\ell, j_M},ρ_{\ell, j_M}).$$}
Intuitively, $\tilde{M}_{\rm out}$ is the graph obtained from $\tilde{M}$  by removing all vertices in $\overline{D}_{y_M}$ except the vertices in $B_{\ell,j_M}$ and adding the isolated vertices $B_{\ell,j_M}\setminus V(\tilde{M})$.

We now observe that, since $(\tilde{M}, \tilde{T})$ is a topological minor model of $H$ in $G\setminus S'$ and $S_{\bar{\bf J}}$ is a subset of $S'$, $\tilde{\bf M}_{\rm in}$ is a subgraph of $(G_{\ell, j_M}, B_{\ell, j_M}, \rho_{\ell, j_M})\setminus S_{\bar{\bf J}}$.
The latter, together with the fact that $\tilde{\bf M}_{\rm out}$ is compatible with $(G_{\ell, j_M}, B_{\ell, j_M}, \rho_{\ell, j_M})\setminus S_{\bar{\bf J}}$, implies that 
\begin{eqnarray}\label{eqdgj}
H\preceq \tilde{\bf M}_{\rm out} \oplus (G_{\ell, j_M}, B_{\ell, j_M}, \rho_{\ell, j_M})\setminus S_{\bar{\bf J}}.
\end{eqnarray}
Also, the fact that $(\bar{\bf G}_\ell, {\cal Z}_\ell, {\cal V}_\ell)\setminus S_{\bar{\bf J}_{S}}\equiv_{h+\lambda}^{(h+1)} (\bar{\bf G}_\ell, {\cal Z}_\ell, {\cal V}_\ell)\setminus S_{\rm in}$ implies that 
\begin{eqnarray}\label{ddflkgmdfl}
(G_{\ell, j_M}, B_{\ell, j_M}, \rho_{\ell, j_M})\setminus S_{\rm in}\equiv_{h+\lambda}(G_{\ell, j_M}, B_{\ell, j_M}, \rho_{\ell, j_M})\setminus S_{\bar{\bf J}}.
\end{eqnarray}
By (\ref{eqdgj}) and (\ref{ddflkgmdfl}), we obtain that $H\preceq \tilde{\bf M}_{\rm out} \oplus (G_{\ell, j_M}, B_{\ell, j_M}, \rho_{\ell, j_M})\setminus S_{\rm in}$, which, in turn, implies that $H\preceq G\setminus S$, a contradiction.
\end{proof}

\subsection{Finding an irrelevant area}
\label{gdgafssdfgsfhdsdhgshghsfgssghhfsgj}

The next lemma intuitively states that there exists an algorithm that given a \seg\ $G$ and a ``big enough'' railed annulus of $G$, outputs a ``big enough'' wall $W$ of $G$ whose compass is a subset of $\Delta$ and such that for every hitting set $S$ outside $\Delta$, the vertex set of the compass of $W$ in $G$ is an irrelevant part of the instance.

\begin{lemma}
\label{fsfsdfdsdsafsadffasdasfd1}
There exist two functions $\newfun{fnkdslnfgkldsang},\newfun{dsgkdgnlkdfsngkdfsnl}: \Bbb{N}^2\to \Bbb{N}$,
and an algorithm with the following specifications:
	
	\smallskip\noindent {\bf Find\_irrelevant\_area}$(b,h,w,{\cal F},\Delta,G,{R},{\cal C},{\cal P})$

\noindent{\sl Input:} three integers $b\in\Bbb{N}_{\geq 3}$ and $h,w\in\Bbb{N}$, a finite set of graphs ${\cal F}$ where $h\leq h({\cal F})$,
a partially $\Delta$-embedded graph $G$  whose compass has treewidth at most $w$,
a set $R\subseteq V(G)\setminus \Delta$,
and an  $(\funref{fnkdslnfgkldsang}(h,b),\funref{dsgkdgnlkdfsngkdfsnl}(h,b))$-railed annulus ${\cal A}=({\cal C},{\cal P})$ of $G$.
		
\noindent{\sl Output:} a $b$-wall $W$ of $G$ such that 
\begin{itemize}

\item $V({\sf compass}(W))\subseteq \Delta$ and

\item if $(G\setminus V({\sf compass}(W)),R,k)$ is a ${\bf tm}_{\cal F}$-triple then $(G, R, k)$ is a ${\bf tm}_{\cal F}$-triple.
\end{itemize} 
Moreover,
$\funref{fnkdslnfgkldsang}(h,b)=\mathcal{O}_{h} (b)$,
$\funref{dsgkdgnlkdfsngkdfsnl}(h,b)=\mathcal{O}_h (b)$,
and this algorithm runs in $2^{\mathcal{O}_{h} (w\log w)}\cdot b\cdot n$ time, or, alternatively, in $ \mathcal{O}_h (b\cdot n^3)+ 2^{\mathcal{O}_{h}(w)}\cdot b\cdot n$ time.
\end{lemma}

\begin{proof}
Let $g:= \genfrac{(}{)}{0pt}{}{h}{2}$, $\lambda:=\funref{axfsfsd}(g) + 1$, and $μ:=\funref{axfsfadsrdsfsd}(g)+3$.
We set
\begin{align*}
	\ell:= & (h+2)μ+b+1, \\
	r := & \ \funref{dskglasdgklnfdkn}(h+\lambda,1,\lambda) \cdot \ell, \\	
	q:= & \max\{5/2\cdot \funref{axfsfsd}(g),\lambda+2b\},\\
	\funref{fnkdslnfgkldsang}(h,b): = &\ r, \text{ and }\\
	\funref{dsgkdgnlkdfsngkdfsnl}(h,b):= & q.
\end{align*}
For every $i\in[r]$, we consider the $\lambda$-boundaried graph ${\bf G}_{i,\lambda}=(G_{i},B_{i,\lambda},ρ_{i,\lambda})$
where $G_{i}=G\cap \overline{D}_{i}$, $B_{i,\lambda}=\{r_{i,1}\,\ldots,r_{i,\lambda}\}$ and, for $j\in[\lambda]$, $ρ_{i,\lambda}(r_{i,\lambda})=j$.
Also, for every $i\in[r]$, let ${\bf J}_i$ be the representative of ${\bf G}_{i,\lambda}$.
To compute ${\bf J}_i$, we first observe that as the underlying graph of each ${\bf G}_{i,\lambda}$ is a subgraph of the compass of $G$, we 
	have that $\tw({\bf G}_{i,\lambda})\leq w+\lambda=\mathcal{O}_{h}(w)$.
	Therefore, for each representative ${\bf J}\in {\cal R}_{h+\lambda}^{(\lambda)}$, we call the algorithm  {\bf Compute\_rep}$(h+\lambda,1,\lambda,w+\lambda,0,{\bf G}_{i,\lambda}, \emptyset, {\bf J})$
	of \autoref{sssswfgfegwergewgwergwegr} to check whether ${\bf G}_{i,\lambda}\equiv_{h+\lambda} {\bf J}$.	
	The overall running time needed to compute the representative of each ${\bf G}_{i,\lambda}$ is $r\cdot |{\cal R}_{h+\lambda}^{(\lambda)}| \cdot 2^{\mathcal{O}_{h} (w\log w)}\cdot n =  2^{\mathcal{O}_{h} (w\log w)}\cdot b\cdot n$, or, alternatively, $r\cdot \mathcal{O}(n^3)+r\cdot |{\cal R}_{h+\lambda}^{(\lambda)}| \cdot 2^{\mathcal{O}_{h}(w)}\cdot n= \mathcal{O}_h (b\cdot n^3)+ 2^{\mathcal{O}_{h}(w)}\cdot b\cdot n$.	
	
	Since $|{\cal R}_{h+\lambda}^{(\lambda)}|\leq \funref{dskglasdgklnfdkn}(h+\lambda,1,\lambda)$, then there is a set $I\subseteq [r]$ of size $\ell$ such that for every $p,q\in I$, ${\bf J}_{p}={\bf J}_{q}$.
Let $i'$ be the maximum element of $I$ and note that $i'\geq (h+2)μ+b+1$.
We define $\Delta'$ to be the arc-wise connected component of $\ann({\cal C}, i'-b, i'-1)\setminus (P_{\lambda+1}\cup P_{\lambda +2b})$
that does not intersect $P_1$ (see \autoref{sadfdgfsdgdsgdfg}).

Notice that the graph obtained by the union of $\bigcup_{j\in[b]}C_{i'-j}$ and $\bigcup_{j\in[2b]} P_{\lambda+j}$ contains a $b$-wall $W$ as a subgraph such that $V({\sf compass}(W))$ is a subset of the closure of $\Delta'$.
 We set $K:={\sf compass}(W)$ and keep in mind that $V(K)$ (and, therefore, also $\Delta'$) is a subset of $\Delta$ that does not intersect the cycle $C_{i'}$. 
 
\begin{figure}[ht]
		\centering\scalebox{1}{
		\sshow{0}{\begin{tikzpicture}[scale=0.4]
		\begin{scope}
		
		\draw[-] (7,2) -- (7, 10);
		\draw[opacity=0.3] (7.5,2) -- (7.5,10);
		\draw[opacity=0.1] (8,2) -- (8,10);

		\draw[opacity=0.1] (29,2) -- (29,10);
		\draw[opacity=0.3] (29.5,2) -- (29.5,10);
		\draw[-] (30,2) -- (30, 10);
		\end{scope}

		\foreach \i in {11, 17, 25 }{
			\begin{scope}[xshift=\i cm]

			\draw[opacity=0.1] (-1,2) -- (-1,10);
			\draw[opacity=0.3] (-.5,2) -- (-.5,10);
			\draw[-] (0,2) -- (0, 10);
			\draw[opacity=0.3] (0.5,2) -- (0.5,10);
			\draw[opacity=0.1] (1,2) -- (1,10);
			
			\end{scope}
		}

		\node () at (7,11) {$C_{1}$};
		\node () at (11,11) {};
		\node[] () at (17,11) {$C_{i'-b}$};
		\draw[red] (17,10.6) -- (17,10.2);
		\node[small red] () at (17,10.2) {};
		\node[] () at (25,11) {$C_{i'}$};
		\draw[red] (25,10.6) -- (25,10.2);
		\node[small red] () at (25,10.2) {};
		\node () at (30,11) {$C_{r}$};
		
		\draw [decorate,decoration={brace,amplitude=5pt, mirror}] (17,1.5) -- (24.5,1.5) node [black,midway,yshift=-.3cm] {\footnotesize $b$};
		
		\filldraw[fill=applegreen!60!white] (17,9) --  (24.5,9)  -- (24.5,5) -- (17,5) -- cycle;
		\node () at (20,7) {$\Delta'$};
		
		\draw[black!30!white, line width=0.7pt] (16.5, 8.5) -- (17,9) (17,9) -- (24.5,9) -- (24.5,8) -- (25,8) -- (25.5,8) -- (25.5,9) (16.5, 5.5) -- (17,6) -- (17,5) -- (24.5,5) -- (25,6) -- (25.5,6.5);
		\draw[black!10!white, line width=0.7pt] (16,9) -- (16.5, 8.5) (25.5,9) -- (26,9) (16,5) -- (16.5, 5.5) (25.5,6.5) -- (26,6);
		\node () at (27.5,9) {$P_{\lambda+2b}$};
		\node () at (27.5,6) {$P_{\lambda+1}$};
		
		\node[track node 1] () at (17,9) {};
		\node[track node 1] () at (25,8)  {};
		\node[track node 1] () at (17,6) {};
		\node[track node 1] () at  (25,6) {};
		\node[track node 1] () at  (17,5) {};

		\node[track node 2] () at (16.5, 8.5) {};
		\node[track node 2] () at  (16.5, 5.5) {};
		\node[track node 2] () at  (25.5,8) {};
		\node[track node 2] () at (25.5,9){};
		\node[track node 2] () at (25.5,6.5) {};
		\node[track node 2] () at (24.5,9) {};
		\node[track node 2] () at  (24.5,8) {};
		\node[track node 2] () at  (24.5,5) {};
		
		\node[track node 3] () at (16,9) {};
		\node[track node 3] () at (26,9) {};
		\node[track node 3] () at (16,5)  {};
		\node[track node 3] () at (26,6) {};
		
		\end{tikzpicture}}}
		\caption{An example showing the disk $\Delta'$, whose closure contains the vertices of the compass of the obtained $b$-wall $W$.}
		\label{sadfdgfsdgdsgdfg}
	\end{figure}
	
We now aim to prove that if $(G\setminus V(K),R,k)$ is a ${\bf tm}_{\cal F}$-triple then $(G, R, k)$ is a ${\bf tm}_{\cal F}$-triple.
In order to prove this, we argue that if $H$ is a graph on at most $h$ vertices and
$S$ is a subset of $R$ such that $H\npreceq(G\setminus V(K))\setminus S$, then it holds that  $H\npreceq G\setminus S$.

Let  $H$ be a graph on at most $h$ vertices (and, therefore, of at most $g$ edges) and  $S\subseteq R \subseteq V(G)\setminus \Delta$ such that  $H\npreceq(G\setminus V(K))\setminus S$.
Suppose towards a contradiction that the graph $G\setminus S$ contains some topological minor model  $(M,T)$  of $H$ as a subgraph. In what follows, we argue how to obtain a subgraph of $(G\setminus V(K))\setminus S$ that is a subdivision of $H$, thus arriving at a contradiction.

As $|T|\leq h$ and $\ell=(h+2)μ +b+1$, there is some
	$y\in I\setminus \{i'\}$  such that $$T\cap \ann({\cal C},y-\lfloor μ/2\rfloor,y+\lfloor μ/2\rfloor)=\emptyset.$$
We  consider the  $(μ,q)$-railed 
annulus ${\cal A}'=({\cal C}',{\cal P}')$ of $G$ where 
	\begin{itemize}
		\item ${\cal C}'=[C_{1}',\ldots,C_{μ}']:=[C_{y-\lfloor μ/2\rfloor},\ldots,C_{y+\lfloor μ/2\rfloor}]$ and
		
		\item ${\cal P}'=[P_{1}',\ldots,P_{q}']:=[P_{1}\cap\ann({\cal C}'),\ldots,P_{q}\cap\ann({\cal C}')].$ (See \autoref{sadfdfsgffdgshdfsagdadggsfdn}.)
	\end{itemize}

\begin{figure}[H]
	\centering
	\sshow{0}{\begin{tikzpicture}[scale=0.3]
	\begin{scope}
	
	\draw[-] (5,0) -- (5, 10);
	\draw[opacity=0.3] (5.5,0) -- (5.5,10);
	\draw[opacity=0.1] (5,0) -- (5,10);

	\draw[opacity=0.1] (29,0) -- (29,10);
	\draw[opacity=0.3] (29.5,0) -- (29.5,10);
	\draw[-] (30,0) -- (30, 10);
	\end{scope}
	
	\begin{scope}[xshift=12cm]
	\fill[mustard] (0,10) -- (4,10) -- (4,0) -- (0,0) -- cycle;
	\node () at (2,8) {${\cal A}'$};
	\end{scope}

	\foreach \i in {8,12,16, 25 }{
		\begin{scope}[xshift=\i cm]
		\draw[opacity=0.1] (-1,0) -- (-1,10);
		\draw[opacity=0.3] (-.5,0) -- (-.5,10);
		\draw[-] (0,0) -- (0, 10);
		\draw[opacity=0.3] (0.5,0) -- (0.5,10);
		\draw[opacity=0.1] (1,0) -- (1,10);
		\end{scope}
	}

	\node () at (5,11) {$C_{1}$};
	\node () at (8,11) {};
	\node () at (11.5,12) {$C_{y-\lfloor \mu/2\rfloor}$};
	\draw[red] (12,11.1) -- (12,10.5);
	\node[small red] () at (12,10.5) {};
	\node () at (16.5,12) {$C_{y+\lfloor \mu/2\rfloor}$};
	\draw[red] (16,11.1) -- (16,10.5);
	\node[small red] () at (16,10.5) {};
	\node[] () at (25,12) {$C_{i'}$};
	\draw[red] (25,11.3) -- (25,10.5);
	\node[small red] () at (25,10.5) {};
	\node () at (30,11) {$C_{r}$};
	
	\draw [decorate,decoration={brace,amplitude=5pt, mirror}] (12,-.5) -- (16,-.5) node [black,midway,yshift=-.4cm] {\footnotesize $μ$};

	\filldraw[draw=applegreen,fill=applegreen!60!white] (22.5,9) --  (24.5,9)  -- (24.5,7) -- (22.5,7) -- cycle;
	\node () at (23.2,8.3) {$\Delta'$};

	\begin{scope}[xshift=5cm]
	
	\node[model node] (M1) at (5.2,8) {};
	\node[model node] (M2) at (2.7,5) {};
	\node[model node] (M3) at (6.3,4) {};
	\node[model node] (M4) at (15,7) {};
	\node[model node] (M5) at (19,7.5) {};
	\node[model node] (M6) at (13,4) {};
	\node[model node] (M7) at (18,2.5) {};
	
	\draw[celestialblue] (M1) -- (M2) (M2) -- (M3) (M1) -- (M3) (M3) -- (M4) (M3) -- (M6) (M4) -- (M5) (M6) -- (M7) (M4) -- (M6) (M5)-- (M7) (M5) -- (M6);
	
	\end{scope}
	\end{tikzpicture}}\caption{An example showing the $(μ, q)$-railed annulus ${\cal A}'$.}
	\label{sadfdfsgffdgshdfsagdadggsfdn}
\end{figure}

Observe that, since $S\subseteq V(G)\setminus \Delta$, ${\cal A}'$ remains intact after removing the vertices of $S$ from $G$, i.e., ${\cal A}'$ is also an $(μ, q)$-railed annulus of $G\setminus S$.
Since $μ= \funref{axfsfadsrdsfsd}(g)+3$ and $q\geq 5/2 \cdot \funref{axfsfsd}(g)$, we are in position to apply \autoref{jklnlnlk} for $s:=1$, $H$, $G:=G\setminus S$, ${\cal A}:={\cal A}'$, $r:=μ$, $M$,  and $I=[\lambda]$.
We deduce the existence of 
a topological minor model $(\tilde{M}, \tilde{T})$ of $H$ in $G\setminus S$ such that
$\tilde{T} =T$,  $\tilde{M}$ is $(1,I)$-confined in ${\cal A}'$, and $\tilde{M}\setminus \ann({\cal A}')\subseteq M\setminus \ann({\cal A}')$ (see \autoref{snfcajksfn}).

\begin{figure}[ht]
	\centering\scalebox{1}{
	\sshow{0}{\begin{tikzpicture}[scale=0.38]
	
	\begin{scope}[xshift=10cm]
		\fill[mustard!50!white] (0,10) -- (6,10) -- (6,0) -- (0,0) -- cycle;		\end{scope}
	\node () at (12,-1) {${\cal A}'$};
			
	\foreach \i in {10,16, 25 }{
		\begin{scope}[xshift=\i cm]
			\draw[opacity=0.1] (-1,0) -- (-1,10);
			\draw[opacity=0.3] (-.5,0) -- (-.5,10);
			\draw[-] (0,0) -- (0, 10);
			\draw[opacity=0.3] (0.5,0) -- (0.5,10);
			\draw[opacity=0.1] (1,0) -- (1,10);
		\end{scope}
	}
		
	\node () at (10,12) {$C_{y-\lfloor \mu/2\rfloor}$};
	\draw[red] (10,11.3) -- (10,10.5);
	\node[small red] () at (10,10.5) {};
	\node () at (13,12) {$C_{y}$};
	\draw[red] (13,11.3) -- (13,10.5);
	\node[small red] () at (13,10.5) {};
	\node () at (16,12) {$C_{y+\lfloor \mu/2\rfloor}$};
	\draw[red] (16,11.3) -- (16,10.5);
	\node[small red] () at (16,10.5) {};
	\node[] () at (25,12) {$C_{i'}$};
	\draw[red] (25,11.3) -- (25,10.5);
	\node[small red] () at (25,10.5) {};
		
	\filldraw[draw=applegreen,fill=applegreen!60!white]
	(22.5,9) --  (24.5,9)  -- (24.5,7) -- (22.5,7) -- cycle;
	\node () at (23.2,8.3) {$\Delta'$};

	\begin{scope}[xshift=5cm]
		
		\node[model node] (M1) at (3.2,8) {};
		\node[model node] (M2) at (0.7,5) {};
		\node[model node] (M3) at (4.3,4) {};
		\node[model node] (M4) at (15,7) {};
		\node[model node] (M5) at (19,7.5) {};
		\node[model node] (M6) at (13,4) {};
		\node[model node] (M7) at (18,2.5) {};
		\draw[celestialblue]
		(M1) -- (M2) (M2) -- (M3) (M1) -- (M3)
		(M4) -- (M5) (M6) -- (M7) (M4) -- (M6)
		(M5)-- (M7) (M5) -- (M6);
		\draw[celestialblue] (M3) -- (8,6) --  (M4)
		(M3) -- (8,3) -- (8,2) -- (M6);
	
	\end{scope}
		
	\begin{scope}[xshift=10cm]
		
		\draw[opacity=0.3] (1,0) -- (1,10);
		\draw[opacity=0.1] (1.5,0) -- (1.5,10);
		
		\draw[opacity=0.2] (2.5,0) -- (2.5,10);
		\draw[opacity=0.2] (3.5,0) -- (3.5,10);
		
		\draw[opacity=0.1] (4.5,0) -- (4.5,10);
		\draw[opacity=0.3] (5,0) -- (5,10);

		\draw[-, line width=0.7pt] (0,8) -- (0.5,8) -- (0.5,7);
		\draw[black!30!white, line width=0.7pt] (0.5,7) -- (1,7);
		\draw[black!10!white, line width=0.7pt] (1,7) -- (1.5,8);
		\draw[-, line width=0.7pt] (0,6) -- (0,5) -- (0.5,5.5);
		\draw[black!30!white, line width=0.7pt] (0.5,5.5) -- (1,6);
		\draw[black!10!white, line width=0.7pt] (1,6) -- (1.5,6);
		\draw[-, line width=0.7pt] (0,3) -- (0.5,4);
		\draw[black!30!white, line width=0.7pt] (0.5,4) -- (1,4) -- (1,3);
		\draw[black!10!white, line width=0.7pt] (1,3) -- (1.5,2.5);
		
		\node[track node 1] () at (0,8) {};
		\node[track node 1] () at (0,6) {};
		\node[track node 1] () at (0,3) {};
		\node[track node 1] () at (0,5) {};
		\node[track node 1] () at (0.5,8) {};
		\node[track node 1] () at (0.5,7) {};
		\node[track node 1] () at (0.5,5.5) {};
		\node[track node 1] () at (0.5, 4) {};
		
		\node[track node 2] () at (1,7) {};
		\node[track node 2] () at (1,6) {};
		\node[track node 2] () at (1,4) {};
		\node[track node 2] () at (1,3) {};
		
		\node[track node 3] () at (1.5,8) {};
		\node[track node 3] () at (1.5,6) {};
		\node[track node 3] () at (1.5,2.5) {};

		\draw[black!30!white, line width=0.7pt] (2.5,8) -- (3,9) -- (3.5,9);
		\draw[black!30!white, line width=0.7pt] (2.5,6) -- (2.5,5) -- (3,6) -- (3.5,5.5);
		\draw[black!30!white, line width=0.7pt] (2.5,2.5) -- (3,3);
		\draw[-, line width=0.7pt] (3,3) -- (3,2);
		\draw[black!30!white, line width=0.7pt] (3,2) -- (3.5,2);
		\draw[-] (3,0) -- (3,10);
		
		\node[track node 1] () at (3,9) {};
		\node[track node 1] () at (3,6)  {};
		\node[track node 1] () at (3,3) {};
		\node[track node 1] () at (3,2) {};
		
		\node[track node 2] () at (2.5,8) {};
		\node[track node 2] () at (3.5,9) {};
		\node[track node 2] () at (2.5,6) {};
		\node[track node 2] () at (2.5,5) {};
		\node[track node 2] () at  (3.5,5.5) {};
		\node[track node 2] () at (2.5,2.5) {};
		\node[track node 2] () at  (3.5,2) {};

		\draw[-, line width=0.7pt] (5.5, 7) -- (6,8);
		\draw[black!30!white, line width=0.7pt] (5,8) -- (5,7) -- (5.5, 7);
		\draw[black!10!white, line width=0.7pt] (4.5,9) -- (5,8);
		\draw[-, line width=0.7pt]  (5.5,6) -- (5.5,5) -- (6,5);
		\draw[black!30!white, line width=0.7pt](5,6) -- (5.5,6);
		\draw[black!10!white, line width=0.7pt](4.5,5.5) -- (5,6);
		\draw[-, line width=0.7pt] (5.5,1) -- (6,1) -- (6,0.5);
		\draw[black!30!white, line width=0.7pt] (5,2) -- (5.5,1);
		\draw[black!10!white, line width=0.7pt] (4.5,2) -- (5,2);
		
		\node[track node 1] () at (6,8) {};
		\node[track node 1] () at (5.5, 7) {};
		\node[track node 1] () at (5.5,6) {};
		\node[track node 1] () at (5.5,5) {};
		\node[track node 1] () at (6,5) {};
		\node[track node 1] () at (5.5,1) {};
		\node[track node 1] () at (6,1) {};
		\node[track node 1] () at (6,0.5) {};
		
		\node[track node 2] () at (5,8) {};
		\node[track node 2] () at (5,7) {};
		\node[track node 2] () at (5,6) {};
		\node[track node 2] () at (5,2) {};
		
		\node[track node 3] () at (4.5,9) {};
		\node[track node 3] () at (4.5,5.5) {};
		\node[track node 3] () at (4.5,2) {};
		
		\end{scope}
\end{tikzpicture}}}
\caption{An example of $(\tilde{M}, \tilde{T})$, the result of applying \autoref{jklnlnlk} in the railed annulus ${\cal A}'$.}\label{snfcajksfn}
\end{figure}

Notice now that $\tilde{M}\cap C_{y}$ is the union of some of the paths in $\{P_{y,1},\ldots,P_{y,\lambda}\}$.
Suppose that these paths are $\{P_{y,c_{1}},\ldots,P_{y,c_{z}}\}$ where $\{c_{1},\ldots,c_{z}\}\subseteq [\lambda]$.
We  consider the $\lambda$-boundaried graph ${\bf M}_{y}=(M_y,B_{y,\lambda},ρ_{y,\lambda})$ where $M_y=(\tilde{M}\cap \overline{D}_{y})\cup (B_{y,\lambda}, \emptyset)$ (i.e. the graph $\tilde{M}\cap \overline{D}_{y}$ together with the isolated vertices $B_{y,\lambda}$).
We also define
$$\tilde{M}_{\rm out}=(\tilde{M}\setminus (\overline{D}_{y}\setminus B_{y,\lambda}))\cup (B_{y, \lambda}, \emptyset)\mbox{~and~} \tilde{\bf M}_{\rm out}=(\tilde{M}_{\rm out},B_{y,\lambda},ρ_{y,\lambda}).$$
Keep in mind that $\tilde{M}_{\rm out}$ is a subgraph of $G\setminus S$ that does not contain vertices of $K$ (see \autoref{dasfndajksfnjrk}).
	
\begin{figure}[ht]
\centering\scalebox{1}{
	\sshow{0}{\begin{tikzpicture}[scale=0.36]
		\begin{scope}[xshift=0cm]

		\draw[black!30!white, line width=0.7pt] (2.5,8) -- (3,9) -- (3.5,9);
		\draw[black!30!white, line width=0.7pt] (2.5,6) -- (2.5,5) -- (3,6) -- (3.5,5.5);
		\draw[black!30!white, line width=0.7pt] (2.5,2.5) -- (3,3);
		\draw[-, line width=0.7pt] (3,3) -- (3,2);
		\draw[black!30!white, line width=0.7pt] (3,2) -- (3.5,2);
		\draw[-] (3,1) -- (3,10);
		\node () at (3,10.5) {$C_{y}$}; 
		\node (r1) at (4,6.7) {$r_{y,c_{1}}$};
		\node (r2) at (4,3.5) {$r_{y,c_{2}}$};
		\node[blue node] () at (3,9) {};
		\node[track node 1] () at (3,6)  {};
		\node[track node 1] () at (3,3) {};
		\node[track node 1] () at (3,2) {};
		
		\node[track node 2] () at (2.5,8) {};
		\node[track node 2] () at (3.5,9) {};
		\node[track node 2] () at (2.5,6) {};
		\node[track node 2] () at (2.5,5) {};
		\node[track node 2] () at  (3.5,5.5) {};
		\node[track node 2] () at (2.5,2.5) {};
		\node[track node 2] () at  (3.5,2) {};

		\begin{scope}[xshift=-5cm]
		
		\node[blue node] (M1) at (3.2,8) {};
		\node[blue node] (M2) at (0.7,5) {};
		\node[blue node] (M3) at (4.3,4) {};
		
		\node[blue node] (m1) at (8,6) {};
		\node[blue node] (m2) at (8,3) {};

		\draw[blue] (M1) -- (M2) (M2) -- (M3) (M1) -- (M3);
	
		\draw[blue] (M3) -- (m1);
		\draw[blue] (M3) -- (m2);

		\end{scope}
		\end{scope}
		
		\begin{scope}[xshift=6cm]

		\draw[black!30!white, line width=0.7pt] (2.5,8) -- (3,9) -- (3.5,9);
		\draw[black!30!white, line width=0.7pt] (2.5,6) -- (2.5,5) -- (3,6) -- (3.5,5.5);
		\draw[black!30!white, line width=0.7pt] (2.5,2.5) -- (3,3);
		\draw[-, line width=0.7pt] (3,3) -- (3,2);
		\draw[black!30!white, line width=0.7pt] (3,2) -- (3.5,2);
		\draw[-] (3,1) -- (3,10);
		\node () at (3,10.5) {$C_{y}$}; 
		\node (r1) at (4,6.7) {$r_{y,c_{1}}$};
		\node (r2) at (4,3.5) {$r_{y,c_{2}}$};
		\node[red node] () at (3,9) {};
		\node[track node 1] () at (3,6)  {};
		\node[track node 1] () at (3,3) {};
		\node[track node 1] () at (3,2) {};
		
		\node[track node 2] () at (2.5,8) {};
		\node[track node 2] () at (3.5,9) {};
		\node[track node 2] () at (2.5,6) {};
		\node[track node 2] () at (2.5,5) {};
		\node[track node 2] () at  (3.5,5.5) {};
		\node[track node 2] () at (2.5,2.5) {};
		\node[track node 2] () at  (3.5,2) {};

		\end{scope}
		\begin{scope}[xshift=1cm]

		\node[red node] (M4) at (13,7) {};
		\node[red node] (M5) at (17,7.5) {};
		\node[red node] (M6) at (11,4) {};
		\node[red node] (M7) at (16,2.5) {};
		\node[red node] (m1) at (8,6) {};
		\node[red node] (m2) at (8,3) {};
		\node[red node] (m3) at (8,2) {};
		
		\draw[red] (M4) -- (M5) (M6) -- (M7) (M4) -- (M6) (M5)-- (M7) (M5) -- (M6);
		\draw[red] (m1) --  (M4);
		\draw[red] (m2) -- (m3) -- (M6);
		
		\end{scope}

		\end{tikzpicture}}}
	\caption{The graphs $\tilde{M}_{\rm out}$ (depicted in blue) and $M_{y}$ (depicted in red).}
	\label{dasfndajksfnjrk}
\end{figure}

Notice that ${\bf M}_y$ is a subgraph of ${\bf G}_{y,\lambda}$ and therefore the graph $\tilde{\bf M}_{\rm out}\oplus {\bf M}_y$
(that is the graph $\tilde{M}\cup (B_{y,\lambda},\emptyset)$)
is a subgraph of $\tilde{\bf M}_{\rm out}\oplus {\bf G}_{y,\lambda}$.
Thus, $H\preceq \tilde{\bf M}_{\rm out}\oplus {\bf G}_{y,\lambda}$.
Since ${\bf J}_y = {\bf J}_{i'}$, we have that $H\preceq \tilde{\bf M}_{\rm out}\oplus {\bf G}_{i',\lambda}$.
Therefore, ${\bf G}_{i',\lambda}$ contains as a subgraph a $\lambda$-boundaried graph ${\bf M}_{i'}=(M_{i'},B_{i',\lambda},ρ_{i',\lambda})$,  
such that $\tilde{\bf M}_{\rm out}\oplus {\bf M}_{i'}$ {contains $H$ as a topological minor}.
Notice that $M_{i'}$ is a subgraph of $G\setminus S$ that does not intersect $V({\sf compass}(W))$.
\begin{figure}[ht]
		\centering\scalebox{1}{
		\sshow{0}{\begin{tikzpicture}[scale=0.4]

		\node () at (12,-1) {${\cal A}'$};
		
		\foreach \i in {10,16, 23 }{
			\begin{scope}[xshift=\i cm]
			\draw[opacity=0.1] (-1,0) -- (-1,10);
			\draw[opacity=0.3] (-.5,0) -- (-.5,10);
			\draw[-] (0,0) -- (0, 10);
			\draw[opacity=0.3] (0.5,0) -- (0.5,10);
			\draw[opacity=0.1] (1,0) -- (1,10);
			\end{scope}
		}
		
		\node () at (10,12) {$C_{y-\lfloor \mu/2\rfloor}$};
		\draw[red] (10,11.1) -- (10,10.5);
		\node[small red] () at (10,10.5) {};
		\node () at (16,12) {$C_{y+\lfloor \mu/2\rfloor}$};
		\draw[red] (16,11.1) -- (16,10.5);
		\node[small red] () at (16,10.5) {};
		\node[] () at (23,12) {$C_{i'}$};
		\draw[red] (23,11.3) -- (23,10.5);
		\node[small red] () at (23,10.5) {};
		\node () at (13,11) {$C_{y}$};

		\filldraw[draw=applegreen,fill=applegreen!60!white] (20.5,9) --  (22.5,9)  -- (22.5,7) -- (20.5,7) -- cycle;
		\node () at (21.2,8.3) {$\Delta'$};

		\begin{scope}[xshift=5cm]
		
		\node[model node small] (M1) at (3.2,8) {};
		\node[model node small] (M2) at (0.7,5) {};
		\node[model node small] (M3) at (4.3,4) {};
		
		\node[rep node] (m1) at (8,6) {};
		\node[rep node] (m2) at (8,3) {};
		\node[rep node] (m3) at (8,2) {};
		
		\draw[celestialblue] (M1) -- (M2) (M2) -- (M3) (M1) -- (M3);
		
		\draw[celestialblue] (M3) -- (m1);
		\draw[celestialblue] (M3) -- (m2);
		
		\end{scope}
		
		\begin{scope}[xshift=12cm]
		
		\node[model node small] (M4) at (15,7) {};
		\node[model node small] (M5) at (18,7.5) {};
		\node[model node small] (M6) at (13,4) {};
		\node[model node small] (M7) at (17,3) {};
		\node[rep node] (mi1) at (11,6) {};
		\node[rep node] (mi3) at (11,2) {};
		\node () at (16,8.5) {$\hat{M}_{i'}$};
		
		\draw[celestialblue] (M4) -- (M5) (M6) -- (M7) (M4) -- (M6) (M5)-- (M7) (M5) -- (M6);
		\draw[celestialblue] (mi1) --  (M4);
		\draw[celestialblue] (mi3) -- (M6);

		\node[rep node] (a3) at (10.5, 5) {};
		\node[rep node] (b3) at (10.5,1) {};
		\node[rep node] (b2) at (10.5,2) {};
		\node[rep node] (b4) at (11,1) {};
		\node[rep node] (a2) at (10,4) {};
		\node[rep node] (a1) at (10,5) {};
		\node[rep node] (b1) at (10,2) {};
		
		\end{scope}

		\begin{scope}[xshift=10cm]
		
		\draw[red] (m1) -- (3.5,5.5) --(4.5,5.5) -- (5,6) --  (5.5,6) -- (5.5,5) -- (6,5) -- (a1)  -- (a2)  -- (a3) -- (mi1);
		\draw[red] (m2) -- (m3) --  (4.5,2) -- (5,2) -- (5.5,1) -- (6,1) -- (6,0.5) -- (b1) -- (b2) -- (b3) -- (b4) -- (mi3);
		
		\node () at (8.5,6) {$P_{1}^{*}$};
		\node () at (8.5,2) {$P_{2}^{*}$};

		\node[rep node] () at  (3.5,5.5) {};
		\node[rep node] () at  (3.5,2) {};

		\node[rep node] () at (5.5,6) {};
		\node[rep node] () at (5.5,5) {};
		\node[rep node] () at (6,5) {};
		\node[rep node] () at (5.5,1) {};
		\node[rep node] () at (6,1) {};
		\node[rep node] () at (6,0.5) {};
		
		\node[rep node] () at (5,6) {};
		\node[rep node] () at (5,2) {};
		
		\node[rep node] () at (4.5,5.5) {};
		\node[rep node] () at (4.5,2) {};
		
		\end{scope}

		\begin{scope}[on background layer,xshift=10cm]
		
		\draw[opacity=0.3] (1,0) -- (1,10);
		\draw[opacity=0.1] (1.5,0) -- (1.5,10);
		
		\draw[opacity=0.2] (2.5,0) -- (2.5,10);
		\draw[opacity=0.2] (3.5,0) -- (3.5,10);
		
		\draw[opacity=0.1] (4.5,0) -- (4.5,10);
		\draw[opacity=0.3] (5,0) -- (5,10);

		\draw[-, line width=0.7pt] (0,8) -- (0.5,8) -- (0.5,7);
		\draw[black!30!white, line width=0.7pt] (0.5,7) -- (1,7);
		\draw[black!10!white, line width=0.7pt] (1,7) -- (1.5,8);
		\draw[-, line width=0.7pt] (0,6) -- (0,5) -- (0.5,5.5);
		\draw[black!30!white, line width=0.7pt] (0.5,5.5) -- (1,6);
		\draw[black!10!white, line width=0.7pt] (1,6) -- (1.5,6);
		\draw[-, line width=0.7pt] (0,3) -- (0.5,4);
		\draw[black!30!white, line width=0.7pt] (0.5,4) -- (1,4) -- (1,3);
		\draw[black!10!white, line width=0.7pt] (1,3) -- (1.5,2.5);
		
		\node[track node 1] () at (0,8) {};
		\node[track node 1] () at (0,6) {};
		\node[track node 1] () at (0,3) {};
		\node[track node 1] () at (0,5) {};
		\node[track node 1] () at (0.5,8) {};
		\node[track node 1] () at (0.5,7) {};
		\node[track node 1] () at (0.5,5.5) {};
		\node[track node 1] () at (0.5, 4) {};
		
		\node[track node 2] () at (1,7) {};
		\node[track node 2] () at (1,6) {};
		\node[track node 2] () at (1,4) {};
		\node[track node 2] () at (1,3) {};
		
		\node[track node 3] () at (1.5,8) {};
		\node[track node 3] () at (1.5,6) {};
		\node[track node 3] () at (1.5,2.5) {};

		\draw[black!30!white, line width=0.7pt] (2.5,8) -- (3,9) -- (3.5,9);
		\draw[black!30!white, line width=0.7pt] (2.5,6) -- (2.5,5) -- (3,6) -- (3.5,5.5);
		\draw[black!30!white, line width=0.7pt] (2.5,2.5) -- (3,3);
		\draw[-, line width=0.7pt] (3,3) -- (3,2);
		\draw[black!30!white, line width=0.7pt] (3,2) -- (3.5,2);
		\draw[-] (3,0) -- (3,10);
		
		\node[track node 1] () at (3,9) {};
		\node[track node 1] () at (3,6)  {};
		\node[track node 1] () at (3,3) {};
		\node[track node 1] () at (3,2) {};
		
		\node[track node 2] () at (2.5,8) {};
		\node[track node 2] () at (3.5,9) {};
		\node[track node 2] () at (2.5,6) {};
		\node[track node 2] () at (2.5,5) {};
		\node[track node 2] () at  (3.5,5.5) {};
		\node[track node 2] () at (2.5,2.5) {};
		\node[track node 2] () at  (3.5,2) {};

		\draw[-, line width=0.7pt] (5.5, 7) -- (6,8);
		\draw[black!30!white, line width=0.7pt] (5,8) -- (5,7) -- (5.5, 7);
		\draw[black!10!white, line width=0.7pt] (4.5,9) -- (5,8);
		\draw[-, line width=0.7pt]  (5.5,6) -- (5.5,5) -- (6,5);
		\draw[black!30!white, line width=0.7pt](5,6) -- (5.5,6);
		\draw[black!10!white, line width=0.7pt](4.5,5.5) -- (5,6);
		\draw[-, line width=0.7pt] (5.5,1) -- (6,1) -- (6,0.5);
		\draw[black!30!white, line width=0.7pt] (5,2) -- (5.5,1);
		\draw[black!10!white, line width=0.7pt] (4.5,2) -- (5,2);
		
		\node[track node 1] () at (6,8) {};
		\node[track node 1] () at (5.5, 7) {};
		\node[track node 1] () at (5.5,6) {};
		\node[track node 1] () at (5.5,5) {};
		\node[track node 1] () at (6,5) {};
		\node[track node 1] () at (5.5,1) {};
		\node[track node 1] () at (6,1) {};
		\node[track node 1] () at (6,0.5) {};
		
		\node[track node 2] () at (5,8) {};
		\node[track node 2] () at (5,7) {};
		\node[track node 2] () at (5,6) {};
		\node[track node 2] () at (5,2) {};
		
		\node[track node 3] () at (4.5,9) {};
		\node[track node 3] () at (4.5,5.5) {};
		\node[track node 3] () at (4.5,2) {};
		
		\end{scope}
	\end{tikzpicture}}}
	\caption{Visualization of the last part of the proof.}
	\label{sdgdfsgfaggdfgsdsfngdsgfn}
\end{figure}

For every $j\in[z]$, we define $P^*_{j}$ to be the path in $P_j$
starting from $r_{y,j}$ and finishing to $r_{i',y}$,
i.e.,  $P^*_{j}= (P_{j} \cap \overline{D}_y )\setminus(\overline{D}_{i'}\setminus r_{i',j})$ and
${\cal P}^*=\{P^*_{j}\mid j\in[z]\}$.
Observe that none of the paths in ${\cal P}^*$ intersects $V(K)$.
Let $\hat{M}_{i'}$ (resp. $\hat{M}_{\rm out}$) be the graph
obtained from $M_{i'}$ (resp. $\tilde{M}_{\rm out}$) after removing,
for every $j\in[\lambda]\setminus \{c_{1}, \ldots, c_{z}\}$,
the vertices $r_{i',j}$ (resp. $r_{y,j}$).
Consider now the graph
$M_{0}:=\hat{M}_{\rm out}\cup \hat{M}_{i'}\cup\cupall {\cal P}^{*}$
and observe that
$M_{0}$ is a subgraph of $(G\setminus V(K))\setminus S$
that is a subdivision of $H$.
Therefore $H\preceq(G\setminus V(K))\setminus S$,
a contradiction (see \autoref{sdgdfsgfaggdfgsdsfngdsgfn}).
\end{proof}

\section{Proof of \autoref{dfasdfdsad}}
\label{sasdfsdfdsfsdfsdgffdsgdfgsdfgd}

In this section, having all necessary results, we are in position to present the proof of \autoref{dfasdfdsad}.

\begin{proof}[Proof of \autoref{dfasdfdsad}]
Let
	\begin{align*}
	 x  := & \max\{\funref{dsfndnvgdsqwert}(h,k),  \funref{dfljhklgdfjhklgfj}(h)\}=\mathcal{O}_{h} (k),\\ 
	 y  := & \max\{\funref{fnkdslnfgkldsang}(h,3),\funref{dsgkdgnlkdfsngkdfsnl}(h,3)\},\\
	 z :=&  \funref{dsalgmdlaotpopyrt}(h,k)+1=\mathcal{O}_{h} (k^2), ~\mbox{and}\\
	q:= & \funref{ddansjbndaj}(x,y,z)= \mathcal{O}(x+y\sqrt{z})= \mathcal{O}_{h} (k).
	\end{align*}
	
We first call the algorithm {\bf Find\_Wall}$(G,q)$ of \autoref{something_good} which outputs either a $q$-wall $W$ of $G$ whose compass has treewidth at most $\conref{sdafsdfsd}\cdot q$ or a tree decomposition of $G$ of width at most $\conref{sdafsdfsd}\cdot q$. This algorithm runs in $2^{\mathcal{O} (q^2)}\cdot n=2^{\mathcal{O}_h (k^2)}\cdot n$ time, or, alternatively, in $\mathcal{O}(n^{2})$ time.
We consider the first case.
		
Let $\Delta$ be the closed disk whose boundary is the perimeter of $W$ and contains ${\sf compass}(W)$.
We call the algorithm  {\bf Find\_Collection\_of\_Annuli}$(x,y,z,\Delta,G,W)$ of \autoref{gdgdasgfhjkuyi} which, in $\mathcal{O}(n)$ time, outputs a closed disk $\Delta'\subseteq \Delta$ and a collection $\mathfrak{A}=\{{\cal A}_{0},{\cal A}_{1},\ldots,{\cal A}_{z}\}$ of railed annuli of the compass of $G$ such that 
	\begin{itemize}
		\item ${\cal A}_{0}$ is an $(x,x)$-railed annulus  whose outer disk is  $\Delta$ and whose inner disk is $\Delta'$,   
		\item for $i\in[z]$, ${\cal A}_{i}$ is a $(y,y)$-railed annulus of $G\cap \inter(\Delta')$, and  
		\item for every $i,j\in [z]$ where $i\neq j$, the outer disk of ${\cal A}_{i}$ and the outer disk of ${\cal A}_{j}$ are disjoint.
	\end{itemize}
Then, we call the algorithm {\bf Reduce\_Solution\_Space}$(k,h,w,\Delta,G,R,{\cal C},{\cal P})$  of \autoref{fsfsdfdsdsafsasdsdadffasdasfd} for $({\cal C}, {\cal P}):={\cal A}_{0}$ and $w:=\conref{sdafsdfsd}\cdot q$  which outputs a set $R'\subseteq R$ such that
\begin{itemize}
\item $|R' \cap \inter(\Delta')|\leq\funref{dsalgmdlaotpopyrt}(h,k) = z-1$ and 
\item if $(G,R,k)$ is a ${\bf tm}_{\cal F}$-triple then $(G,R',k)$ is a ${\bf tm}_{\cal F}$-triple.
\end{itemize}
	This algorithm runs in  $2^{\mathcal{O}_{h} (q \log q)}\cdot k\cdot n=2^{\mathcal{O}_{h} (k \log k)}\cdot n$ time,  or, alternatively, in $\mathcal{O}(k\cdot n^3)+2^{\mathcal{O}_{h}(q)}\cdot k\cdot n=\mathcal{O}(k\cdot n^3)+2^{\mathcal{O}_{h}(k)}\cdot n$ time.
Since $|R' |<z$ then there exists a $j\in[z]$ such that $R'\cap \ann({{\cal A}_{j}})=\emptyset$.
Let $({\cal C}^{(j)}, {\cal P}^{(j)}):={\cal A}_{j}$ and $\Delta_j$ be the closure of the outer disk of ${\cal A}_j$.
Now, for $b:=3$, the algorithm  {\bf Find\_irrelevant\_area}$(h,b,w,{\Delta_j},G,R',{\cal C}^{(j)}, {\cal P}^{(j)})$ of  \autoref{fsfsdfdsdsafsadffasdasfd1} outputs a $b$-wall $W$ of $G$ such that 
\begin{itemize}

\item $V({\sf compass}(W))\subseteq \Delta$ and
 if and only if
\item if $(G\setminus V({\sf compass}(W)),R,k)$ is a ${\bf tm}_{\cal F}$-triple then $(G, R, k)$ is a ${\bf tm}_{\cal F}$-triple.
\end{itemize} 
This algorithm runs in $2^{\mathcal{O}_{h} (q\log q)}\cdot n=2^{\mathcal{O}_{h} (k\log k)}\cdot n$ time, or, alternatively, in $ \mathcal{O}_h (n^3)+ 2^{\mathcal{O}_{h}(q)}\cdot n= \mathcal{O}_h (n^3)+ 2^{\mathcal{O}_{h}(k)}\cdot n$ time.

Therefore, if we pick a vertex $v\in V(G)\cap \Delta''$ then it holds that 
$(G,R,k)$ is a ${\bf tm}_{\cal F}$-triple if and only if $(G\setminus v, R', k)$ is a ${\bf tm}_{\cal F}$-triple.
The overall running time of the whole procedure is $2^{\mathcal{O}_h (k^2)}\cdot n$, or, alternatively, $\mathcal{O}(k\cdot n^3)+\mathcal{O}_{h}(n^3)+2^{\mathcal{O}_{h}(k)}\cdot n$.
\end{proof}

\section{Discussion}
\label{dnfjklsdgnklsdnbm}

In this paper we prove that  \textsc{${\cal F}$-TM-Deletion} is Fixed Parameter Tractable on planar graphs.

\subsection{Running time dependency on \texorpdfstring{$h$}{h}} The parametric dependency of our {\sf FPT}-algorithm 
is $2^{\mathcal{O}_h (k^2)}$ and it can be dropped to $2^{\mathcal{O}_h (k)}$
if we admit a cubic polynomial dependency on $n$. 
However both these parametric dependencies hide huge dependency on $h$. To estimate this,
one may observe that the complexity of the dynamic programming algorithm of \autoref{sssswfgfegwergewgwergwegr} dominates the overall running time of the 
algorithm of \autoref{dfasdfdsad}, in terms of the contribution of $h$. This permits us to
estimate that the  algorithm of \autoref{dfassdfasfdfdsad} 
runs in  $2^{k^2\cdot 2^{2^{2^{2^{\mathcal{O}(h)}}}}}\cdot n^2$ time, or, alternatively, in
$\mathcal{O}(k\cdot n^4)+2^{2^{2^{2^{2^{\mathcal{O}(h)}}}}}\cdot n^4+2^{k\cdot 2^{2^{2^{2^{\mathcal{O}(h)}}}}}\cdot n^2$.

\subsection{Extensions to bounded genus graphs}
\label{asdfsdgsgdfgsdfhgdgh}
In this subsection, we show how to extend our results to graphs of Euler genus at most $\gamma$.
In particular, we obtain an algorithm 
for  \textsc{${\cal F}$-TM-Deletion} on graphs of Euler genus at most $γ$ that runs in $2^{\mathcal{O}_{h,γ} (k^2)}\cdot n^2$ time, or, alternatively, in $\mathcal{O}_{γ}(k\cdot n^4)+\mathcal{O}_{h,γ}(n^4)+2^{\mathcal{O}_{h,γ}(k)}\cdot n^{\mathcal{O}(1)}$ time.

\begin{theorem} \label{thrm_eulergen}
There exists an algorithm that given a finite set of graphs ${\cal F}$, a $k\in\Bbb{N}$, and an $n$-vertex graph $G$ of Euler genus at most $γ$,
outputs whether ${\bf tm}_{\cal F}(G)\leq k$ in  $2^{\mathcal{O}_{h,γ} (k^2)}\cdot n^2$ time, or, alternatively, $\mathcal{O}_γ (k\cdot n^4)+\mathcal{O}_{h,γ}(n^4)+2^{\mathcal{O}_{h,γ}(k)}\cdot n^{\mathcal{O}(1)}$ time, where $h=h({\cal F})$. 
\end{theorem}

To prove~\autoref{thrm_eulergen}, we can follow the same approach
as in the proof of \autoref{dfassdfasfdfdsad}, i.e., reduce the problem
to instances of bounded treewidth by removing vertices and reducing the set $R$.
This is done using the following result, which is an analogue of~\autoref{dfasdfdsad} for graphs of bounded Euler genus.

\begin{lemma}\label{irr_boundedeulergenus}
There exists a function $\newfun{fsdfdsytr}:\Bbb{N}\to\Bbb{N}$, and an algorithm that
given two integers $k,h\in \Bbb{N}$, an $n$-vertex graph $G$ of Euler genus at most $γ$,
and a set $R\subseteq V(G)$, outputs
either
a vertex $v\in V(G)$ and a set $R'\subseteq R$ such that,
for every graph class ${\cal F}$ where $h({\cal F})\leq h$,
$(G,R,k)$ is a  ${\bf tm}_{\cal F}$-triple if and only if
$(G\setminus v, R', k)$  is a ${\bf tm}_{\cal F}$-triple or 
a tree decomposition of $G$ of width
at most $\funref{fsdfdsytr}(h)\cdot k$.
Moreover, this algorithm runs in  $2^{\mathcal{O}_{h,\gamma} (k^2)}\cdot n$ time, or, alternatively, $\mathcal{O}_\gamma(k\cdot n^3)+\mathcal{O}_{h,\gamma}(n^3)+2^{\mathcal{O}_{h,\gamma}(k)}\cdot n$ time.
\end{lemma}

The proof of~\autoref{irr_boundedeulergenus} is analogous to the proof of~\autoref{dfasdfdsad}.
We can use both subroutines in \autoref{asfdsasdfdsf}
(i.e., the algorithms of \autoref{fsfsdfdsdsafsasdsdadffasdasfd} and \autoref{fsfsdfdsdsafsadffasdasfd1})
since they
are designed to work when the input graph is partially $\Delta$-embedded (and not necessarily planar).
The only missing ingredient for the proof of~\autoref{irr_boundedeulergenus}
is an extension of \autoref{something_good} on graphs of bounded Euler genus.

\begin{lemma}\label{find_wall_eulergen}
There exists a constant $\newcon{dsdfsffds}$ and an algorithm that given
an $n$-vertex graph $G$ of Euler genus at most $\gamma$ and an integer $q\in \Bbb{N}_{\geq 3}$,
outputs either a disk-embedded $q$-wall $W$ of $G$ whose
compass has treewidth at most $\conref{dsdfsffds}\cdot q$ or  
a tree decomposition of $G$ of width at most $\conref{dsdfsffds}\cdot q$.
Moreover, this algorithm runs in $2^{\mathcal{O}_\gamma(q^2)}\cdot n$ time,
or, alternatively, in $2^{\mathcal{O}_\gamma(q)}\cdot n^2$ time.	
\end{lemma}	

\begin{proof}
As a first step we use the single exponential $5$-approximation algorithm of
\cite{BodlaenderDDFLP16} in order to check whether $\tw(G)=\mathcal{O}(q)$.
If not, we aim to find a disk-embedded 
$q$-wall, whose existence is guaranteed by the grid exclusion theorem on bounded genus graphs (see e.g., \cite{DemaineFHT05sube,DemaineHT04theb,FominGT11cont,DemaineH08line}).
To find the $q$-wall, we may again use the algorithm of \cite{AdlerDFST11} to first detect a wall of $G$ and then a subwall of it that is disk-embedded (see \cite{GeelenRS04embe}).
This way, we derive an algorithm running in $2^{\mathcal{O}_{γ}(q^2)}\cdot n$ time.
Ιf we want to avoid the exponential dependence on $q^2$, we may find the $q$-wall 
by the following alternative approach: 
as a first step, we may find a set $S$ of $\mathcal{O}_{γ}(q)$  vertices whose 
removal from $G$ will give  either a planar graph of treewidth $Ω(q)$ or a non-planar embedded graph of face-width  $Ω(q)$. This can be done by successively 
finding minimum-size non-contractible cycles on $\mathcal{O}(q)$ vertices by using a polynomial 
time algorithm (see e.g., \cite{Fox13shor,CabelloCE13mult,Thomassen90embe,EricksonH04opti,CabelloVL16find}). If the outcome is that $S$ is a set of vertices whose removal from $G$ gives a planar graph,
then, as in the planar case, we use the polynomial algorithm of \cite{GuT12} to find the $q$-wall. Otherwise, the $q$-wall can be found by using the polynomial algorithmic procedure described in the proof of \cite[Lemma 3.3]{DemaineFHT05sube}. The overall running time of the above procedure is $2^{\mathcal{O}_{γ}(q)}\cdot n^{\mathcal{O}(1)}$.
\end{proof}

To complete the proof of~\autoref{thrm_eulergen} and achieve the claimed parametric dependencies in its running times, we may adapt the dynamic programming algorithm of \autoref{sssswfgfegwergewgwergwegr}
which 
runs in {$2^{\mathcal{O}_{h,γ}(w\log w)}\cdot n$ time}.
Again, if we want to avoid the $\log w$ contribution in the exponent, we may alternatively use  
dynamic programming 
in \cite[Theorem 10.1]{BasteST20hittI} (based on an extension of sphere-cut decompositions called {\sl surface-cut decompositions}) and derive a dynamic programming algorithm that runs in 
 $\mathcal{O}_{γ}(n^3)+2^{\mathcal{O}_{h,γ}(w)}\cdot n$ time.
\medskip

We stress that, in the above analysis, we insisted on a single-exponential dependence on $k$ on the running time of the algorithms. The reason is that this implies that, for every  finite set of graphs ${\cal F}$ with detail $h$ and every $γ\in\Bbb{N}$, the following problem is polynomially solvable.\medskip

\begin{center}
\fbox{
\begin{minipage}{12cm}
\noindent \textsc{$({\cal F},γ)$-Log-TM-Deletion}\\
\noindent {\bf Input}: an $n$-vertex graph $G$ of Euler genus $γ$.\\
{\bf Question}:  Does $G$ contain a set  $S$ of $\log n$ vertices such that $G\setminus S$\\
\phantom{{\bf Question}:}   excludes every graph in ${\cal F}$ as a topological minor?
\end{minipage}
}
\end{center}

\subsection{Recent  advances on the general problem}
\label{asdfsfgddfhdgshfghg}

The remaining question is whether 
the same result can be derived for {\sl all} graphs. Recently an $\mathcal{O}_{h,k}(n^4)$ algorithm for the general   \textsc{${\cal F}$-TM-Deletion} problem was proposed by Fomin et al.~\cite{FominLP0Z20hitt}. Using the words of  \cite{FominLP0Z20hitt}, the parametric 
dependency of this algorithm, on $k$ and $h$, is humongous. However, in the same paper, it was proven
that better parametric dependancies can be achieved when restricting the problem to graphs of bounded Euler genus. According to the results of ~\cite{FominLP0Z20hitt}, \textsc{${\cal F}$-TM-Deletion} on graphs of Euler genus at most $γ$ can be solved by an algorithm running in $2^{2^{\mathcal{O}_{h,γ}(k)}}\cdot n^2$ time. The algorithms claimed in \autoref{asdfsdgsgdfgsdfhgdgh}
can be seen as improvements of this result.

\subsection{Open problems}
\label{asdgasdfgdsfgdfg}

We believe that the techniques developed in this paper can be applied to other instances 
of the ${\cal P}$-{\sc deletion} problem. In particular, we conjecture the following.

\begin{conjecture}
\label{asdfsdff}
If ${\cal F}$ is a finite set of graphs and  $\leq$ is the contraction relation, then
the problem {\sc ${\cal P}_{{\cal F},\leq }$-deletion}, with inputs restricted on  graphs of Euler genus at most $γ$, can be solved 
by an algorithm that runs in $\mathcal{O}_{h,k}(n^c)$ time, for some constant $c$.
\end{conjecture}

\begin{conjecture}
\label{asdfsdffs}
If ${\cal F}$ is a finite set of graphs and  $\leq$ is the induced minor relation, then
the problem {\sc ${\cal P}_{{\cal F},\leq }$-deletion},  with inputs restricted on  graphs of Euler genus at most $γ$, can be solved 
by an algorithm that runs in $\mathcal{O}_{h,k}(n^c)$ time, for some constant $c$.
\end{conjecture}

A possible pathway for proving \autoref{asdfsdff} is to use the fact that, given an graph embedding $Γ$,
edge contractions on $G$ correspond to topological minors on the dual embedding $Γ^*$.
This ``translation'' of contractions to topological minors  was proposed in~\cite{KaminskiT12} in order to devise an algorithm for the problem of checking whether 
a graph of Euler genus at most $γ$ contains a graph $H$ as a contraction (this result is 
the case $k=0$ of \autoref{asdfsdff}). Under this setting, the only significant 
change is that instead of looking for a set of vertices to remove, we must find a set of faces to ``shrink''.
Therefore, the main missing ingredient for a proof of \autoref{asdfsdff}  is a dynamic programming framework for this shrinking variant 
of the problem on surfaces.\medskip

For \autoref{asdfsdffs} one may directly attempt to build counterparts of all the algorithms of this paper 
for the induced minor relation. The only missing combinatorial ingredient for this is an ``induced'' version of 
the model combing theorem (\autoref{jklnlnlk}), that is proved in \cite{GolovachST22comb}.

\section*{Acknowledgements} We wish to thank the anonymous reviewers for their comments and remarks that improved the presentation of this paper. Moreover, we are especially indebted to Fedor V. Fomin for his valuable comments and advise.
%

\end{document}